\documentclass[cmp]{svjour}
\pdfoutput=1

\usepackage{amsmath}
\usepackage{mathtools}
\usepackage{lmodern}
\usepackage[T1]{fontenc}
\usepackage{amsfonts}
\usepackage{bbm}
\DeclareMathAlphabet{\mathbfi}{OML}{cmm}{b}{it}

\let\originalleft\left
\let\originalright\right
\renewcommand{\left}{\mathopen{}\mathclose\bgroup\originalleft}
\renewcommand{\right}{\aftergroup\egroup\originalright}
\makeatletter
\newcommand{\vast}{\bBigg@{4}}
\newcommand{\Vast}{\bBigg@{5}}
\makeatother

\makeatletter
\newenvironment{equations}[1][]{\subequations\ifx\relax#1\relax\else\label{#1}\fi\align\ignorespaces}{\endalign\ignorespacesafterend\endsubequations}
\def\@spliteq#1{\begin{equation}\begin{split}#1\end{split}\end{equation}}
\def\@spliteqstar#1{\begin{equation*}\begin{split}#1\end{split}\end{equation*}}
\def\splitequation{\collect@body\@spliteq}
\expandafter\def\csname splitequation*\endcsname{\collect@body\@spliteqstar}

\expandafter\def\csname endsplitequation*\endcsname{\ignorespacesafterend}
\makeatother

\makeatletter
\def\normalthmheadings{\def\@spbegintheorem##1##2##3##4{\trivlist
                 \item[\hskip\labelsep{##3##1\ ##2\@thmcounterend}]##4\addcontentsline{toc}{subsection}{##1\ ##2\@thmcounterend}}
\def\@spopargbegintheorem##1##2##3##4##5{\trivlist
      \item[\hskip\labelsep{##4##1\ ##2}]{##4(##3)\@thmcounterend\ }\addcontentsline{toc}{subsection}{##1\ ##2\@thmcounterend}##5}}
\normalthmheadings
\def\reversethmheadings{\def\@spbegintheorem##1##2##3##4{\trivlist
                 \item[\hskip\labelsep{##3##2\ ##1\@thmcounterend}]##4\addcontentsline{toc}{subsection}{##2\ ##1\@thmcounterend}}
\def\@spopargbegintheorem##1##2##3##4##5{\trivlist
      \item[\hskip\labelsep{##4##2\ ##1}]{##4(##3)\@thmcounterend\ }\addcontentsline{toc}{subsection}{##2\ ##1\@thmcounterend}##5}}
\makeatother
\spnewtheorem*{remark*}{Remark}{\itshape}{\rmfamily}

\spnewtheorem*{example*}{Example}{\itshape}{\rmfamily}

\let\oldendproof\endproof
\def\endproof{\hfill\squareforqed\oldendproof}

\def\blank{{\circ}}

\let\oldre\Re
\let\oldim\Im
\renewcommand{\Re}{\oldre\mathfrak{e}\,}
\renewcommand{\Im}{\oldim\mathfrak{m}\,}
\newcommand{\mathe}{\mathrm{e}}
\newcommand{\mathi}{\mathrm{i}}
\newcommand{\total}{\mathop{}\!\mathrm{d}}
\newcommand{\laplace}{\mathop{}\!\bigtriangleup}
\newcommand{\abs}[1]{{\left\lvert{#1}\right\rvert}}
\newcommand{\norm}[1]{{\left\lVert{#1}\right\rVert}}
\newcommand{\1}{\mathbbm{1}}
\newcommand{\eqend}[1]{\,#1}
\newcommand{\bigo}[1]{\mathcal{O}\left({#1}\right)}
\newcommand{\op}{\mathcal{O}}
\newcommand{\istest}{\in \mathcal{S}(\mathbb{R}^2)}
\newcommand{\expect}[1]{\left\langle{#1}\right\rangle}
\newcommand{\bessel}[3]{\mathop{}\!\mathrm{#1}_{#2}\left(#3\right)}

\DeclareMathOperator{\sgn}{sgn}

\usepackage{cite}

\usepackage[unicode=true,bookmarksopen=true,colorlinks=true,allcolors=blue,linktocpage=true,pdfa=true]{hyperref}
\usepackage{bookmark}

\allowdisplaybreaks

\usepackage{calc}
\journalname{Annales Henri Poincaré}

\begin{document}

\title{Local operators in the Sine-Gordon model: $\partial_\mu \phi \, \partial_\nu \phi$ and the stress tensor}
\titlerunning{Local operators in the Sine-Gordon model}

\author{Markus B. Fr{\"o}b \and Daniela Cadamuro}
\institute{Institut f{\"u}r Theoretische Physik, Universit{\"a}t Leipzig,\\ Br{\"u}derstra{\ss}e 16, 04103 Leipzig, Germany\\ \email{\href{mailto:mfroeb@itp.uni-leipzig.de}{mfroeb@itp.uni-leipzig.de}, \href{mailto:cadamuro@itp.uni-leipzig.de}{cadamuro@itp.uni-leipzig.de}}}
\authorrunning{M. B. Fr{\"o}b \and D. Cadamuro}

\date{23. October 2022, revised 18. October 2024 and 03. March 2025}

\maketitle
\begin{abstract}
We consider the simplest non-trivial local composite operators in the massless Sine-Gordon model, which are $\partial_\mu \phi \, \partial_\nu \phi$ and the stress tensor $T_{\mu\nu}$. We show that even in the finite regime $\beta^2 < 4 \pi$ of the theory, these operators need additional renormalisation (beyond the free-field normal-ordering) at each order in perturbation theory. We further prove convergence of the renormalised perturbative series for their expectation values, both in the Euclidean signature and in Minkowski space-time, and for the latter in an arbitrary Hadamard state. Lastly, we show that one must add a quantum correction (proportional to $\hbar$) to the renormalised stress tensor to obtain a conserved quantity.
\end{abstract}

\section{Introduction}
\label{sec_intro}

The Sine--Gordon model is a well-studied example of a two-dimensional interacting quantum field theory. Its classical Euclidean action is given by
\begin{equation}
S[\phi] = \int \left[ \frac{1}{2} \partial^\mu \phi \partial_\mu \phi + \frac{1}{2} m_0^2 \phi^2 - g ( V_\beta + V_{-\beta} ) \right] \total^2 x \eqend{,}
\end{equation}
where $V_{\pm \beta} \equiv \mathe^{\pm \mathi \beta \phi}$ are the vertex operators, $m_0 \geq 0$ is a mass parameter, $\beta > 0$ is the coupling constant and $g$ is the interaction cutoff. While a priori, one takes $g$ as a function of compact support or rapid decay (Schwartz function), ultimately one is interested in the adiabatic or infinite-volume limit $g \to \text{const}$. The quantisation of the Sine--Gordon model, in various ranges of the parameters and using diverse approaches, has been treated by many authors, with the earliest results in the framework of Euclidean Constructive Quantum Field Theory. It turns out that for $\beta^2 < 4 \pi$ (the finite regime), after Wick ordering of the vertex operators a convergent perturbation theory in $g$ is obtained. In this regime and for sufficiently large $\abs{ m_0/g }$, Fr{\"o}hlich and Seiler~\cite{froehlich1976,froehlichlectures,froehlichseiler1976} proved convergence of the perturbation expansion of the Euclidean correlation functions of the field $\phi$ and the vertex operators $V_{\pm \beta}$ in infinite volume. In their proof, the mass term was essential to regulate the infrared problems that appear for the massless free scalar field. They also showed the existence of single-particle states and non-trivial scattering.

Still in the finite regime $\beta^2 < 4 \pi$, the existence of the massless limit for the infinite-volume Euclidean correlation functions of vertex operators $V_{\pm \beta}$ and the derivative of the interacting field $\partial_\mu \phi$ was shown by Fr{\"o}hlich and Park~\cite{park1977,froehlichpark1977}, who also proved that the Osterwalder--Schrader axioms are satisfied. For $4 \pi \leq \beta^2 < 8 \pi$ (the super-renormalisable regime), Wick ordering is not sufficient anymore, and a new divergent term that needs to be renormalised appears in perturbation theory each time $\beta^2$ crosses a threshold $n/(n+1) 8 \pi$. In this regime, the ultraviolet stability of the Sine--Gordon model (i.e., the convergence of the renormalised partition function in finite volume together with exponential bounds) has been shown by Benfatto, Gallavotti, Nicol{\`o}, Renn and Steinmann~\cite{benfattogallavottinicolo1982,nicolo1983,nicolorennsteinmann1986} using cluster expansions and Dimock and Hurd~\cite{dimockhurd1993,dimockhurd2000} and Renn and Steinmann~\cite{rennsteinmann1986} using renormalisation group techniques, for both the massive and the massless case. For $\beta^2 < 16/3 \pi$ (the second threshold), Dimock~\cite{dimock1998} has also shown the existence of correlation functions of the vertex operators in finite volume. Using Hamilton--Jacobi-like (or Wilson--Polchinski) flow equations, Brydges and Kennedy~\cite{brydgeskennedy1987} showed convergence of the partition function in the massive case and infinite volume for $\beta^2 < 16/3 \pi$, and Bauerschmidt and Bodineau~\cite{bauerschmidtbodineau2020} extended their results to $\beta^2 < 6 \pi$.

Finally, for $\beta^2 = 8 \pi$ the Sine--Gordon model becomes strictly renormalisable. Here, only perturbative renormalisability has been proven by Nicol{\`o} and Perfetti~\cite{nicoloperfetti1989}, while the non-perturbative existence of the model is unknown.

In the full range $0 \leq \beta^2 \leq 8 \pi$, the Sine-Gordon model has been conjectured to be equivalent to the Thirring model, entailing a boson-fermion equivalence (Coleman's equivalence~\cite{coleman1975}, see also~\cite{schroertruong1977,morchiopierottistrocchi1992}). Under this equivalence, correlation functions of vertex operators and derivatives of $\phi$ in the Sine--Gordon model are equal to correlation functions of fermion bilinears and currents in the Thirring model if the parameters of both models are suitably identified; in particular, the Sine--Gordon coupling $g$ is identified with the mass of the fermion while $\beta$ is related to the current-current coupling $\lambda$ in the Thirring model. This equivalence has been proven for $\beta^2 < 4 \pi$ by Fr{\"o}hlich and Seiler~\cite{froehlichseiler1976} in the massive case, by Benfatto, Falco and Mastropietro~\cite{benfattofalcomastropietro2009} in the massless case and in finite volume, and by Dimock~\cite{dimock1998} for $\beta^2 = 4 \pi$ (where the Thirring model becomes free) also in finite volume. Only recently, Bauerschmidt and Webb~\cite{bauerschmidtwebb2020} achieved a proof of this correspondence for the massless case in infinite volume, also for $\beta^2 = 4\pi$.

On the other hand, the classical massless Sine--Gordon model is integrable and one expects that integrability survives quantisation, such that the S-matrix factorises into two-particle scattering functions. Their form for the Sine--Gordon model has been conjectured by Zamolodchikov and Zamolodchikov~\cite{zamolodchikov1979}, but the integrable structure and the factorisation of the S-Matrix are not visible in the previous constructions. The conjectured S-matrix of the massless Sine--Gordon model has been studied in the form factor programme by Babujian, Karowski and collaborators~\cite{babujianetal1999,babujiankarowski2002}. In this approach one computes Wightman $n$-point functions of interacting pointlike local fields in terms of certain matrix components (``form factors'') of these. However, one has to deal with infinite expansions whose convergence is difficult to control, and therefore the existence of the fields themselves is currently out of reach in this approach.

A rigorous construction of the massless Sine--Gordon model has been achieved directly in Minkowski spacetime in the framework of perturbative Algebraic Quantum Field Theory (pAQFT) by Bahns and Rejzner~\cite{bahnsrejzner2018}. In the finite regime, they proved that the perturbation series for the S-matrix with fixed interaction cutoff, as well as the derivative of the interacting field $\partial_\mu \phi$ and the vertex operators $V_{\pm\beta}$, which are given as formal power series both in the coupling $g$ and in $\hbar$, converge. However, the expected factorisation of the S-matrix has not been shown, which is probably only visible in the adiabatic (infinite volume) limit. In a later paper together with Fredenhagen~\cite{bahnsrejznerfredenhagen2021}, the authors also constructed a family of unitary operators (relative S-matrices) which generate the local algebras of observables (vertex operators and derivative of $\phi$) of the model, and discussed the equivalence with the massive Thirring model.

In the general framework of Algebraic Quantum Field Theory, an alternative new approach to the construction of integrable quantum field theories has been carried out by Lechner starting from an idea of Schroer~\cite{schroer1997}. In this approach, a model is characterized in terms of its $C^\ast$-algebras of local observables obeying certain consistency conditions (Haag-Kastler axioms). The factorized S-matrix is an input to the construction of the theory. This approach uses as its starting point observables localized in wedge regions (wedge commutativity) and shows existence of strictly local observables in a second step, using abstract methods based on the theory of von Neumann algebras. This leads to a fully rigorous construction of the theory for a large class of scalar S-matrices~\cite{lechner2008}. While this class does not include the massless Sine--Gordon model, there are recent steps towards the Sine--Gordon model by Cadamuro and Tanimoto~\cite{cadamurotanimoto2018}. A characterisation of local observables on the level of expansion coefficients into an infinite series of interacting creation and annihilation operators has been carried out for scalar S-matrix models by Bostelmann and Cadamuro~\cite{bostelmanncadamuro2015}; for the massive Ising model, this leads to a rigorous construction of local observables~\cite{bostelmanncadamuro2019}. However, outside these special cases, showing convergence of the series remains difficult.

The Sine--Gordon model has also been studied in the framework of stochastic quantisation, where one introduces a stochastic partial differential equation depending on an auxiliary (``stochastic'') time $\tau$. Computing equal-time stochastic expectation values of the solutions to this stochastic PDE, the Euclidean correlation functions in the quantum theory arise in the limit $\tau \to \infty$. Using Hairer's framework of regularity structures to solve the stochastic PDE for the Sine--Gordon model, Chandra, Hairer and Shen~\cite{hairershen2014,chandrahairershen2018} have shown the short-time existence of solutions in the finite and super-renormalisable regime $0 \leq \beta^2 < 8 \pi$, but the existence of the infinite-time limit is unproven so far. The measure of the massive Sine--Gordon model in finite volume and for $\beta^2 < 4 \pi$ has also been constructed using stochastic control techniques by Oh, Robert, Sosoe and Wang~\cite{ohrobertsosoewang2021} and by Barashkov~\cite{barashkov2022}.

In this paper, we consider the simplest local composite operators of the massless ($m_0 = 0$) Sine--Gordon model (beyond vertex operators and the derivative of the field), which are $\op_{\mu\nu} \equiv \partial_\mu \phi \, \partial_\nu \phi$ and the stress tensor $T_{\mu\nu} \equiv \op_{\mu\nu} - \frac{1}{2} \eta_{\mu\nu} \eta^{\rho\sigma} \op_{\rho\sigma} + g \eta_{\mu\nu} ( V_\beta + V_{-\beta} )$ (with $\eta_{\mu\nu}$ replaced by $\delta_{\mu\nu}$ in the Euclidean case). We consider both Euclidean and Minkowski signature and work in the finite regime ($\beta^2 < 4 \pi$) of the massless theory. To that end, we use the well-known Gell-Mann--Low formula for the perturbation series of interacting fields, and show convergence of the renormalized perturbation series with an adiabatic interaction cutoff $g$, after removal of the initial IR and UV cutoffs. While we do not attempt to remove the adiabatic cutoff, we show convergence for arbitrary Schwartz functions $g$, which in the Minkowski case represents a technical improvement over~\cite[Thm.~6]{bahnsrejzner2018}, whose results only hold if the support of $g$ is small enough.

In the Minkowski case, we use the general framework of pAQFT, whose main advantage is the clean separation of algebraic issues (including renormalisation) from the construction of a state. In contrast to the traditional theoretical physics approach to renormalisation, which involves the introduction of a UV cutoff and counterterms, in the pAQFT approach the renormalisation problem is reduced to the extension of certain distributions to coinciding points, and no explicit counterterms are needed. Nevertheless, it has been shown~\cite{popineaustora2016,pinter2001,duetschfredenhagen2004} that the freedom in the extension exactly corresponds to the usual finite renormalisation freedom, i.e., to the choice of the finite parts of counterterms. However, to show the similarities with the Euclidean case and connect to the traditional approach, we extend the established pAQFT approach and consider explicit cutoffs also in the Minkowski case. This has both advantages and disadvantages: with a finite UV cutoff, the appearing distributions can be extended to coinciding points simply by continuity, since the regularised distributions are smooth functions. Since in general the construction of the extension is quite involved~\cite{brunettifredenhagen2000,hollandswald2002,bahnswrochna2012}, this is a major simplification. However, the problem reappears in the limit of vanishing cutoff, where we have to show that suitable counterterms exist to cancel the resulting divergences. This is a disadvantage of our approach, which is not present in the usual pAQFT approach. Moreover, we have to show that these counterterms are local, i.e., proportional to Dirac $\delta$'s and their derivatives. This then also shows (for the distributions in question) that we obtain the same results as the conventional pAQFT approach, since the original distribution was already well-defined for non-coinciding points and only the extension to the diagonal must be performed. In particular, in Lemmas~\ref{lemma_euclid_distribution} and~\ref{lemma_mink_distribution} we explicitly determine the finite renormalisation freedom, which is the same for the Euclidean and Minkowski case. Requiring the extension (or renormalised distribution) to be Euclidean or Lorentz covariant and preserving the scaling degree, we find exactly the same result as in the pAQFT approach. We also note that a similar approach to ours was previously used to implement a version of dimensional regularisation in the pAQFT framework~\cite{duetschfredenhagenkellerrejzner2014}.

Our results are as follows: We start with the case of Euclidean signature. In view of the Gell-Mann--Low formula, we first show that the renormalised expectation values of $\op_{\rho\sigma}$ and $T_{\rho\sigma}$ are well-defined (in the sense of distributions) after removal of the IR and UV cutoffs:
\begin{theorem}[Renormalisation in Euclidean signature]
\label{thm_euclid_renorm}
Consider the massless Euclidean Sine-Gordon model in the finite regime $\beta^2 < 4\pi$ and with the free-field covariance $C^{\Lambda,\epsilon}$ with IR cutoff $\Lambda$ and UV cutoff $\epsilon$. There exists a choice of local counterterms (diverging logarithmically with the UV cutoff $\epsilon$) such that the renormalised expectation values
\begin{equation}
\label{eq:thm_euclid_renorm_expect}
\expect{ \mathcal{N}_\mu\left[ \op_{\rho\sigma}(z) \right] \prod_{j=1}^n \mathcal{N}_\mu\left[ V_{\sigma_j \beta}(x_j) \right] }^{\Lambda,\epsilon}_{0,\text{ren}}
\end{equation}
in the free theory, with $\mathcal{N}_\mu$ denoting normal ordering with respect to the covariance $C^{\mu,\epsilon}$ (where $\mu$ is a fixed scale) and $\sigma_j = \pm 1$, are well-defined distributions in the physical limit $\Lambda,\epsilon \to 0$. For the stress tensor, the physical limit
\begin{equation}
\lim_{\Lambda,\epsilon \to 0} \expect{ \mathcal{N}_\mu\left[ T_{\rho\sigma}(z) \right] \prod_{j=1}^n \mathcal{N}_\mu\left[ V_{\sigma_j \beta}(x_j) \right] }^{\Lambda,\epsilon}_{0,\text{ren}}
\end{equation}
is finite even without subtracting counterterms. The expectation values involving $\op_{\rho\sigma}$ vanish in the physical limit unless the neutrality condition $\sum_{j=1}^n \sigma_j = 0$ is fulfilled, while the ones involving the stress tensor vanish unless $\sum_{j=1}^n \sigma_j \in \{-1,0,1\}$.
\end{theorem}
\begin{remark*}
The required counterterms can be obtained by combining Eqs.~\eqref{eq:euclid_renorm_expecteps}, \eqref{eq:umunu_def} and~\eqref{eq:umunu_div}, and read (in the limit $\Lambda \to 0$)
\begin{equation}
- \frac{\beta^2}{4 \pi} \delta_{\rho\sigma} \ln(\mu \epsilon) \prod_{1 \leq j < k \leq n} \left[ \mu^2 (x_j-x_k)^2 \right]^{\sigma_j \sigma_k \frac{\beta^2}{4 \pi}} \sum_{k=1}^n \delta(x_k-z) \eqend{,}
\end{equation}
i.e., for each vertex operator there is one local counterterm. Since the counterterms are proportional to $\delta_{\rho\sigma}$, their contribution to the expectation value of the stress tensor cancels.
\end{remark*}

We then show convergence of the renormalised Gell-Mann--Low perturbation series with an adiabatic cutoff in the physical limit, i.e., for vanishing IR and UV cutoffs:
\begin{theorem}[Convergence of the renormalised perturbation series]
\label{thm_euclid_conv}
Under the same assumptions as in Theorem~\ref{thm_euclid_renorm} and with a non-negative adiabatic cutoff function $0 \leq g \istest$, the perturbative series for the (normalised, interacting) Gell-Mann--Low expectation value of $\op_{\rho\sigma}$
\begin{splitequation}
\label{eq:thm_euclid_conv_series}
&\expect{ \mathcal{N}_\mu\left[ \op_{\rho\sigma}(f) \right] }_\text{int,ren} \equiv \\
&\frac{ \sum_{n=0}^\infty \frac{1}{n!} \int\dotsi\int \sum_{(\sigma_1,\ldots,\sigma_n) \in \{\pm 1\}^n} \expect{ \mathcal{N}_\mu\left[ \op_{\rho\sigma}(f) \right] \prod_{j=1}^n \mathcal{N}_\mu\left[ V_{\sigma_j \beta}(x_j) \right] }^{0,0}_{0,\text{ren}} \prod_{i=1}^n g(x_i) \total^2 x_i }{ \sum_{n=0}^\infty \frac{1}{n!} \int\dotsi\int \sum_{(\sigma_1,\ldots,\sigma_n) \in \{\pm 1\}^n} \expect{ \prod_{j=1}^n \mathcal{N}_\mu\left[ V_{\sigma_j \beta}(x_j) \right] }^{0,0}_{0,\text{ren}} \prod_{i=1}^n g(x_i) \total^2 x_i }
\end{splitequation}
is convergent, where the operator $\op_{\rho\sigma}$ is smeared with a test function $f \istest$ and the physical limit $\Lambda, \epsilon \to 0$ is taken termwise. There exists a constant $K > 0$ (depending on $g$ and $\beta$) and a constant $C > 0$ (depending on $f$) such that the renormalised expectation value is bounded by
\begin{equation}
\label{eq:thm_euclid_conv_bound}
C \sum_{n=0}^\infty n^2 K^n \left.\begin{cases} (n!)^{-1} & \beta^2 < 2 \pi \\ (n!)^{\frac{\beta^2}{2 \pi} - 2} & 2 \pi \leq \beta^2 < 4 \pi \end{cases} \right\} < \infty \eqend{.}
\end{equation}
The same holds for the smeared stress tensor $T_{\rho\sigma}(f)$, with the constant $C$ also depending on $g$.
\end{theorem}
\begin{remark*}
While we take the limit $\Lambda, \epsilon \to 0$ termwise, we believe that it would also be possible to show convergence with finite cutoffs and take the limit after resumming the series. The main issue here is a technical one: to prove the required bounds, we use a well-known formula~\eqref{eq:cauchy_determinant} to transform the renormalised expectation value at each order into the determinant of a Cauchy matrix, which is necessary to show absolute convergence of the series. For finite cutoffs, it is still possible to write the renormalised expectation value at each order as the determinant of a matrix (see for example~\cite{bahnsrejzner2018}), which is however much more complicated. We do not see any fundamental difficulty to prove convergence of the perturbative series for the Gell-Mann--Low expectation value of $\op_{\rho\sigma}$ also in this case, such that taking the limit term by term is justified. However, the computations become much more involved, such that we restrict to the termwise limit.
\end{remark*}

Lastly, we show that a modified stress tensor, obtained by a rescaling of the coupling constant, fulfills the continuity equation in the quantum theory:
\begin{theorem}[Conservation of the stress tensor]
\label{thm_euclid_cons}
Under the same assumptions as in Theorem~\ref{thm_euclid_conv}, a modified stress tensor $\hat{T}_{\rho\sigma}$ is conserved in the quantum theory in expectation: we have
\begin{equation}
\label{eq:thm_euclid_cons_eq}
\expect{ \mathcal{N}_\mu\left[ \hat{T}_{\rho\sigma}(\partial^\rho f) \right] }_\text{int,ren} = 0
\end{equation}
for all $f \istest$ such that $g$ is constant on the support of $f$. The required modification is a rescaling of the coupling $g$:
\begin{equation}
\label{eq:thm_euclid_cons_stresstensor}
\hat{T}_{\mu\nu} = \op_{\mu\nu} - \frac{1}{2} \delta_{\mu\nu} \op_\rho{}^\rho + g \left( 1 - \frac{\beta^2}{8 \pi} \right) \delta_{\mu\nu} ( V_\beta + V_{-\beta} ) \eqend{.}
\end{equation}
\end{theorem}
\begin{remark*}
In general, we expect that the modified stress tensor fulfills appropriate Ward identities associated to the conservation, of the form
\begin{splitequation}
\label{eq:euclid_cons_ward}
&\expect{ \mathcal{N}_\mu\left[ \hat{T}_{\rho\sigma}(\partial^\rho f) \right] \phi(x_1) \cdots \phi(x_n) }_\text{int,ren} \\
&\quad= - \sum_j \partial_\sigma f(x_j) \expect{ \phi(x_1) \cdots \phi(x_{j-1}) \phi(x_{j+1}) \cdots \phi(x_n) }_\text{int,ren} \eqend{.}
\end{splitequation}
The result~\eqref{eq:thm_euclid_cons_eq} is then the case $n = 0$, and we believe that it is possible to show the identities~\eqref{eq:euclid_cons_ward} using the methods that we are using. However, the computations become much more involved, and we leave a verification of the Ward identities to future work.
\end{remark*}

Similar results apply to the case of Minkowski signature. Namely, we show that the renormalised expectation values of time-ordered products involving $\op_{\mu\nu}$ and $T_{\mu\nu}$ in any quasi-free Hadamard state, regularized with IR and UV cutoffs, are well-defined in the sense of distributions when the cutoffs are removed. We prove:
\begin{theorem}[Renormalisation in Minkowski space-time]
\label{thm_mink_renorm}
Consider the massless Lorentzian Sine-Gordon model in the finite regime $\beta^2 < 4\pi$ and a quasi-free state $\omega^{\Lambda,\epsilon}$ in the vacuum sector whose two-point function has an IR cutoff $\Lambda$ and UV cutoff $\epsilon$. There exists a choice of local counterterms (diverging logarithmically with the UV cutoff $\epsilon$) such that the renormalised expectation values of time-ordered products
\begin{equation}
\omega^{\Lambda,\epsilon}\left( \mathcal{T}\left[ \op_{\mu\nu}(z) \otimes \bigotimes_{j=1}^n V_{\sigma_j\beta}(x_j) \right] \right)
\end{equation}
with $\sigma_j = \pm 1$ in the free theory are well-defined distributions in the physical limit $\Lambda,\epsilon \to 0$. For the stress tensor, the physical limit
\begin{equation}
\lim_{\Lambda,\epsilon \to 0} \omega^{\Lambda,\epsilon}\left( \mathcal{T}\left[ T_{\mu\nu}(z) \otimes \bigotimes_{j=1}^n V_{\sigma_j \beta}(x_j) \right] \right)
\end{equation}
is finite even without subtracting counterterms. The expectation values involving $\op_{\mu\nu}$ vanish in the physical limit unless the neutrality condition $\sum_{j=1}^n \sigma_j = 0$ is fulfilled, while the ones involving the stress tensor vanish unless $\sum_{j=1}^n \sigma_j \in \{-1,0,1\}$.
\end{theorem}
\begin{remark*}
While there are no explicit counterterms in the conventional pAQFT treatment and renormalisation is performed by extending certain distributions to the diagonal, here we have chosen to connect pAQFT with the traditional approach using cutoffs and counterterms. We show that the pAQFT methods can be extended to this case, and that the divergent terms can be removed by a local redefinition of time-ordered products in the spirit of pAQFT. This redefinition furnishes the required counterterms, which are given by Eqs.~\eqref{eq:mink_renorm_redef}, \eqref{eq:mink_renorm_redef_2} and~\eqref{eq:min_renorm_divconst}, and which read explicitly
\begin{equation*}
\delta \mathcal{T}\left[ \op_{\mu\nu}(z) \otimes V_\alpha(x) \right] = \frac{\beta^2}{(4 \pi)^2} \, H^\text{div}_{\mu\nu}(z,x) \, \mathcal{T}\left[ V_\alpha(x) \right] \eqend{.}
\end{equation*}
The distribution $H^\text{div}_{\mu\nu}$ is defined by equation~\eqref{eq:hmunu_div}, and is proportional to $\delta(x-z)$ and thus local as required.
\end{remark*}

In analogy to the Euclidean case, we then show convergence of the renormalised perturbation series given by the Gell-Mann--Low formula for $\op_{\mu\nu}$ and $T_{\mu\nu}$ with interaction cutoff, and under an additional assumption on the state-dependent part of the two-point function of the quasi-free state.
\begin{theorem}[Convergence of the renormalised perturbation series]
\label{thm_mink_conv}
We make the same assumptions as in Theorem~\ref{thm_mink_renorm}, and in addition require that the state-dependent part $W$ of the two-point function of the state $\omega$ satisfies:
\begin{enumerate}
\item $W(x,y)$ and its first and second derivatives grow at most polynomially,
\item $\sum_{i,j=1}^n [ W(x_i,x_j) - W(y_i,x_j) - W(x_i,y_j) + W(y_i,y_j) ] \geq 0$ for any configuration of points $x_i$ and $y_i$ and any $n \in \mathbb{N}$.
\end{enumerate}
Then for any adiabatic cutoff function $g \istest$, the perturbative series for the (normalised, interacting) Gell-Mann--Low expectation value
\begin{splitequation}
\label{eq:thm_mink_conv_series}
&\omega_\text{int}\left( \op_{\mu\nu}(f) \right) \equiv \\
&\frac{ \sum_{n=0}^\infty \frac{(-\mathi)^n}{n!} \int\dotsi\int \sum_{(\sigma_1,\ldots,\sigma_n) \in \{\pm 1\}^n} \omega^{0,0}\left( \mathcal{T}\left[ \op_{\mu\nu}(f) \otimes \bigotimes_{j=1}^n V_{\sigma_j \beta}(x_j) \right] \right) \prod_{i=1}^n g(x_i) \total^2 x_i }{ \sum_{n=0}^\infty \frac{(-\mathi)^n}{n!} \int\dotsi\int \sum_{(\sigma_1,\ldots,\sigma_n) \in \{\pm 1\}^n} \omega^{0,0}\left( \mathcal{T}\left[ \bigotimes_{j=1}^n V_{\sigma_j \beta}(x_j) \right] \right) \prod_{i=1}^n g(x_i) \total^2 x_i }
\end{splitequation}
is convergent, where the operator $\op_{\mu\nu}$ is smeared with a test function $f \istest$ and the physical limit $\Lambda, \epsilon \to 0$ is taken termwise. There exists a constant $K > 0$ (depending on $g$, $\beta$ and $W$) and a constant $C > 0$ (depending on $f$, $\beta$ and $W$) such that the renormalised expectation value is bounded by
\begin{equation}
\label{eq:thm_mink_conv_bound}
C \sum_{n=0}^\infty n^2 K^n (n!)^{\frac{\beta^2}{4 \pi} - 1} < \infty
\end{equation}
if $\norm{ g }_\infty$ is small enough (depending on $g$ and $\beta$). The same holds for the smeared stress tensor $T_{\mu\nu}(f)$, with the constant $C$ also depending on $g$.
\end{theorem}
\begin{remark*}
The same remark on the termwise limit $\Lambda, \epsilon \to 0$ given after Thm.~\ref{thm_euclid_conv} applies here. It is difficult to make the condition on $\norm{ g }_\infty$ explicit, since it involves a \emph{lower} bound on the series in the denominator which is not easy to obtain. In contrast to the Euclidean case, where the condition $g \geq 0$ is enough to give a lower bound of $1$ for the denominator (from the $n = 0$ term), in the Lorentzian case the integrand does not have a definite sign for any choice of $g$, and cancellations may occur. However, since the series is convergent, we may always rescale $g \to \hat{g} = \lambda g$ with a sufficiently small $\lambda > 0$ to obtain a lower bound of (say) $\frac{1}{2}$ for the denominator. This results in a small $\norm{ \hat{g} }_\infty = \lambda \norm{ g }_\infty$, but the choice of $\lambda$ heavily depends on the concrete $g$.

While the first condition on the state-dependent part $W$ of the two-point function is very reasonable, the second one might seem strange at first sight. It is however necessary for convergence of the perturbative series, and one can easily construct a wide range of Hadamard states that fulfill it. For example, any state for which $W$ can be written as
\begin{equation}
W(x,y) = \int \mathe^{\mathi p (x-y)} \total \mu_W(p)
\end{equation}
with a positive measure $\total \mu_W(p)$ fulfills the condition, since then
\begin{splitequation}
&\sum_{i,j=1}^n [ W(x_i,x_j) - W(y_i,x_j) - W(x_i,y_j) + W(y_i,y_j) ] \\
&\quad= \int \abs{ \sum_{i=1}^n \left( \mathe^{\mathi p x_i} - \mathe^{\mathi p y_i} \right) }^2 \total \mu_W(p) \geq 0 \eqend{.}  
\end{splitequation}
Examples of such states are states that resemble a thermal state in a wide range of energies between $E_0$ and $E_1$, where
\begin{equation}
\total \mu_W(p) = \Theta\left( \abs{p^0} \in [E_0, E_1] \right) \frac{\pi [ \delta(p^1+p^0) + \delta(p^1-p^0) ]}{\abs{p^0} \left[ \exp\left( \frac{\abs{p^0}}{k_\text{B} T} \right) - 1 \right]} \frac{\total^2 p}{(2\pi)^2}
\end{equation}
with $k_\text{B}$ the Boltzmann constant and $T$ the temperature, which is clearly a positive measure. Here the restriction on the energy range is necessary, since we start with a free massless theory where thermal states do not exist in a strict sense, which can be seen from the singularity of the prefactor for small $\abs{p^0}$. On the other hand, for the massive Sine--Gordon model true thermal states have recently been constructed in~\cite{bahnspinamontirejzner2021}. While it is expected that the interacting theory acquires a mass even if $m_0 = 0$ (and in fact explicit expressions for the mass exist~\cite{zamolodchikov1995}), such that thermal states would exist also for the interacting massless Sine--Gordon model, this mass generation has only been proven very recently for the special value $\beta^2 = 4 \pi$~\cite{bauerschmidtwebb2020}.
\end{remark*}

Finally, we show that analogous to the Euclidean case we need to modify the stress tensor by a rescaling of the coupling constant, such that it fulfills the continuity equation also in the Minkowski signature. This modification agrees with the one proposed in the form factor programme~\cite{babujiankarowski2002}, and thus proves its correctness. At one-loop order, the conservation of the modified stress tensor had been shown previously by Alberti, Schlesier and Zahn~\cite{albertietal2023}, and our result extends conservation to all orders.
\begin{theorem}[Conservation of the stress tensor]
\label{thm_mink_cons}
Under the same assumptions as in Theorem~\ref{thm_mink_conv}, the expectation value of the modified stress tensor
\begin{equation}
\label{eq:thm_mink_cons_stresstensor}
\hat{T}_{\mu\nu} \equiv \op_{\mu\nu} - \frac{1}{2} \eta_{\mu\nu} \op_\rho{}^\rho + g \left( 1 - \frac{\beta^2}{8 \pi} \right) \eta_{\mu\nu} ( V_\beta + V_{-\beta} )
\end{equation}
is conserved in the quantum theory:
\begin{equation}
\label{eq:thm_mink_cons_eq}
\omega_\text{int}\left( \hat{T}_{\mu\nu}(\partial^\mu f) \right) = 0
\end{equation}
for all $f \istest$ such that $g$ is constant on the support of $f$.
\end{theorem}
\begin{remark*}
We expect that Ward identities analogous to the Euclidean case~\eqref{eq:euclid_cons_ward} hold. However, their verification becomes quite involved, and we leave them to future work.
\end{remark*}
The remaining sections are dedicated to the proof of these theorems.

Our conventions are as follows: We set $\hbar = c = 1$. Note that if we would keep $\hbar$ explicit, it would only arise in the combination $\hbar \beta^2$, so that for example the modifications of the stress tensor $\hat{T}_{\mu\nu}$ are really quantum modifications. Interestingly, they are one-loop exact, i.e., no terms of order $\hbar^2$ or higher arise. For norms of functions, we also use the notation
\begin{equation}
\norm{ f(\blank) }_p \equiv \left[ \int \abs{f(x)}^p \total^d x \right]^\frac{1}{p}
\end{equation}
with $d \in \{1,2\}$ (which will always be clear from the context), $p \geq 1$ and the usual modification for $p = \infty$.

An obvious extension of our work would be the generalisation of the results to the higher conserved currents of the Sine--Gordon model which exist classically as a consequence of integrability, as well as in perturbation theory~\cite{flume1975,kulishnissimov1976,nissimov1977,lowensteinspeer1978}. Their classical expression is obtained from an explicit recurrence relation, and in the form factor programme it has been proposed that --- in contrast to the stress tensor --- they do not receive quantum corrections~\cite{babujiankarowski2002}. While renormalisability of the interacting higher currents has been shown~\cite{zanello2023}, both the renormalisation of correlation functions involving them and the proof of convergence of the perturbative series are still missing. Another important step that is still unclear is the adiabatic (infinite-volume) limit $g \to \text{const}$, which already for correlation functions of the vertex operators is a difficult task. To show the existence of this limit, one probably has to first show the non-perturbative generation of a finite mass (Debye screening), which is known to arise at least for small $\abs{\beta g}$~\cite{brydgesfederbush1980,yang1987} and in the case $\beta^2 = 4 \pi$~\cite{bauerschmidtwebb2020}. Also the extension of our results to the super-renormalisable range $4 \pi \leq \beta^2 < 8 \pi$ is a worthwhile task to achieve in the future. Lastly, while the equivalence with the Thirring model has been proven for correlation functions of the vertex operators and derivatives of the field, their equality for correlation functions involving the stress tensor or higher-order conserved currents (which exist in the Thirring model~\cite{nissimov1977,lowensteinspeer1978}) is so far unknown, and needs to be studied.

\section{Euclidean case}
\label{sec_euclid}

In the construction of the theory in Euclidean signature, we follow the well-established treatment of Euclidean quantum field theories via functional integrals. An introduction can be found in~\cite{salmhofer1999}, from which we take formulas and results without specifying their source explicitly.

\subsection{Preliminaries}
\label{sec_euclid_pre}

We consider a centred Gaussian measure $\total \mu^{\Lambda,\epsilon}(\phi)$ with covariance $C^{\Lambda,\epsilon}(x,y)$ depending on an IR cutoff $\Lambda$ and a UV cutoff $\epsilon$. Centred and Gaussian means that
\begin{equation}
\label{eq:gaussian_measure}
\int \phi(x_1) \cdots \phi(x_n) \total \mu^{\Lambda,\epsilon}(\phi) = \sum_{\pi} \prod_{\{i,j\} \in \pi} C^{\Lambda,\epsilon}(x_i,x_j) \eqend{,}
\end{equation}
where $x_i \in \mathbb{R}^2$, $i \in \{1,\ldots,n\}$, and the sum runs over all partitions $\pi$ of the set $\{ 1,\ldots,n \}$ into unordered pairs $\{i,j\}$. Therefore, the right-hand side vanishes if $n$ is odd. In general, we have the characteristic function
\begin{equation}
\label{eq:gaussian_characteristic}
\int \mathe^{\mathi (J,\phi)} \total \mu^{\Lambda,\epsilon}(\phi) = \mathe^{- \frac{1}{2} (J, C^{\Lambda,\epsilon} \ast J)}
\end{equation}
where we introduced the scalar product
\begin{equation}
\label{eq:scalar_product}
\left( f, g \right) \equiv \int f(x) g(x) \total^2 x
\end{equation}
and the convolution
\begin{equation}
\label{eq:convolution}
\left( C \ast g \right)(x) \equiv \int C(x,y) g(y) \total^2 y \eqend{,}
\end{equation}
from which the expression~\eqref{eq:gaussian_measure} follows by functional differentiation with respect to $J \istest$. For finite cutoffs, the covariance is assumed to be a smooth function, and hence the Gaussian measure is supported on smooth functions $\phi \istest$. In general, and in particular as the cutoffs are removed, the measure is supported on distributions. However, the combinatorics are unaffected by the support properties of the measure. Obviously, as the cutoffs are removed in the physical limit $\Lambda,\epsilon \to 0$, the covariance must turn into the free propagator. Regulated expectation values in the free theory are defined by the Gaussian integral~\eqref{eq:gaussian_measure}
\begin{equation}
\expect{ A_1[\phi] \cdots A_n[\phi] }^{\Lambda,\epsilon}_0 \equiv \int A_1[\phi] \cdots A_n[\phi] \total \mu^{\Lambda,\epsilon}(\phi) \eqend{,}
\end{equation}
where the $A_i$ are local functionals of the field, and the interacting expectation values are computed from the Gell-Mann--Low formula
\begin{equation}
\expect{ A_1[\phi] \cdots A_n[\phi] }^{\Lambda,\epsilon}_\text{int} \equiv \frac{\expect{ A_1[\phi] \cdots A_n[\phi] \mathe^{- S_\text{int}[\phi]} }^{\Lambda,\epsilon}_0}{\expect{ \mathe^{- S_\text{int}[\phi]} }^{\Lambda,\epsilon}_0} \eqend{,}
\end{equation}
where $S_\text{int}$ is the interaction which is assumed to be bounded from below.

It is well known that the massless scalar field in two dimensions is not well defined because of infrared problems~\cite{schroer1963,wightmanlectures}. In our case, this manifests as a logarithmic divergence of the regulated covariance
\begin{equation}
\label{eq:covariance}
C^{\Lambda,\epsilon}(x,y) \equiv - \frac{1}{4 \pi} \ln \left[ \Lambda^2 \left[ (x-y)^2 + \epsilon^2 \right] \right]
\end{equation}
as the IR cutoff $\Lambda$ is removed. To verify that~\eqref{eq:covariance} is a suitable regularisation, we note that the UV cutoff $\epsilon$ obviously regulates the UV singular behaviour as $x \to y$, while the IR cutoff $\Lambda$ arises from the small-mass limit of the massive propagator:
\begin{splitequation}
\label{eq:massive_covariance}
\int \frac{\mathe^{\mathi p (x-y)}}{p^2 + m^2} \frac{\total^2 p}{(2\pi)^2} &= \frac{1}{2\pi} \int_0^\infty \frac{q}{q^2 + m^2} \bessel{J}{0}{q \abs{x-y}} \total q \\
&= \frac{1}{2\pi} \bessel{K}{0}{m \abs{x-y}} = - \frac{1}{4\pi} \ln\left[ \frac{m^2 \mathe^{2\gamma}}{4} (x-y)^2 \right] + \bigo{m} \eqend{,}
\end{splitequation}
where we used the integrals~\cite[Eqs.~(10.9.1) and (10.22.43)]{dlmf} and the known expansion~\cite[Eq.~(10.31.2)]{dlmf} of the modified Bessel function $\mathrm{K}_0$ for small argument. We will see later on that this divergence is responsible for the super-selection sectors of the theory, and for the vacuum sector which we consider results in a neutrality condition~\cite{swieca1977}. However, the regulated covariance~\eqref{eq:covariance} is not positive definite, and thus does not define a Gaussian measure. This issue ultimately stems from the fact that the limit $m \to 0$ of the massive covariance~\eqref{eq:massive_covariance} does not exist, since the massive covariance is positive definite for all $m > 0$. One could solve the problem by use a covariance with a different IR regulator (for example, keeping a finite mass), but this would complicate the subsequent formulas. Moreover, it is unnecessary for our purposes, since we are only interested in the limit $\Lambda \to 0$ and correlation functions for which this limit exists, which from the above results is the same as the limit $m \to 0$. We can therefore develop the theory assuming a positive covariance, but for concrete computations nevertheless use the covariance~\eqref{eq:covariance}.

In the range $\beta^2 < 4 \pi$, one only needs to normal-order the interaction to obtain finite correlation functions of the basic field $\phi$. Normal-ordering $\mathcal{N}$ is a linear operation that can be defined with respect to any given covariance, but we only need it with respect to the covariance $C^{\mu,\epsilon}$~\eqref{eq:covariance} for a fixed IR cutoff $\mu$. As is well known, we have to keep $\mu$ fixed and different from the true IR cutoff $\Lambda$ of the theory (which we will later send to zero) to obtain a finite correlation function in the physical limit. We have the explicit formula for exponentials
\begin{equation}
\label{eq:normal_ordering_exponential_general}
\mathcal{N}_\mu\left[ \mathe^{\mathi (J,\phi)} \right] = \mathe^{\frac{1}{2} (J, C^{\mu,\epsilon} \ast J)} \mathe^{\mathi (J,\phi)} \eqend{,}
\end{equation}
from which normal-ordering of monomials can be obtained by functional differentiation with respect to $J$. Moreover, for $J(y) = \pm \beta \delta(y-x)$ we obtain the normal-ordering of the interaction of the Sine-Gordon theory:
\begin{equations}[eq:normal_ordering_exponential]
\mathcal{N}_\mu\left[ 2 \cos(\beta \phi(x)) \right] &= \mathcal{N}_\mu\left[ \mathe^{\mathi \beta \phi(x)} \right] + \mathcal{N}_\mu\left[ \mathe^{- \mathi \beta \phi(x)} \right] \eqend{,} \\
\mathcal{N}_\mu\left[ \mathe^{\mathi \alpha \phi(x)} \right] &= \mathe^{\frac{1}{2} \alpha^2 C^{\mu,\epsilon}(x,x)} \mathe^{\mathi \alpha \phi(x)} = ( \mu \epsilon )^{- \frac{\alpha^2}{4 \pi}} \mathe^{\mathi \alpha \phi(x)} \eqend{.}
\end{equations}

Normal-ordering has the property that
\begin{equation}
\label{eq:expectation_normal_ordered}
\expect{ \mathcal{N}_\Lambda\left[ (J_1,\phi) \cdots (J_n,\phi) \right] }^{\Lambda,\epsilon}_0 = \int \mathcal{N}_\Lambda\left[ (J_1,\phi) \cdots (J_n,\phi) \right] \total \mu^{\Lambda,\epsilon}(\phi) = \delta_{n,0} \eqend{,}
\end{equation}
which makes evaluation easy, and is the reason for the name ``normal-ordering''. For a normal ordering with respect to a different covariance, one first has to change the normal-ordering according to
\begin{splitequation}
&\mathcal{N}_\nu\left[ (J_1,\phi) \cdots (J_n,\phi) \right] \\
&\quad= \exp\left[ - \frac{1}{2} \left( \frac{\delta}{\delta \phi}, \left( C^{\nu,\epsilon} - C^{\mu,\epsilon} \right) \ast \frac{\delta}{\delta \phi} \right) \right] \mathcal{N}_\mu\left[ (J_1,\phi) \cdots (J_n,\phi) \right] \eqend{,}
\end{splitequation}
which is proven as in~\cite[Thm.~2.4]{salmhofer1999}. We note that for unsmeared exponentials and the covariance~\eqref{eq:covariance}, there is a particularly simple relation:
\begin{equation}
\label{eq:change_normal_ordering_exponential}
\mathcal{N}_\mu\left[ \mathe^{\pm \mathi \beta \phi(x)} \right] = \left( \frac{\Lambda}{\mu} \right)^\frac{\beta^2}{4 \pi} \mathcal{N}_\Lambda\left[ \mathe^{\pm \mathi \beta \phi(x)} \right] \eqend{,}
\end{equation}
and since the UV cutoffs $\epsilon$ on both sides are the same, this relation holds also in the limit $\epsilon \to 0$.

In the proofs in the following sections, we also need a formula which (to our knowledge) first appeared in~\cite[Lemma~2.6]{bauerschmidtwebb2020}, and which we generalise and formulate as a Lemma:
\begin{lemma}
\label{lemma_euclid_exponential_phi}
The expectation value of a product of exponentials and basic fields can be decomposed as
\begin{splitequation}
\label{eq:correlator_exponential_nphi}
&\expect{ \prod_{j=1}^m \mathe^{\mathi (f_j,\phi)} \prod_{k=1}^n (g_k,\phi) }^{\Lambda,\epsilon}_0 = \expect{ \prod_{j=1}^m \mathe^{\mathi \left( f_j,\phi \right)} }^{\Lambda,\epsilon}_0 \\
&\quad\times \left. \prod_{i=1}^n \frac{\partial}{\partial \sigma_{i}} \exp\left[ \mathi \sum_{j=1}^m \sum_{k=1}^n \sigma_k \left( f_j, C^{\Lambda,\epsilon} \ast g_k \right) + \sum_{\mathclap{1 \leq k < \ell \leq n}} \sigma_k \sigma_\ell \left( g_k, C^{\Lambda,\epsilon} \ast g_\ell \right) \right] \right\rvert_{\sigma_i = 0} \eqend{.}
\end{splitequation}
In particular, for $n = 2$ we obtain
\begin{splitequation}
\label{eq:correlator_exponential_twophi}
&\expect{ (g_1,\phi) (g_2,\phi) \prod_{j=1}^m \mathe^{\mathi (f_j,\phi)} }^{\Lambda,\epsilon}_0 = \expect{ \prod_{j=1}^m \mathe^{\mathi (f_j,\phi)} }^{\Lambda,\epsilon}_0 \\
&\quad\times \left[ \left( g_1, C^{\Lambda,\epsilon} \ast g_2 \right) - \sum_{j,k=1}^m \left( f_j, C^{\Lambda,\epsilon} \ast g_1 \right) \left( f_k, C^{\Lambda,\epsilon} \ast g_2 \right) \right] \eqend{.}
\end{splitequation}
\end{lemma}
\begin{proof}
We essentially transcribe the proof of~\cite{bauerschmidtwebb2020} from Gaussian random variables to functional integrals. We first recall the formula for a shifted Gaussian measure
\begin{equation}
\label{eq:gaussian_measure_shift}
\total \mu^{\Lambda,\epsilon}\left( \phi + C^{\Lambda,\epsilon} \ast h \right) = \mathe^{- \frac{1}{2} (h,C^{\Lambda,\epsilon} \ast h)} \mathe^{- (\phi,h)} \total \mu^{\Lambda,\epsilon}(\phi)
\end{equation}
with $h \istest$, which is easily proven by showing that the characteristic functions~\eqref{eq:gaussian_characteristic} of both measures are the same. 
We then compute
\begin{splitequation}
\label{eq:expectation_exponential_shift}
&\expect{ \prod_{j=1}^m \mathe^{\mathi (f_j,\phi)} \mathe^{(h,\phi)} }^{\Lambda,\epsilon}_0 = \int \prod_{j=1}^m \mathe^{\mathi (f_j,\phi)} \mathe^{(h,\phi)} \total \mu^{\Lambda,\epsilon}(\phi) \\
&= \int \prod_{j=1}^m \mathe^{\mathi (f_j,\phi + C^{\Lambda,\epsilon} \ast h)} \exp\left[ \left( h, \phi + C^{\Lambda,\epsilon} \ast h \right) \right] \total \mu^{\Lambda,\epsilon}\left( \phi + C^{\Lambda,\epsilon} \ast h \right) \\
&= \int \prod_{j=1}^m \mathe^{\mathi (f_j,\phi + C^{\Lambda,\epsilon} \ast h)} \exp\left[ \left( h, \phi + C^{\Lambda,\epsilon} \ast h \right) - \frac{1}{2} \left( h,C^{\Lambda,\epsilon} \ast h \right) - (\phi,h) \right] \total \mu^{\Lambda,\epsilon}(\phi) \\
&= \int \prod_{j=1}^m \mathe^{\mathi (f_j,\phi + C^{\Lambda,\epsilon} \ast h)} \exp\left[ \frac{1}{2} \left( h, C^{\Lambda,\epsilon} \ast h \right) \right] \total \mu^{\Lambda,\epsilon}(\phi) \\
&= \expect{ \prod_{j=1}^m \mathe^{\mathi (f_j,\phi)} }^{\Lambda,\epsilon}_0 \exp\left[ \mathi \sum_{j=1}^m ( f_j, C^{\Lambda,\epsilon} \ast h ) + \frac{1}{2} \left( h, C^{\Lambda,\epsilon} \ast h \right) \right] \eqend{,}
\end{splitequation}
where we used the formula~\eqref{eq:gaussian_measure_shift} in the third equality. We then set
\begin{equation}
h = \sum_{k=1}^n \sigma_k g_k
\end{equation}
and use that
\begin{equation}
\prod_{k=1}^n (g_k,\phi) = \left. \prod_{k=1}^n \frac{\partial}{\partial \sigma_k} \mathe^{(h,\phi)} \right\rvert_{\sigma_i = 0}
\end{equation}
to obtain equation~\eqref{eq:correlator_exponential_nphi}. Equation~\eqref{eq:correlator_exponential_twophi} follows immediately.
\end{proof}

\subsection{Proof of theorem~\ref{thm_euclid_renorm} (Renormalisation)}
\label{sec_euclid_renorm}

We begin with $\op_{\mu\nu} = \partial_\mu \phi \, \partial_\nu \phi$. Taking two functional derivatives of equation~\eqref{eq:normal_ordering_exponential_general} with respect to $J$ and setting $J$ to zero, we obtain
\begin{equation}
\mathcal{N}_\mu\left[ \phi(x) \phi(y) \right] = \phi(x) \phi(y) - C^{\mu,\epsilon}(x,y) \eqend{,}
\end{equation}
and taking derivatives and setting $x = y = z$, it follows that
\begin{equation}
\label{eq:euclid_renorm_normordopmunu}
\mathcal{N}_\mu\left[ \op_{\mu\nu}(z) \right] = \op_{\mu\nu}(z) + \lim_{x \to z} \partial_\mu \partial_\nu C^{\mu,\epsilon}(x,z) = \op_{\mu\nu}(z) - \frac{1}{2 \pi \epsilon^2} \delta_{\mu\nu} \eqend{,}
\end{equation}
using the explicit form of the covariance~\eqref{eq:covariance}. The normal-ordering of the vertex operators $V_\alpha(x) = \mathe^{\mathi \alpha \phi(x)}$ is given in equation~\eqref{eq:normal_ordering_exponential}, such that we obtain (with $\sigma_j = \pm 1$)
\begin{splitequation}
&\expect{ \mathcal{N}_\mu\left[ \op_{\mu\nu}(z) \right] \prod_{j=1}^n \mathcal{N}_\mu\left[ V_{\sigma_j \beta}(x_j) \right] }^{\Lambda,\epsilon}_0 \\
&= ( \mu \epsilon )^{- n \frac{\beta^2}{4 \pi}} \expect{ \left[ \op_{\mu\nu}(z) - \frac{1}{2 \pi \epsilon^2} \delta_{\mu\nu} \right] \prod_{j=1}^n V_{\sigma_j \beta}(x_j) }^{\Lambda,\epsilon}_0 \eqend{.}
\end{splitequation}
To compute the expectation value, we use Lemma~\ref{lemma_euclid_exponential_phi} and in particular equation~\eqref{eq:correlator_exponential_twophi}, where we take the limits $g_1(x) \to \partial_\mu \delta(z-x)$, $g_2(x) \to \partial_\nu \delta(z-x)$ and $f_j(x) \to \sigma_j \beta \delta(x-x_j)$. (These limits are well-defined for finite cutoffs since then the covariance $C^{\Lambda,\epsilon}$ is a smooth function.) This results in
\begin{splitequation}
&\expect{ \mathcal{N}_\mu\left[ \op_{\mu\nu}(z) \right] \prod_{j=1}^n \mathcal{N}_\mu\left[ V_{\sigma_j \beta}(x_j) \right] }^{\Lambda,\epsilon}_0 \\
&= - \beta^2 ( \mu \epsilon )^{- n \frac{\beta^2}{4 \pi}} \expect{ \prod_{j=1}^n V_{\sigma_j \beta}(x_j) }^{\Lambda,\epsilon}_0 \sum_{j,k=1}^n \sigma_j \sigma_k \partial_\mu C^{\Lambda,\epsilon}(x_j,z) \partial_\nu C^{\Lambda,\epsilon}(x_k,z) \eqend{,}
\end{splitequation}
where the term coming from the normal ordering~\eqref{eq:euclid_renorm_normordopmunu} has canceled with the first term in equation~\eqref{eq:correlator_exponential_twophi}. Since the covariance $C^{\Lambda,\epsilon}(x,y)$ is a function of $x-y$ only, it is no loss of generality to let the derivatives always act on the first argument, which we will do here and in the remainder of this work. Finally, using equation~\eqref{eq:gaussian_characteristic} with $J(x) = \sum_{j=1}^n \sigma_j \beta \delta(x-x_j)$ we obtain
\begin{splitequation}
\label{eq:euclid_renorm_expecteps}
&\expect{ \mathcal{N}_\mu\left[ \op_{\mu\nu}(z) \right] \prod_{j=1}^n \mathcal{N}_\mu\left[ V_{\sigma_j \beta}(x_j) \right] }^{\Lambda,\epsilon}_0 \\
&= - \beta^2 ( \mu \epsilon )^{- n \frac{\beta^2}{4 \pi}} \exp\left[ - \frac{1}{2} \beta^2 \sum_{j,k=1}^n \sigma_j \sigma_k C^{\Lambda,\epsilon}(x_j,x_k) \right] \\
&\qquad\times \sum_{j,k=1}^n \sigma_j \sigma_k \partial_\mu C^{\Lambda,\epsilon}(x_j,z) \partial_\nu C^{\Lambda,\epsilon}(x_k,z) \\
&= - \frac{\beta^2}{4 \pi^2} \left( \frac{\Lambda}{\mu} \right)^{\frac{\beta^2}{4 \pi} \left( \sum_{j=1}^n \sigma_j \right)^2} \prod_{1 \leq j < k \leq n} \left[ \mu^2 \left[ (x_j-x_k)^2 + \epsilon^2 \right] \right]^{\sigma_j \sigma_k \frac{\beta^2}{4 \pi}} \\
&\qquad\times \sum_{j,k=1}^n \sigma_j \sigma_k \frac{(x_j-z)_\mu}{(x_j-z)^2 + \epsilon^2} \frac{(x_k-z)_\nu}{(x_k-z)^2 + \epsilon^2} \\
\end{splitequation}
with the explicit form of the covariance~\eqref{eq:covariance}.

The terms $\left[ \mu^2 \left[ (x_j-x_k)^2 + \epsilon^2 \right] \right]^{\sigma_j \sigma_k \frac{\beta^2}{4 \pi}}$ are singular in the limit $\epsilon \to 0$ if $\sigma_j \sigma_k = -1$, but the singularity is integrable since we are in the finite regime $\beta^2 < 4 \pi$. The same holds for the terms $(x_j-z)_\mu/[(x_j-z)^2 + \epsilon^2] (x_k-z)_\nu/[(x_k-z)^2 + \epsilon^2]$ if $j \neq k$, but for $j = k$ their scaling degree at $x_j = z$ is $2 = \operatorname{dim} \mathbb{R}^2$ in the limit $\epsilon \to 0$, such that we have a logarithmic singularity in this case and need to renormalise. We recall that the scaling degree of a distribution $u \in \mathcal{S}'(\mathbb{R}^k)$ at $x = 0$ is defined as
\begin{equation}
\label{eq:scaling_degree}
\operatorname{sd}(u) \equiv \inf\left\{ a \in \mathbb{R}\colon \lim_{\lambda \to 0} \lambda^a u(f_\lambda) = 0 \ \forall f \in \mathcal{S}(\mathbb{R}^k) \right\} \eqend{,}
\end{equation}
where $f_\lambda(x) \equiv f(\lambda x)$. On the other hand, if we consider $\ell > 2$ points at once, the total scaling degree in the limit $\epsilon \to 0$ is always less than $2 (\ell-1) = \operatorname{dim} (\mathbb{R}^2)^{\times (\ell-1)}$, such that no further renormalisation is necessary. Namely, considering $x_{k_1} \cdots x_{k_\ell}$, the scaling degree at $x_{k_1} = \cdots = x_{k_\ell}$ is given by
\begin{equation}
\label{eq:euclid_renorm_multiscalingdegree}
- 2 \sum_{1 \leq i < j \leq \ell} \sigma_{k_i} \sigma_{k_j} \frac{\beta^2}{4 \pi} = \frac{\beta^2}{4 \pi} \left[ \ell - \left( \sum_{i=1}^\ell \sigma_{k_i} \right)^2 \right] \leq \frac{\beta^2}{4 \pi} \ell < 2 (\ell-1) \eqend{,}
\end{equation}
where we used that $\sigma_i = \pm 1$ and that $\beta^2 < 4 \pi$. If $z$ is also included such that we consider $\ell+1$ points at once, the analogous computation shows that the scaling degree is at most $2 + \frac{\beta^2}{4 \pi} \ell < 2 + 2 (\ell-1) = 2 \ell$, and the same conclusion holds.

\begin{lemma}
\label{lemma_euclid_distribution}
Consider the family of distributions $u_{\rho\sigma}^\epsilon$ defined for $\epsilon > 0$ by
\begin{equation}
\label{eq:umunu_def}
u_{\rho\sigma}^\epsilon(f) \equiv \int \frac{x_\rho x_\sigma}{( x^2 + \epsilon^2 )^2} f(x) \total^2 x \eqend{,}
\end{equation}
which are well-defined for $f \istest$ with $f(0) = 0$ also in the limit $\epsilon \to 0$. If $f(0) \neq 0$, it holds that
\begin{equation}
\label{eq:umunu_decomposition}
\lim_{\epsilon \to 0} \left[ u_{\rho\sigma}^\epsilon(f) - u^\text{div}_{\rho\sigma}(f) - u^\text{ren}_{\rho\sigma}(f) \right] = 0 \eqend{,}
\end{equation}
where the divergent part $u_{\rho\sigma}^\text{div}$ is given by
\begin{equation}
\label{eq:umunu_div}
u_{\rho\sigma}^\text{div}(f) = - \pi \delta_{\rho\sigma} \ln(\mu \epsilon) f(0) \eqend{,}
\end{equation}
and the renormalised part $u_{\rho\sigma}^\text{ren}$ reads
\begin{equation}
\label{eq:umunu_ren}
u_{\rho\sigma}^\text{ren}(f) = - \frac{1}{4} \int \ln\left( \mu^2 x^2 \right) \left[ \partial_\rho \partial_\sigma f(x) + \delta_{\rho\sigma} \frac{x^\alpha}{x^2} \partial_\alpha f(x) \right] \total^2 x \eqend{.}
\end{equation}
Requiring that $u_{\rho\sigma}^\text{ren}(f) = \lim_{\epsilon \to 0} u_{\rho\sigma}^\epsilon(f)$ for all $f \istest$ with $f(0) = 0$ and that it preserves Euclidean covariance and the scaling degree, $u_{\rho\sigma}^\text{ren}$ is unique up to the choice of the arbitrary scale $\mu$.
\end{lemma}
\begin{proof}
We first show that $u^\epsilon_{\rho\sigma}$ is a well-defined distribution for all $f \istest$ with $f(0) = 0$ for all $\epsilon$ including the limit. We compute for $\epsilon > 0$ that
\begin{equation}
\label{eq:umunu_logarithm_derivative}
\partial_\rho \partial_\sigma \ln\left[ \mu^2 \left( x^2 + \epsilon^2 \right) \right] = \frac{2 \delta_{\rho\sigma}}{x^2 + \epsilon^2} - \frac{4 x_\rho x_\sigma}{( x^2 + \epsilon^2 )^2} \eqend{,}
\end{equation}
such that
\begin{splitequation}
u^\epsilon_{\rho\sigma}(f) &= \int \left[ \frac{1}{2} \frac{\delta_{\rho\sigma}}{x^2 + \epsilon^2} - \frac{1}{4} \partial_\rho \partial_\sigma \ln\left[ \mu^2 \left( x^2 + \epsilon^2 \right) \right] \right] f(x) \total^2 x \\
&= \frac{1}{2} \delta_{\rho\sigma} \int \frac{1}{x^2 + \epsilon^2} f(x) \total^2 x - \frac{1}{4} \int \ln\left[ \mu^2 \left( x^2 + \epsilon^2 \right) \right] \partial_\rho \partial_\sigma f(x) \total^2 x \eqend{.}
\end{splitequation}
In the second term, we can take the limit $\epsilon \to 0$ for all $f$ since the singularity of the integrand at $x = 0$ is logarithmic and thus integrable. For the first term, we switch to polar coordinates and compute for $\epsilon > 0$
\begin{splitequation}
\int \frac{1}{x^2 + \epsilon^2} f(x) \total^2 x &= \int_0^{2\pi} \int_0^\infty \frac{r}{r^2 + \epsilon^2} f(r,\phi) \total r \total \phi \\
&= \frac{1}{2} \int_0^{2\pi} \int_0^\infty \partial_r \ln\left[ \mu^2 \left( r^2 + \epsilon^2 \right) \right] f(r,\phi) \total r \total \phi \\
&= - 2 \pi \ln(\mu \epsilon) f(0) - \frac{1}{2} \int_0^{2\pi} \int_0^\infty \ln\left[ \mu^2 \left( r^2 + \epsilon^2 \right) \right] \partial_r f(r,\phi) \total r \total \phi \\
&= - 2 \pi \ln(\mu \epsilon) f(0) - \frac{1}{2} \int \ln\left[ \mu^2 \left( x^2 + \epsilon^2 \right) \right] \frac{x^\alpha}{x^2} \partial_\alpha f(x) \total^2 x \eqend{,}
\end{splitequation}
where we used that $f$ vanishes faster than any polynomial as $r \to \infty$ and that $r \partial_r = x^\alpha \partial_\alpha$. Since the singularity of the second integrand at $x = 0$ is logarithmic and thus integrable, we can again take the limit $\epsilon \to 0$ there. It follows that $\lim_{\epsilon \to 0} u^\epsilon_{\rho\sigma}(f)$ is finite whenever $f(0) = 0$, and that for general $f$ the decomposition~\eqref{eq:umunu_decomposition} holds.

Consider now a different renormalisation. By construction, the difference can only be a distribution supported at $x = 0$. Enforcing Euclidean covariance and the preservation of the scaling degree, it must be proportional to $\delta_{\rho\sigma} f(0)$, which is parametrised by the choice of the scale $\mu$.
\end{proof}

The negative of $u_{\mu\nu}^\text{div}$ is the required counterterm, which is local and diverges logarithmically with the UV cutoff $\epsilon$ as required. To obtain the renormalised expectation value, we separate the terms with $j = k$ and $j \neq k$ in the last sum in the unrenormalised expectation value~\eqref{eq:euclid_renorm_expecteps}, and then replace $u_{\mu\nu}$ with $u_{\mu\nu}^\text{ren}$ in the terms with $j = k$. Taking the limit $\epsilon \to 0$, we obtain
\begin{splitequation}
\label{eq:euclid_renorm_expectren}
&\expect{ \mathcal{N}_\mu\left[ \op_{\mu\nu}(z) \right] \prod_{j=1}^n \mathcal{N}_\mu\left[ V_{\sigma_j \beta}(x_j) \right] }^{\Lambda,0}_{0,\text{ren}} \\
&= - \frac{\beta^2}{4 \pi^2} \left( \frac{\Lambda}{\mu} \right)^{\frac{\beta^2}{4 \pi} \left( \sum_{j=1}^n \sigma_j \right)^2} \prod_{1 \leq j < k \leq n} \left[ \mu^2 (x_j-x_k)^2 \right]^{\sigma_j \sigma_k \frac{\beta^2}{4 \pi}} \\
&\qquad\times \left[ \sum_{k=1}^n u_{\mu\nu}^\text{ren}(x_k-z) + 2 \sum_{1 \leq j < k \leq n} \sigma_j \sigma_k \frac{(x_j-z)_{(\mu}}{(x_j-z)^2} \frac{(x_k-z)_{\nu)}}{(x_k-z)^2} \right] \eqend{,}
\end{splitequation}
where $u_{\mu\nu}^\text{ren}(x)$ is the formal integral kernel of the distribution defined by equation~\eqref{eq:umunu_ren}. We see that taking the IR cutoff $\Lambda \to 0$, we have a non-vanishing result only if the sum of all $\sigma_i$ vanishes, which is the super-selection criterion or neutrality condition of the vacuum sector~\cite{swieca1977}.

For the renormalised expectation value of the stress tensor $T_{\mu\nu} = \op_{\mu\nu} - \frac{1}{2} \delta_{\mu\nu} \op_\rho{}^\rho + g \delta_{\mu\nu} ( V_\beta + V_{-\beta} )$, we obtain a sum of four terms, the first two of which are obtained from equation~\eqref{eq:euclid_renorm_expectren}. Since the divergent part of $u_{\mu\nu}$~\eqref{eq:umunu_div} is proportional to $\delta_{\mu\nu}$, it cancels out between the first two terms, i.e., we have
\begin{splitequation}
&\expect{ \mathcal{N}_\mu\left[ \op_{\mu\nu}(z) - \frac{1}{2} \delta_{\mu\nu} \op_\rho{}^\rho(z) \right] \prod_{j=1}^n \mathcal{N}_\mu\left[ V_{\sigma_j \beta}(x_j) \right] }^{\Lambda,0}_{0,\text{ren}} \\
&\quad= \lim_{\epsilon \to 0} \expect{ \mathcal{N}_\mu\left[ \op_{\mu\nu}(z) - \frac{1}{2} \delta_{\mu\nu} \op_\rho{}^\rho(z) \right] \prod_{j=1}^n \mathcal{N}_\mu\left[ V_{\sigma_j \beta}(x_j) \right] }^{\Lambda,\epsilon}_0
\end{splitequation}
and the limit exists also without subtracting counterterms. For the last two terms, we compute
\begin{splitequation}
\label{eq:euclid_renorm_expectvertex}
&\expect{ \mathcal{N}_\mu\left[ V_\beta(z) \right] \prod_{j=1}^n \mathcal{N}_\mu\left[ V_{\sigma_j \beta}(x_j) \right] }^{\Lambda,\epsilon}_0 = ( \mu \epsilon )^{- (n+1) \frac{\beta^2}{4 \pi}} \\
&\qquad\times \exp\left[ - \frac{\beta^2}{2} C^{\Lambda,\epsilon}(z,z) - \beta^2 \sum_{j=1}^n \sigma_j C^{\Lambda,\epsilon}(z,x_j) - \frac{\beta^2}{2} \sum_{j,k=1}^n \sigma_j \sigma_k C^{\Lambda,\epsilon}(x_j,x_k) \right] \\
&= \left( \frac{\Lambda}{\mu} \right)^{\frac{\beta^2}{4 \pi} \left( 1 + \sum_{j=1}^n \sigma_j \right)^2} \prod_{j=1}^n \left[ \mu^2 \left[ (z-x_j)^2 + \epsilon^2 \right] \right]^{\sigma_j \frac{\beta^2}{4 \pi}} \\
&\qquad\times \prod_{1 \leq j < k \leq n} \left[ \mu^2 \left[ (x_j-x_k)^2 + \epsilon^2 \right] \right]^{\sigma_j \sigma_k \frac{\beta^2}{4 \pi}} \eqend{,}
\end{splitequation}
using the normal-ordering of vertex operators~\eqref{eq:normal_ordering_exponential}, equation~\eqref{eq:gaussian_characteristic} with $J(x) = \beta \delta(x-z) + \sum_{j=1}^n \sigma_j \beta \delta(x-x_j)$, and the explicit form of the covariance~\eqref{eq:covariance}. Since we are in the finite regime $\beta^2 < 4 \pi$, the singularities that arise for $\epsilon = 0$ as $x_j \to x_k$ and $x_j \to z$ are integrable, and so for this term no further renormalisation beyond the normal-ordering is required. If we consider more than two points at once, the same argument as before using the scaling degree~\eqref{eq:euclid_renorm_multiscalingdegree} shows that also those singularities are integrable. Moreover, we again see how the neutrality condition appears: as $\Lambda \to 0$, we obtain a vanishing result unless $\sum_{j=1}^n \sigma_j = -1$. The last term with $V_{-\beta}$ then results in the same result with $\sigma_j$ replaced by $-\sigma_j$ on the right-hand side. \hfill\squareforqed

\subsection{Proof of theorem~\ref{thm_euclid_conv} (Convergence)}
\label{sec_euclid_conv}

In this whole section, we tacitly employ Fubini's theorem to interchange absolutely convergent integrals. We consider numerator and denominator of equation~\eqref{eq:thm_euclid_conv_series} separately. We start with the denominator, and compute analogously to equation~\eqref{eq:euclid_renorm_expectvertex} that
\begin{splitequation}
\label{eq:euclid_conv_expectvertex}
\expect{ \prod_{j=1}^n \mathcal{N}_\mu\left[ V_{\sigma_j \beta}(x_j) \right] }^{0,0}_{0,\text{ren}} &= \lim_{\Lambda,\epsilon \to 0} \expect{ \prod_{j=1}^n \mathcal{N}_\mu\left[ V_{\sigma_j \beta}(x_j) \right] }^{\Lambda,\epsilon}_0 \\
&= \delta_{0,\sum_{j=1}^n \sigma_j} \prod_{1 \leq j < k \leq n} \left[ \mu^2 (x_j-x_k)^2 \right]^{\sigma_j \sigma_k \frac{\beta^2}{4 \pi}} \eqend{.}
\end{splitequation}
Since $\sigma_j = \pm 1$, to obtain a non-vanishing result we must have $n = 2m$ with $m$ positive $\sigma_j$ and $m$ negative ones. We then rename the $x_j$ with $\sigma_j = -1$ to $y_j$ and renumber them. Taking into account that there are $\binom{n}{m} = (2m)!/(m!)^2$ possibilities to choose $m$ positive $\sigma_j$ from a total of $n = 2m$ ones (since equation~\eqref{eq:euclid_conv_expectvertex} is symmetric under a permutation of the (renamed) $x_i$ and $y_j$ among themselves), the denominator of equation~\eqref{eq:thm_euclid_conv_series} reduces to
\begin{splitequation}
\label{eq:euclid_conv_denom}
&\sum_{m=0}^\infty \frac{1}{(m!)^2} \int\dotsi\int \expect{ \prod_{j=1}^m \mathcal{N}_\mu\left[ V_{\beta}(x_j) \right] \, \mathcal{N}_\mu\left[ V_{-\beta}(y_j) \right] }^{0,0}_{0,\text{ren}} \prod_{i=1}^m g(x_i) g(y_i) \total^2 x_i \total^2 y_i \\
&= \sum_{m=0}^\infty \frac{\mu^{-m \frac{\beta^2}{2 \pi}}}{(m!)^2} \int\dotsi\int \left[ \frac{\prod_{1 \leq j < k \leq m} (x_j-x_k)^2 (y_j-y_k)^2}{\prod_{j,k=1}^m (x_j-y_k)^2} \right]^\frac{\beta^2}{4 \pi} \prod_{i=1}^m g(x_i) g(y_i) \total^2 x_i \total^2 y_i \eqend{.}
\end{splitequation}
As in previous works~\cite{deutschlavaud1974,froehlich1976}, to bound the term at order $2m$ we introduce complex variables $\chi_j \equiv x_j^1 + \mathi x_j^2$, $\upsilon_j \equiv y_j^1 + \mathi y_j^2$ such that
\begin{equation}
(x_i-y_j)^2 = (x_i^1 - y_j^1)^2 + (x_i^2 - y_j^2)^2 = \abs{ \chi_i - \upsilon_j }^2
\end{equation}
and analogously for $(x_i-x_j)^2 = \abs{ \chi_i - \chi_j }^2$ and $(y_i - y_j)^2 = \abs{ \upsilon_i - \upsilon_j }^2$. It follows that
\begin{splitequation}
\label{eq:cauchy_determinant}
&\left[ \frac{\prod_{1 \leq i < j \leq n} (x_i-x_j)^2 \prod_{1 \leq i < j \leq n} (y_i-y_j)^2}{\prod_{i,j=1}^n (x_i-y_j)^2} \right]^p \\
&= \abs{ \frac{\prod_{1 \leq i < j \leq n} (\chi_i-\chi_j) \prod_{1 \leq i < j \leq n} (\upsilon_i-\upsilon_j)}{\prod_{i,j=1}^n (\chi_i-\upsilon_j)} }^{2p} = \abs{ \det\left( \frac{1}{\chi_i - \upsilon_j} \right)_{i,j=1}^n }^{2 p} \eqend{,}
\end{splitequation}
where the last equality is the well-known Cauchy determinant formula~\cite{cauchy1841}. We estimate
\begin{equation}
\label{eq:cauchy_estimate1}
\abs{ \det\left( \frac{1}{\chi_i - \upsilon_j} \right)_{i,j=1}^n } \leq \sum_{\pi} \prod_{j=1}^n \frac{1}{\abs{\chi_j - \upsilon_{\pi(j)}}} \eqend{,}
\end{equation}
where the sum runs over all permutations $\pi$ of $\{ 1,\ldots,n \}$. Using the inequality
\begin{equation}
\label{eq:cauchy_estimate2}
\left( \sum_{j=1}^k \abs{a_k} \right)^p \leq \left.\begin{cases} k^{p-1} & p \geq 1 \\ 1 & 0 < p < 1 \end{cases}\right\} \sum_{j=1}^k \abs{a_k}^p \leq \max\left( k^{p-1}, 1 \right) \sum_{j=1}^k \abs{a_k}^p \eqend{,}
\end{equation}
it follows that
\begin{splitequation}
\label{eq:determinant_bound}
&\left[ \frac{\prod_{1 \leq i < j \leq n} (x_i-x_j)^2 \prod_{1 \leq i < j \leq n} (y_i-y_j)^2}{\prod_{i,j=1}^n (x_i-y_j)^2} \right]^p \leq \left[ \sum_{\pi} \prod_{j=1}^n \frac{1}{\abs{\chi_j - \upsilon_{\pi(j)}}} \right]^{2 p} \\
&\qquad\leq \max\left( (n!)^{2p-1} , 1 \right) \sum_{\pi} \prod_{j=1}^n \abs{\chi_j - \upsilon_{\pi(j)}}^{-2p} \\
&\qquad= (n!)^{\max(2p-1,0)} \sum_{\pi} \prod_{j=1}^n \left[ (x_j-y_{\pi(j)})^2 \right]^{-p} \eqend{,}
\end{splitequation}
since the number of permutations $\pi$ of $n$ elements is $n!$.

Inserting the estimate~\eqref{eq:determinant_bound} with $p = \beta^2/(4\pi)$ into equation~\eqref{eq:euclid_conv_denom}, the denominator of equation~\eqref{eq:thm_euclid_conv_series} is estimated by (we recall that $g \geq 0$)
\begin{equation}
\label{eq:euclid_conv_denomdetbound}
\sum_{m=0}^\infty \frac{1}{(m!)^\gamma} \int\dotsi\int \prod_{j=1}^m \left[ \mu^2 (x_j-y_j)^2 \right]^{-\frac{\beta^2}{4\pi}} \prod_{i=1}^m g(x_i) g(y_i) \total^2 x_i \total^2 y_i
\end{equation}
with $\gamma \equiv 1-\max(\beta^2/(2\pi)-1,0)$, and where an additional factor of $m!$ arose since all $m!$ permutations of the $y_j$ give the same contribution. We see that we only need to bound
\begin{equation}
\iint \left[ \mu^2 (x-y)^2 \right]^{-\frac{\beta^2}{4\pi}} g(x) g(y) \total^2 x \total^2 y \eqend{,}
\end{equation}
which we split in two parts: the region where $\mu^2 (x-y)^2 > 1$ and which we bound by $\iint g(x) g(y) \total^2 x \total^2 y = \norm{ g }_1^2$, and the region where $\mu^2 (x-y)^2 \leq 1$ and which we bound using Young's inequality. For this, we use it in the form~\cite{brascamplieb1976}
\begin{equation}
\label{eq:young3}
\abs{ \left( f, g \ast h \right) } = \abs{ \iint f(x) g(x-y) h(y) \total^2 x \total^2 y } \leq \norm{ f }_p \norm{ g }_q \norm{ h }_r \eqend{,}
\end{equation}
where $p,q,r \geq 1$ with $1/p + 1/q + 1/r = 2$. Taking
\begin{equation}
\label{eq:young_qpr_choice}
q = 1 + \frac{4 \pi - \beta^2}{8 \pi} \eqend{,} \quad p = r = 1 + \frac{4 \pi}{8 \pi - \beta^2} \eqend{,}
\end{equation}
the condition $p,q,r \geq 1$ is fulfilled since we are in the finite regime $\beta^2 < 4 \pi$, as well as $p,q,r < 2$. We obtain
\begin{equation}
\iint_{\mu \abs{x-y} \leq 1} \left[ \mu^2 (x-y)^2 \right]^{-\frac{\beta^2}{4\pi}} g(x) g(y) \total^2 x \total^2 y \leq \norm{ g }_p^2 \norm{ \Theta(1-\mu \abs{\blank}) (\mu\abs{\blank})^{-\frac{\beta^2}{2\pi}} }_q
\end{equation}
and
\begin{splitequation}
\label{eq:young_thetamu_norm}
\norm{ \Theta(1-\mu \abs{\blank}) (\mu\abs{\blank})^{-\frac{\beta^2}{2\pi}} }_q^q &= \int_{\mu\abs{x} \leq 1} (\mu\abs{x})^{-\frac{\beta^2}{2\pi} q} \total^2 x \\
&= \frac{32 \pi^3}{(4 \pi - \beta^2) (8 \pi - \beta^2)} \mu^{-2} \eqend{,}
\end{splitequation}
where the choice~\eqref{eq:young_qpr_choice} that we made for $q$ ensures that the integral is finite. Summing the contributions from both regions, we have shown that there exists a constant $K$ (depending on $\beta$ and $g$) such that
\begin{equation}
\label{eq:euclid_conv_bound_simplefactor}
\iint \left[ \mu^2 (x-y)^2 \right]^{-\frac{\beta^2}{4\pi}} g(x) g(y) \total^2 x \total^2 y \leq K \eqend{,}
\end{equation}
and hence the denominator~\eqref{eq:euclid_conv_denom} is bounded from above by
\begin{equation}
\label{eq:euclid_conv_denombound}
\sum_{m=0}^\infty K^m \left.\begin{cases} (m!)^{-1} & \beta^2 < 2 \pi \\ (m!)^{\frac{\beta^2}{2 \pi} - 2} & 2 \pi \leq \beta^2 < 4 \pi \end{cases} \right\} < \infty \eqend{.}
\end{equation}
Moreover, since each term in the sum~\eqref{eq:euclid_conv_denom} is positive due to the condition $g \geq 0$, it is bounded from below by the first term which is 1. Let us remark that the bounds~\eqref{eq:euclid_conv_denombound} are not new and known from~\cite{froehlich1976}; we have repeated their derivation to make the proof self-contained, and since very similar estimates are needed for the numerator.

Consider thus the numerator of equation~\eqref{eq:thm_euclid_conv_series}, where using the result~\eqref{eq:euclid_renorm_expectren} for the renormalised expectation values again only even terms with $n = 2m$ contribute. Taking into account the symmetry under the exchange of variables and renaming integration variables as in the case of the denominator, the numerator reads
\begin{splitequation}
\label{eq:euclid_conv_num}
&- \frac{\beta^2}{4 \pi^2} \sum_{m=0}^\infty \frac{\mu^{-m \frac{\beta^2}{2 \pi}}}{(m!)^2} \int f(z) \int\dotsi\int \left[ \frac{\prod_{1 \leq j < k \leq m} (x_j-x_k)^2 (y_j-y_k)^2}{\prod_{j,k=1}^m (x_j-y_k)^2} \right]^\frac{\beta^2}{4 \pi} \\
&\quad\times\bigg[ \sum_{k=1}^m u_{\mu\nu}^\text{ren}(x_k-z) + \sum_{k=1}^m u_{\mu\nu}^\text{ren}(y_k-z) + 2 \sum_{1 \leq j < k \leq m} \frac{(x_j-z)_{(\mu}}{(x_j-z)^2} \frac{(x_k-z)_{\nu)}}{(x_k-z)^2} \\
&\qquad\qquad- 2 \sum_{j,k=1}^m \frac{(x_j-z)_{(\mu}}{(x_j-z)^2} \frac{(y_k-z)_{\nu)}}{(y_k-z)^2} + 2 \sum_{1 \leq j < k \leq m} \frac{(y_j-z)_{(\mu}}{(y_j-z)^2} \frac{(y_k-z)_{\nu)}}{(y_k-z)^2} \bigg] \\
&\quad\times \prod_{i=1}^m g(x_i) g(y_i) \total^2 x_i \total^2 y_i \total^2 z \eqend{.}
\end{splitequation}
Since we can interchange $x_j$ and $y_j$ without changing the result, we see that there are three different types of terms: the ones involving the renormalised $u_{\mu\nu}^\text{ren}$, the ones involving a double sum and only $x_j$, and the ones with a double sum over both $x_j$ and $y_j$. We start with the last two types, which after the shift $x_j \to x_j + z$, $y_j \to y_j + z$ and with the estimate~\eqref{eq:determinant_bound} with $p = \beta^2/(4\pi)$ are bounded by
\begin{splitequation}
\label{eq:euclid_conv_numbound_sum2}
&\frac{\beta^2}{4 \pi^2} \sum_{m=0}^\infty \frac{1}{(m!)^{1+\gamma}} \int \abs{ f(z) } \int\dotsi\int \sum_{\pi} \prod_{j=1}^m \left[ \mu^2 (x_j-y_{\pi(j)})^2 \right]^{-\frac{\beta^2}{4 \pi}} \\
&\quad\times\bigg[ 4 \sum_{1 \leq j < k \leq m} \frac{1}{\abs{x_j} \abs{x_k}} + 2 \sum_{j,k=1}^m \frac{1}{\abs{x_j} \abs{y_k}} \bigg] \prod_{i=1}^m g(x_i+z) g(y_i+z) \total^2 x_i \total^2 y_i \total^2 z
\end{splitequation}
with the same $\gamma$ as before, defined after equation~\eqref{eq:euclid_conv_denomdetbound}. As in the case of the numerator, all $m!$ permutations of the $y_j$ give the same contribution since the second sum is invariant under permutations of the $y_k$, such that we obtain an additional factor of $m!$ and can take $\pi(j) = j$ for the following estimates. We then exchange the (finite) sums with the integration, and estimate the integrals over $x_j$ and $y_j$ for each $j$ separately. For this, we use again Young's inequality~\eqref{eq:young3}, separating contributions involving a $1/\abs{x_j}$ or $1/\abs{y_j}$ from the sums from the ones without. The contributions without such terms are estimated as before~\eqref{eq:euclid_conv_bound_simplefactor}, taking into account that the norm of $g$ with shifted argument $x+z$ is equal to the norm of $g$ because the integral that defines the norm is translation invariant. For the other contributions, we begin with the first sum, where a single term $1/\abs{x_j}$ may be present. In the region where $\mu \abs{x_j-y_j} > 1$, we then estimate
\begin{splitequation}
&\iint_{\mu \abs{x_j-y_j} > 1} \left[ \mu^2 (x_j-y_j)^2 \right]^{-\frac{\beta^2}{4 \pi}} \frac{1}{\abs{x_j}} g(x_j+z) g(y_j+z) \total^2 x_j \total^2 y_j \\
&\quad\leq \norm{ g }_1 \int \frac{1}{\abs{x_j}} g(x_j+z) \total^2 x_j
\end{splitequation}
using H{\"o}lder's inequality
\begin{equation}
\label{eq:hoelder}
\norm{ f g }_1 \leq \norm{ f }_{r/(r-1)} \norm{ g }_r
\end{equation}
with $r = 1$, and furthermore for $1 \leq p < 2$
\begin{splitequation}
\label{eq:g_xz_norm}
\int \left[ \frac{1}{\abs{x_j}} g(x_j+z) \right]^p \total^2 x_j &= \int_{\mu \abs{x_j} \leq 1} \frac{g^p(x_j+z)}{\abs{x_j}^p} \total^2 x_j + \int_{\mu \abs{x_j} > 1} \frac{g^p(x_j+z)}{\abs{x_j}^p} \total^2 x_j \\
&\leq 2 \pi \norm{ g^p }_\infty \int_0^\frac{1}{\mu} r^{1-p} \total r + \mu^p \int_{\mu \abs{x_j} > 1} g^p(x_j+z) \total^2 x_j \\
&\leq 2 \pi \norm{ g^p }_\infty \frac{\mu^{p-2}}{2-p} + \mu^p \norm{ g^p }_1 \eqend{,}
\end{splitequation}
which we employ with $p = 1$.

In the region where $\mu \abs{x_j-y_j} \leq 1$, we use instead Young's inequality~\eqref{eq:young3} with the same choice~\eqref{eq:young_qpr_choice} of exponents and obtain
\begin{splitequation}
&\iint_{\mu \abs{x_j-y_j} \leq 1} \left[ \mu^2 (x_j-y_j)^2 \right]^{-\frac{\beta^2}{4\pi}} \frac{1}{\abs{x_j}} g(x_j+z) g(y_j+z) \total^2 x_j \total^2 y_j \\
&\quad\leq \norm{ g }_p \norm{ \frac{1}{\abs{\blank}} g(\blank+z) }_p \norm{ \Theta(1-\mu \abs{\blank}) (\mu\abs{\blank})^{-\frac{\beta^2}{2\pi}} }_q \\
&\quad\leq \mu^{\frac{2}{p} - 3} \norm{ g }_p \left[ 2 \pi \norm{ g^p }_\infty \frac{8 \pi - \beta^2}{4 \pi - \beta^2} + \mu^2 \norm{ g^p }_1 \right]^\frac{1}{p} \left[ \frac{32 \pi^3}{(4 \pi - \beta^2) (8 \pi - \beta^2)} \right]^\frac{1}{q} \eqend{,}
\end{splitequation}
where we used the bounds~\eqref{eq:g_xz_norm} and~\eqref{eq:young_thetamu_norm} and our choice of exponents~\eqref{eq:young_qpr_choice}. Thus also in this case there exists a constant $K$ (depending on $\beta$ and $g$) such that
\begin{equation}
\iint \left[ \mu^2 (x_j-y_j)^2 \right]^{-\frac{\beta^2}{4 \pi}} \frac{1}{\abs{x_j}} g(x_j+z) g(y_j+z) \total^2 x_j \total^2 y_j \leq K \eqend{.}
\end{equation}
If the second sum involves $\frac{1}{\abs{x_j} \abs{y_k}}$ with $j \neq k$, we have the same estimates, while for $j = k$ we obtain analogously to the above, using equation~\eqref{eq:g_xz_norm} with $p = 1$,
\begin{splitequation}
&\iint_{\mu \abs{x_j-y_j} > 1} \left[ \mu^2 (x_j-y_j)^2 \right]^{-\frac{\beta^2}{4 \pi}} \frac{1}{\abs{x_j} \abs{y_j}} g(x_j+z) g(y_j+z) \total^2 x_j \total^2 y_j \\
&\quad\leq \left( 2 \pi \mu^{-1} \norm{ g }_\infty + \mu \norm{ g }_1 \right)^2
\end{splitequation}
and
\begin{splitequation}
&\iint_{\mu \abs{x_j-y_j} \leq 1} \left[ \mu^2 (x_j-y_j)^2 \right]^{-\frac{\beta^2}{4\pi}} \frac{1}{\abs{x_j} \abs{y_j}} g(x_j+z) g(y_j+z) \total^2 x_j \total^2 y_j \\
&\quad\leq \norm{ \frac{1}{\abs{\blank}} g(\blank+z) }_p^2 \norm{ \Theta(1-\mu \abs{\blank}) (\mu\abs{\blank})^{-\frac{\beta^2}{2\pi}} }_q \eqend{,}
\end{splitequation}
and inserting the bounds~\eqref{eq:g_xz_norm} and~\eqref{eq:young_thetamu_norm}, also in this case we can bound everything by a constant $K$ depending on $\beta$ and $g$. It follows that equation~\eqref{eq:euclid_conv_numbound_sum2} is bounded by
\begin{splitequation}
&\frac{\beta^2}{4 \pi^2} \sum_{m=0}^\infty \frac{1}{(m!)^\gamma} \int \abs{ f(z) } \left[ 4 \sum_{1 \leq j < k \leq m} K^m + 2 \sum_{j,k=1}^m K^m \right] \total^2 z \\
&\quad\leq C \sum_{m=0}^\infty \frac{1}{(m!)^\gamma} m^2 K^m < \infty
\end{splitequation}
with another constant $C = 4/\pi \norm{ f }_1$.

For the remaining terms in the numerator~\eqref{eq:euclid_conv_num} involving the renormalised $u_{\mu\nu}^\text{ren}$, we recall its definition~\eqref{eq:umunu_ren} from which it follows that
\begin{splitequation}
\label{eq:umunu_ren_shifted}
&\int u_{\mu\nu}^\text{ren}(x-z) f(z) \total^2 z \\
&\quad= - \frac{1}{4} \int \ln\left[ \mu^2 (x-z)^2 \right] \left[ \partial_\mu \partial_\nu f(z) - \delta_{\mu\nu} \frac{(x-z)^\rho}{(x-z)^2} \partial_\rho f(z) \right] \total^2 z \\
&\quad= \frac{1}{4} \int \left[ 2 \frac{(x-z)_\mu}{(x-z)^2} \partial_\nu f(z) + \delta_{\mu\nu} \ln\left[ \mu^2 (x-z)^2 \right] \frac{(x-z)^\rho}{(x-z)^2} \partial_\rho f(z) \right] \total^2 z \eqend{,}
\end{splitequation}
where we integrated the first term by parts. We then again shift $x_j \to x_j + z$, $y_j \to y_j + z$ and use the determinant estimate~\eqref{eq:determinant_bound} with $p = \beta^2/(4\pi)$ to obtain the bound
\begin{splitequation}
\label{eq:euclid_conv_numbound_sum1}
&\frac{\beta^2}{4 \pi^2} \sum_{m=0}^\infty \frac{1}{(m!)^{1+\gamma}} \int \int\dotsi\int \sum_{\pi} \prod_{j=1}^m \left[ \mu^2 (x_j-y_{\pi(j)})^2 \right]^{-\frac{\beta^2}{4 \pi}} \\
&\quad\times \sum_{k=1}^m \left[ \frac{2 + \abs{ \ln(\mu \abs{x_k}) }}{2 \abs{x_k}} \sup_{\rho \in \{1,2\}} \abs{ \partial_\rho f(z) } \right] \prod_{i=1}^m g(x_i+z) g(y_i+z) \total^2 x_i \total^2 y_i \total^2 z
\end{splitequation}
with the same $\gamma$ as before. The sum over permutations $\pi$ again gives a factor $m!$, and for the terms with the $x_k$ not involved in the sum, we have the same bounds~\eqref{eq:euclid_conv_bound_simplefactor} as before. For the other terms, in the region where $\mu \abs{x_k-y_k} > 1$ we estimate that
\begin{splitequation}
&\iint_{\mu \abs{x_k-y_k} > 1} \left[ \mu^2 (x_k-y_k)^2 \right]^{-\frac{\beta^2}{4 \pi}} \frac{2 + \abs{ \ln(\mu \abs{x_k}) }}{2 \abs{x_k}} g(x_k+z) g(y_k+z) \total^2 x_k \total^2 y_k \\
&\quad\leq \norm{ g }_1 \int \frac{2 + \abs{ \ln(\mu \abs{x_k}) }}{2 \abs{x_k}} g(x_k+z) \total^2 x_k \eqend{,}
\end{splitequation}
and then bound the last integral analogously to equation~\eqref{eq:g_xz_norm}. Using that $\ln x \leq \sqrt{x}$ for $x \geq 1$, this results in
\begin{splitequation}
&\int \left[ \frac{2 + \abs{ \ln(\mu \abs{x_k}) }}{2 \abs{x_k}} g(x_k+z) \right]^p \total^2 x_k \\
&\quad\leq 2 \pi \norm{ g^p }_\infty \int_0^\frac{1}{\mu} \left[ 1 - \frac{\ln(\mu r)}{2} \right]^p r^{1-p} \total r + \left( \frac{3}{2} \mu \right)^p \int_{\mu \abs{x_j} > 1} g^p(x_k+z) \total^2 x_k \\
&\quad\leq 5 \pi \norm{ g^p }_\infty \frac{\mu^{p-2}}{(2-p)^3} + \left( \frac{3}{2} \mu \right)^p \norm{ g^p }_1 \eqend{,}
\end{splitequation}
which we employ with $p = 1$.

In the region where $\mu \abs{x_k-y_k} \leq 1$, we use again Young's inequality~\eqref{eq:young3} with the same choice of exponents~\eqref{eq:young_qpr_choice} and obtain
\begin{splitequation}
&\iint_{\mu \abs{x_k-y_k} \leq 1} \left[ \mu^2 (x_k-y_k)^2 \right]^{-\frac{\beta^2}{4\pi}} \abs{ \ln(\mu \abs{x_k}) } g(x_k+z) g(y_k+z) \total^2 x_k \total^2 y_k \\
&\quad\leq \norm{ g }_p \norm{ \abs{ \ln(\mu \abs{\blank}) } g(\blank+z) }_p \norm{ \Theta(1-\mu \abs{\blank}) (\mu\abs{\blank})^{-\frac{\beta^2}{2\pi}} }_q
\end{splitequation}
and, using again H{\"o}lder's inequality~\eqref{eq:hoelder} with $r = 1$,
\begin{splitequation}
\norm{ \abs{ \ln(\mu \abs{\blank}) } g(\blank+z) }_p^p &= \int \frac{\abs{ \ln(\mu \abs{x}) }^p}{(\mu^2 x^2+1)^2} g^p(x+z) (\mu^2 x^2+1)^2 \total^2 x \\
&\leq 6 \mu^{-2} \norm{ (\mu^2 \blank^2+1)^2 g^p(\blank+z) }_\infty < \infty \eqend{,}
\end{splitequation}
and analogously
\begin{splitequation}
&\iint_{\mu \abs{x_k-y_k} \leq 1} \left[ \mu^2 (x_k-y_k)^2 \right]^{-\frac{\beta^2}{4\pi}} \frac{\abs{ \ln(\mu \abs{x_k}) }}{\abs{x_k}} g(x_k+z) g(y_k+z) \total^2 x_k \total^2 y_k \\
&\quad\leq \norm{ g }_p \norm{ \frac{\abs{ \ln(\mu \abs{\blank}) }}{\abs{\blank}} g(\blank+z) }_p \norm{ \Theta(1-\mu \abs{\blank}) (\mu\abs{\blank})^{-\frac{\beta^2}{2\pi}} }_q
\end{splitequation}
with
\begin{splitequation}
\norm{ \frac{\abs{ \ln(\mu \abs{\blank}) }}{\abs{\blank}} g(\blank+z) }_p^p &= \int \frac{\abs{ \ln(\mu \abs{x}) }^p}{\abs{x} (\mu^2 x^2+1)^2} g^p(x+z) (\mu^2 x^2+1)^2 \total^2 x \\
&\leq 25 \mu^{-2} \norm{ (\mu^2 \blank^2+1)^2 g^p(\blank+z) }_\infty < \infty \eqend{.}
\end{splitequation}
It follows that there exists a constant $K$ (depending on $\beta$ and $g$) such that equation~\eqref{eq:euclid_conv_numbound_sum1} is bounded by
\begin{equation}
C \sum_{m=0}^\infty \frac{1}{(m!)^\gamma} m K^m < \infty
\end{equation}
with another constant $C = \frac{1}{\pi} \Big[ \norm{ \sup_{\mu,\nu \in \{1,2\}} \abs{ \partial_\mu \partial_\nu f(\blank)} }_1 + \norm{ \sup_{\rho \in \{1,2\}} \abs{ \partial_\rho f(\blank) } }_1 \Big]$.

Taking all together, the numerator of equation~\eqref{eq:thm_euclid_conv_series} is bounded by~\eqref{eq:thm_euclid_conv_bound}, with $K$ being the maximum of all the constants $K$ in this section, and $C$ being the sum of all the constants $C$ in this section. Since the denominator is bounded from below by $1$, the bound~\eqref{eq:thm_euclid_conv_bound} holds for the full Gell-Mann--Low expectation value~\eqref{eq:thm_euclid_conv_series}.

For the first two terms of the stress tensor $T_{\mu\nu}(z) = \op_{\mu\nu}(z) - \frac{1}{2} \delta_{\mu\nu} \op_\rho{}^\rho(z) + g(z) \delta_{\mu\nu} ( V_\beta(z) + V_{-\beta}(z) )$ we can take over the above bounds. For the third term, we use the result~\eqref{eq:euclid_renorm_expectvertex} and thus have to bound
\begin{splitequation}
\label{eq:euclid_conv_num_vertex}
&\sum_{n=0}^\infty \frac{1}{n!} \int\dotsi\int \sum_{\sigma_i = \pm 1} \expect{ \mathcal{N}_\mu\left[ V_\beta(g f) \right] \prod_{j=1}^n \mathcal{N}_\mu\left[ V_{\sigma_j \beta}(x_j) \right] }^{0,0}_{0,\text{ren}} \prod_{i=1}^n g(x_i) \total^2 x_i \\
&= \sum_{m=0}^\infty \frac{\mu^{-(m+1) \frac{\beta^2}{2 \pi}}}{m! (m+1)!} \int g(z) f(z) \int\dotsi\int \prod_{j=1}^m \left[ \frac{(z-x_j)^2}{(z-y_j)^2} \right]^\frac{\beta^2}{4 \pi} \left[ \frac{1}{(z-y_{m+1})^2} \right]^\frac{\beta^2}{4 \pi} \\
&\quad\times \left[ \frac{\prod_{1 \leq j < k \leq m} (x_j-x_k)^2 (y_j-y_k)^2}{\prod_{j,k=1}^m (x_j-y_k)^2} \prod_{j=1}^m \frac{(y_j-y_{m+1})^2}{(x_j-y_{m+1})^2} \right]^\frac{\beta^2}{4 \pi} \\
&\quad\times \prod_{i=1}^m g(x_i) \total^2 x_i \prod_{j=1}^{m+1} g(y_j) \total^2 y_j \total^2 z \eqend{,}
\end{splitequation}
where we used that because of the neutrality condition only odd terms $n = 2m+1$ give a non-vanishing contribution. Of these, $m$ have a positive $\sigma_j$ and $m+1$ have a negative one, such that the sum over the $\sigma_j$ resulted in a factor of $\binom{2m+1}{m} = (2m+1)!/(m!(m+1)!)$, and as before we renamed the integration variables with a negative $\sigma_j$ to $y_j$. Setting $x_{m+1} \equiv z$, the terms in brackets combine to the expression~\eqref{eq:cauchy_determinant} with $n = m+1$ and $p = \beta^2/(4\pi)$, and we can use the Cauchy determinant formula and the bound~\eqref{eq:determinant_bound} for the determinant. We can thus bound the series~\eqref{eq:euclid_conv_num_vertex} by
\begin{splitequation}
&\sum_{m=0}^\infty \frac{1}{m! [(m+1)!]^\gamma} \int g(z) \abs{ f(z) } \int\dotsi\int \sum_{\pi} \left[ \mu^2 (z-y_{\pi(m+1)})^2 \right]^{-\frac{\beta^2}{4 \pi}} \\
&\quad\times \prod_{j=1}^m \left[ \mu^2 (x_j-y_{\pi(j)})^2 \right]^{-\frac{\beta^2}{4 \pi}} \prod_{i=1}^m g(x_i) \total^2 x_i \prod_{j=1}^{m+1} g(y_j) \total^2 y_j \total^2 z \eqend{,}
\end{splitequation}
and the sum over permutations $\pi$ gives a factor $(m+1)!$ since the remainder of the integrand is symmetric under the interchange of the $y_j$. We can now use the bound~\eqref{eq:euclid_conv_bound_simplefactor} (estimating together the integrals over $z$ and $y_{m+1}$), and obtain for the series~\eqref{eq:euclid_conv_num_vertex} the bound
\begin{equation}
C \sum_{m=0}^\infty \frac{1}{[(m+1)!]^\gamma} (m+1) K^m \leq C \sum_{m=0}^\infty \frac{1}{(m!)^\gamma} m^2 K^m < \infty \eqend{,}
\end{equation}
with $C$ now also depending on $g$ (from the integral over $y_{m+1}$ and $z$). The same bound is obtained analogously for the fourth term in the stress tensor involving $V_{-\beta}$, which switches $x_i$ with $y_i$ in equation~\eqref{eq:euclid_conv_num_vertex}, such that taking all together the bound~\eqref{eq:thm_euclid_conv_bound} holds also for the stress tensor. \hfill\squareforqed

\subsection{Proof of theorem~\ref{thm_euclid_cons} (Conservation)}
\label{sec_euclid_cons}

Since we have shown in the last subsection that the denominator of the Gell-Mann--Low formula~\eqref{eq:thm_euclid_conv_series} is bounded from below, to show that the interacting stress tensor is conserved it is enough to show that the numerator vanishes when smeared with a test function of the form $\partial^\mu f$ with $f \istest$. Consider the numerator for $\op_{\mu\nu}$~\eqref{eq:euclid_conv_num} smeared with $\partial^\mu f$, which contains two different types of terms: the ones with the renormalised $u_{\mu\nu}^\text{ren}$, and the others involving double sums. We start with the second type, and compute
\begin{splitequation}
\label{eq:euclid_cons_twopoint}
&\int \frac{(x-z)_{(\mu}}{(x-z)^2} \frac{(y-z)_{\nu)}}{(y-z)^2} \partial^\mu f(z) \total^2 z \\
&\quad= \frac{1}{4} \partial_{(\mu}^x \partial_{\nu)}^y \int \ln\left[ \mu^2 (x-z)^2 \right] \ln\left[ \mu^2 (y-z)^2 \right] \partial^\mu f(z) \total^2 z \\
&\quad= \frac{1}{8} \laplace_x \partial_\nu^y \int \ln\left[ \mu^2 (x-z)^2 \right] \ln\left[ \mu^2 (y-z)^2 \right] f(z) \total^2 z \\
&\qquad\quad+ \frac{1}{8} \laplace_y \partial_\nu^x \int \ln\left[ \mu^2 (x-z)^2 \right] \ln\left[ \mu^2 (y-z)^2 \right] f(z) \total^2 z \\
&\qquad\quad+ \frac{1}{8} \partial_\mu^x \partial^\mu_y \left( \partial_\nu^x + \partial_\nu^y \right) \int \ln\left[ \mu^2 (x-z)^2 \right] \ln\left[ \mu^2 (y-z)^2 \right] f(z) \total^2 z \\
&\quad= \frac{\pi}{2} \partial_\nu^y \ln\left[ \mu^2 (y-x)^2 \right] f(x) + \frac{\pi}{2} \partial_\nu^x \ln\left[ \mu^2 (x-y)^2 \right] f(y) \\
&\qquad\quad+ \frac{1}{8} \partial_\mu^x \partial^\mu_y \left( \partial_\nu^x + \partial_\nu^y \right) \int \ln\left[ \mu^2 (x-z)^2 \right] \ln\left[ \mu^2 (y-z)^2 \right] f(z) \total^2 z \eqend{,}
\end{splitequation}
using that the logarithm is a fundamental solution of the Laplace equation in 2 dimensions:
\begin{equation}
\laplace \ln\left[ \mu^2 (x-z)^2 \right] = 4 \pi \delta(x-z) \eqend{,}
\end{equation}
which is well-known, but also follows from the limit $m \to 0$ of equation~\eqref{eq:massive_covariance} for the massive covariance after taking derivatives. Analogously, we obtain
\begin{splitequation}
&\int \frac{(x-z)_\rho}{(x-z)^2} \frac{(y-z)^\rho}{(y-z)^2} \partial_\nu f(z) \total^2 z \\
&\quad= \frac{1}{4} \partial_\rho^x \partial^\rho_y \left( \partial_\nu^x + \partial_\nu^y \right) \int \ln\left[ \mu^2 (x-z)^2 \right] \ln\left[ \mu^2 (y-z)^2 \right] f(z) \total^2 z \eqend{,}
\end{splitequation}
which cancels the last line of~\eqref{eq:euclid_cons_twopoint} if we take the combination $\op_{\mu\nu} - \frac{1}{2} \delta_{\mu\nu} \op_\rho{}^\rho$ that appears in the stress tensor. The terms of the second type in this combination thus sum up to
\begin{splitequation}
\label{eq:euclid_cons_twopoint_2}
&\frac{\beta^2}{4 \pi} \sum_{m=0}^\infty \frac{\mu^{-m \frac{\beta^2}{2 \pi}}}{(m!)^2} \int\dotsi\int \left[ \frac{\prod_{1 \leq j < k \leq m} (x_j-x_k)^2 (y_j-y_k)^2}{\prod_{j,k=1}^m (x_j-y_k)^2} \right]^\frac{\beta^2}{4 \pi} \\
&\quad\times\bigg[ \sum_{1 \leq j < k \leq m} [ f(x_j) - f(x_k) ] \partial_\nu \ln\left[ \mu^2 (x_j-x_k)^2 \right] \\
&\qquad\quad- \sum_{j,k=1}^m [ f(x_j) - f(y_k) ] \partial_\nu \ln\left[ \mu^2 (x_j-y_k)^2 \right] \\
&\qquad\quad+ \sum_{1 \leq j < k \leq m} [ f(y_j) - f(y_k) ] \partial_\nu \ln\left[ \mu^2 (y_j-y_k)^2 \right] \bigg] \prod_{i=1}^m g(x_i) g(y_i) \total^2 x_i \total^2 y_i \eqend{.}
\end{splitequation}
We can further simplify this expression by collecting the various derivatives such that they act on the product in the first line. To show this, we compute
\begin{splitequation}
\label{eq:euclid_cons_det_derivative}
&\partial_\nu^{x_\ell} \ln \left( \left[ \frac{\prod_{1 \leq j < k \leq m} (x_j-x_k)^2 (y_j-y_k)^2}{\prod_{j,k=1}^m (x_j-y_k)^2} \right]^\frac{\beta^2}{4 \pi} \right) \\
&\quad= \frac{\beta^2}{2 \pi} \left[ \sum_{k=1}^{\ell-1} \frac{(x_\ell-x_k)_\nu}{(x_\ell-x_k)^2} + \sum_{k=\ell+1}^m \frac{(x_\ell-x_k)_\nu}{(x_\ell-x_k)^2} - \sum_{k=1}^m \frac{(x_\ell-y_k)_\nu}{(x_\ell-y_k)^2} \right] \eqend{,}
\end{splitequation}
and the analogous equation with $x$ and $y$ exchanged, multiply by $f(x_\ell)$, sum over $\ell$ and rename summation indices to obtain
\begin{splitequation}
&\sum_{k=1}^m f(x_k) \partial_\nu^{x_k} \ln \left( \left[ \frac{\prod_{1 \leq j < k \leq m} (x_j-x_k)^2 (y_j-y_k)^2}{\prod_{j,k=1}^m (x_j-y_k)^2} \right]^\frac{\beta^2}{4 \pi} \right) \\
&\quad= \frac{\beta^2}{2 \pi} \left[ \sum_{1 \leq j < k \leq m} [ f(x_j) - f(x_k) ] \frac{(x_j-x_k)_\nu}{(x_j-x_k)^2} - \sum_{j,k=1}^m f(x_j) \frac{(x_j-y_k)_\nu}{(x_j-y_k)^2} \right] \\
&\quad= \frac{\beta^2}{4 \pi} \sum_{1 \leq j < k \leq m} [ f(x_j) - f(x_k) ] \partial^{x_j}_\nu \ln\left[ \mu^2 (x_j-x_k)^2 \right] \\
&\qquad- \frac{\beta^2}{4 \pi} \sum_{j,k=1}^m f(x_j) \partial^{x_j}_\nu \ln\left[ \mu^2 (x_j-y_k)^2 \right] \eqend{,}
\end{splitequation}
as well as the analogous equation with $x$ and $y$ exchanged. We can thus rewrite equation~\eqref{eq:euclid_cons_twopoint_2} in the form
\begin{splitequation}
&\sum_{m=0}^\infty \frac{\mu^{-m \frac{\beta^2}{2 \pi}}}{(m!)^2} \int\dotsi\int \prod_{i=1}^m \total^2 x_i \total^2 y_i \, g(x_i) g(y_i) \\
&\quad\times \sum_{k=1}^m \left[ f(x_k) \partial_\nu^{x_k} + f(y_k) \partial_\nu^{y_k} \right] \left[ \frac{\prod_{1 \leq j < k \leq m} (x_j-x_k)^2 (y_j-y_k)^2}{\prod_{j,k=1}^m (x_j-y_k)^2} \right]^\frac{\beta^2}{4 \pi} \eqend{,}
\end{splitequation}
and since by assumption $g$ is constant on the support of $f$, we can integrate the derivatives by parts such that they act on $f$ and then use the symmetry of the integrand under the exchange of the $x_i$ and the $y_i$, such that the sum over $k$ gives a factor $m$. Since the term with $m = 0$ does not contribute, renaming the summation index $m \to m+1$ we thus obtain
\begin{splitequation}
\label{eq:euclid_cons_twopoint_3}
&- \sum_{m=0}^\infty \frac{\mu^{-(m+1) \frac{\beta^2}{2 \pi}}}{m! (m+1)!} \int\dotsi\int \left[ \frac{\prod_{1 \leq j < k \leq m+1} (x_j-x_k)^2 (y_j-y_k)^2}{\prod_{j,k=1}^{m+1} (x_j-y_k)^2} \right]^\frac{\beta^2}{4 \pi} \\
&\quad\times \left[ \partial_\nu f(x_{m+1}) + \partial_\nu f(y_{m+1}) \right] \prod_{i=1}^{m+1} g(x_i) g(y_i) \total^2 x_i \total^2 y_i
\end{splitequation}
for the terms of the second type in the combination $\op_{\mu\nu} - \frac{1}{2} \delta_{\mu\nu} \op_\rho{}^\rho$. This now cancels exactly the contribution from the last two terms of the stress tensor $T_{\mu\nu} = \op_{\mu\nu} - \frac{1}{2} \delta_{\mu\nu} \op_\rho{}^\rho + g \delta_{\mu\nu} ( V_\beta + V_{-\beta} )$ smeared with $\partial^\mu f$: this is seen by comparing equation~\eqref{eq:euclid_cons_twopoint_3} with the result~\eqref{eq:euclid_conv_num_vertex} (renaming $z = x_{m+1}$ in that equation), and the analogous result for $V_{-\beta}$ which is obtained by exchanging $x$ and $y$.

It follows that the numerator of the Gell-Mann--Low formula for $T_{\mu\nu}$ smeared with $\partial^\mu f$ only involves the terms with the renormalised $u_{\mu\nu}^\text{ren}$. For them, we use equation~\eqref{eq:umunu_ren_shifted} and compute
\begin{splitequation}
&\int u_{\mu\nu}^\text{ren}(x-z) \partial^\mu f(z) \total^2 z \\
&\quad= - \frac{1}{4} \int \ln\left[ \mu^2 (x-z)^2 \right] \left[ \laplace \partial_\nu f(z) - \frac{(x-z)^\rho}{(x-z)^2} \partial_\rho \partial_\nu f(z) \right] \total^2 z \\
&\quad= - \pi \partial_\nu f(x) + \frac{1}{4} \int \ln\left[ \mu^2 (x-z)^2 \right] \frac{(x-z)^\rho}{(x-z)^2} \partial_\rho \partial_\nu f(z) \total^2 z
\end{splitequation}
as well as
\begin{splitequation}
&\int u_\rho{}^\rho{}^\text{ren}(x-z) \partial_\nu f(z) \total^2 z \\
&\quad= - \frac{1}{4} \int \ln\left[ \mu^2 (x-z)^2 \right] \left[ \laplace \partial_\nu f(z) - 2 \frac{(x-z)^\rho}{(x-z)^2} \partial_\rho \partial_\nu f(z) \right] \total^2 z \\
&\quad= - \pi \partial_\nu f(x) + \frac{1}{2} \int \ln\left[ \mu^2 (x-z)^2 \right] \frac{(x-z)^\rho}{(x-z)^2} \partial_\rho \partial_\nu f(z) \total^2 z \eqend{,}
\end{splitequation}
using again that the logarithm is a fundamental solution of the Laplace equation in two dimensions. In the combination $u_{\mu\nu}^\text{ren} - \frac{1}{2} \delta_{\mu\nu} u_\rho{}^\rho{}^\text{ren}$ the integrals again cancel and we are left with the first terms:
\begin{splitequation}
&\sum_{n=0}^\infty \frac{1}{n!} \int\dotsi\int \sum_{\sigma_i = \pm 1} \expect{ \mathcal{N}_\mu\left[ T_{\mu\nu}(\partial^\mu f) \right] \prod_{j=1}^n \mathcal{N}_\mu\left[ V_{\sigma_j \beta}(x_j) \right] }^{0,0}_{0,\text{ren}} \prod_{i=1}^n g(x_i) \total^2 x_i \\
&\quad= \frac{\beta^2}{8 \pi} \sum_{m=0}^\infty \frac{\mu^{-m \frac{\beta^2}{2 \pi}}}{(m!)^2} \int\dotsi\int \left[ \frac{\prod_{1 \leq j < k \leq m} (x_j-x_k)^2 (y_j-y_k)^2}{\prod_{j,k=1}^m (x_j-y_k)^2} \right]^\frac{\beta^2}{4 \pi} \\
&\qquad\times \sum_{k=1}^m \bigg[ \partial_\nu f(x_k) + \partial_\nu f(y_k) \bigg] \prod_{i=1}^m g(x_i) g(y_i) \total^2 x_i \total^2 y_i \eqend{.}
\end{splitequation}
Comparing with the result~\eqref{eq:euclid_conv_num_vertex} with $z = x_{m+1}$, and the analogous result for $V_{-\beta}$ which is obtained by exchanging $x$ and $y$, it follows that
\begin{splitequation}
\sum_{n=0}^\infty \frac{1}{n!} \int\dotsi\int \sum_{\sigma_i = \pm 1} \expect{ \mathcal{N}_\mu\left[ \hat{T}_{\mu\nu}(\partial^\mu f) \right] \prod_{j=1}^n \mathcal{N}_\mu\left[ V_{\sigma_j \beta}(x_j) \right] }^{0,0}_{0,\text{ren}} \prod_{i=1}^n g(x_i) \total^2 x_i = 0
\end{splitequation}
and thus equation~\eqref{eq:thm_euclid_cons_eq}, where
\begin{splitequation}
\label{eq:euclid_cons_modified_stress}
\hat{T}_{\mu\nu} &\equiv T_{\mu\nu} - \frac{\beta^2}{8 \pi} g \delta_{\mu\nu} ( V_\beta + V_{-\beta} ) \\
&= \op_{\mu\nu} - \frac{1}{2} \delta_{\mu\nu} \op_\rho{}^\rho + g \left( 1 - \frac{\beta^2}{8 \pi} \right) \delta_{\mu\nu} ( V_\beta + V_{-\beta} )
\end{splitequation}
is the quantum-corrected stress tensor. \hfill\squareforqed

\section{Minkowski case}
\label{sec_mink}

In the Minkowski case, we use the framework of perturbative algebraic quantum field theory (pAQFT), whose main advantage is the clean separation of algebraic issues (including renormalisation) from the construction of a state. Reviews of pAQFT can be found in~\cite{hollandswald2015,fewsterverch2015,fredenhagenrejzner2016}, and we again take formulas and results without specifying their source explicitly.

\subsection{Preliminaries}
\label{sec_mink_pre}

In pAQFT, one first constructs an algebra of fields $\mathfrak{A}_0$ as the free algebra over smeared fields $\phi(f) = (f,\phi)$ with $f \istest$ and their adjoints $\left[ \phi(f) \right]^\dagger$, with unit $\1$ and the non-commutative product $\star$, quotiened by the commutation relations
\begin{equation}
\label{eq:commutation}
\left[ \phi(f), \phi(g) \right]_\star \equiv \phi(f) \star \phi(g) - \phi(g) \star \phi(f) = \mathi \left( f, \Delta \ast g \right) \1 \eqend{,}
\end{equation}
where the scalar product $(\cdot,\cdot)$ and convolution $\ast$ are defined in equations~\eqref{eq:scalar_product} and~\eqref{eq:convolution}. $\Delta$ is the commutator function defined as the difference between retarded and advanced fundamental solutions of the Klein--Gordon equation
\begin{equation}
\Delta(x,y) \equiv G_\text{ret}(x,y) - G_\text{adv}(x,y) \eqend{,} \qquad \partial^2 G_\text{ret}(x,y) = \delta(x-y) = \partial^2 G_\text{adv}(x,y) \eqend{,}
\end{equation}
which are unique in any globally hyperbolic spacetime, in particular Minkowski space. For the real scalar field that we are considering, one also takes the quotient by the relation $\left[ \phi(f) \right]^\dagger = \phi(f^*)$. A state $\omega$ is given by a linear functional on $\mathfrak{A}_0$, which is normalised $\omega(\1) = 1$ and positive: $\omega(A^\dagger A) > 0$ for all $0 \neq A \in \mathfrak{A}_0$. We consider quasi-free states with vanishing one-point function, which are the analogue of the centred Gaussian covariance in Euclidean signature. That is, these states are characterised by the analogue of equation~\eqref{eq:gaussian_measure}:
\begin{equation}
\omega\left( \phi(f_1) \star \cdots \star \phi(f_n) \right) = \sum_{\pi} \prod_{(i,j) \in \pi} \mathi \left( f_i, G^+ \ast f_j \right) \eqend{,}
\end{equation}
where the sum runs over all partitions $\pi$ of the set $\{ 1,\ldots,n \}$ into ordered pairs $(i,j)$, and where
\begin{equation}
G^+(x,y) \equiv - \mathi \omega\left( \phi(x) \star \phi(y) \right)
\end{equation}
is the two-point function of the state $\omega$, here written in terms of its integral kernel. Summing, we also obtain the analogue of equation~\eqref{eq:gaussian_characteristic} for exponentials:
\begin{equation}
\omega\left( \mathe_\star^{\mathi (J,\phi)} \right) = \omega\left( \sum_{k=0}^\infty \frac{\mathi^k}{k!} \, \underbrace{(J,\phi) \star \cdots \star (J,\phi)}_{k \text{ times}} \right) = \mathe^{- \frac{\mathi}{2} (J, G^+ \ast J)} \eqend{,}
\end{equation}
from which the above follows by functional differentiation with respect to $J$. We note that this relation should be understood as a formal power series in $J$, since the infinite sum that formally defines the exponential is not an element of $\mathfrak{A}_0$. Taking the expectation value of the commutation relations~\eqref{eq:commutation}, we obtain
\begin{equation}
\label{eq:commutation_g_delta}
G^+(x,y) - G^+(y,x) = \Delta(x,y) \eqend{,}
\end{equation}
such that the antisymmetric part of the two-point function is fixed.

For a massless scalar field in two-dimensional Minkowski space, we again have an IR divergence if we try to define the vacuum state as the limit $m \to 0$ of the massive one analogous to equation~\eqref{eq:massive_covariance} in Euclidean signature. Namely, the massive two-point function reads
\begin{splitequation}
G^+(x,0) &= - 2 \pi \mathi \int \mathe^{\mathi p x} \delta(p^2 + m^2) \Theta(p^0) \frac{\total^2 p}{(2\pi)^2} \\
&= - \mathi \lim_{\epsilon \to 0^+} \int \frac{\mathe^{- \mathi \sqrt{(p^1)^2+m^2} x^0 + \mathi p^1 x^1 - \epsilon \abs{p^1}}}{2 \sqrt{(p^1)^2+m^2}} \frac{\total p^1}{2 \pi} \eqend{,}
\end{splitequation}
and if $m = 0$ the $p^1$ integral has a logarithmic singularity at the origin. In the limit $m \to 0$, we obtain
\begin{equation}
G^+(x,0) = \frac{\mathi}{4 \pi} \ln\left[ \frac{m^2 \mathe^{2\gamma}}{4} ( \epsilon - \mathi (x^1-x^0) ) ( \epsilon + \mathi (x^1+x^0) ) \right] + \bigo{m} \eqend{,}
\end{equation}
and we see that for spacelike separations $(x^1)^2 > (x^0)^2$ where we can set $\epsilon = 0$, we recover exactly the result~\eqref{eq:massive_covariance} in Euclidean signature, up to an overall factor of $-\mathi$. However, the antisymmetric part has the well-defined limit
\begin{splitequation}
\Delta(x,0) &= G^+(x,0) - G^+(0,x) \\
&= \frac{\mathi}{4 \pi} \lim_{\epsilon \to 0^+} \Big[ \ln\left[ \epsilon - \mathi (x^1-x^0) \right] + \ln\left[ \epsilon + \mathi (x^1+x^0) \right] \\
&\qquad\qquad\quad- \ln\left[ \epsilon + \mathi (x^1-x^0) \right] - \ln\left[ \epsilon - \mathi (x^1+x^0) \right] \Big] \\
&= \frac{1}{4} \left[ \sgn(x^1-x^0) - \sgn(x^1+x^0) \right] = - \frac{1}{2} \Theta\left[ (x^0)^2 - (x^1)^2 \right] \sgn(x^0) \eqend{,}
\end{splitequation}
and vanishes outside the light cone as required. However, due to the IR divergence the massless limit of the massive two-point function is not positive definite (as one can also check explicitly), and so does not define a state. To construct the Hilbert space of the theory from the algebra $\mathfrak{A}_0$ and a state (via the GNS construction) the IR divergence must be cured, which can be done in different (related) ways:
\begin{itemize}
\item Working with a massive scalar field, and taking the limit $m \to 0$ only for expectation values of operators with a well-defined limit. This maintains positivity, and it is expected (and in some cases proven) that a finite mass arises from non-perturbative effects (\emph{Debye screening})~\cite{brydgesfederbush1980,yang1987,bauerschmidtwebb2020}.
\item The separation of the constant part of $\phi$ and its quantisation as a massless harmonic oscillator, similar to what is done in string theory~\cite{superstring1} and de Sitter QFT~\cite{kirstengarriga1993}. In pAQFT, this is the Derezi{\'n}ski--Meissner representation~\cite{derezinskimeissner2006}.
\item A Krein space construction, where positivity is only maintained in the physical subspace which contains derivatives of $\phi$ and vertex operators~\cite{nakanishi1976,morchiopierottistrocchi1990}.
\item Restricting the algebra $\mathfrak{A}_0$ to be generated by derivatives of $\phi$ and vertex operators, which we will do in the following. This is of course related to the Krein space construction, but works without explicitly introducing the unphysical subspace.
\end{itemize}
All three constructions should be equivalent for our purposes, since the interacting expectation values of $\op_{\mu\nu}$ and $T_{\mu\nu}$, defined again using the Gell-Mann--Low formula~\eqref{eq:thm_mink_conv_series} only involves derivatives of $\phi$ and vertex operators, which have a well-defined massless limit.

Instead of the vacuum state, we take moreover a general quasi-free Hadamard state, which for our purposes can be defined as the quasi-free state $\omega^{\Lambda,\epsilon}$ with two-point function
\begin{equation}
\label{eq:twopf}
G^+(x,y) = \frac{\mathi}{4 \pi} \ln\left[ \Lambda^2 ( \epsilon + \mathi u ) ( \epsilon + \mathi v ) \right] - \mathi W(x,y) \eqend{,}
\end{equation}
where $W(x,y)$ is a smooth and symmetric bisolution of the massless Klein--Gordon equation $\partial_x^2 W(x,y) = \partial_y^2 W(x,y) = 0$, and we introduced the light cone coordinates
\begin{equation}
\label{eq:lightconecoords}
u = u(x,y) \equiv (x^0-y^0) - (x^1-y^1) \eqend{,} \quad v = v(x,y) \equiv (x^0-y^0) + (x^1-y^1) \eqend{.}
\end{equation}
As in the Euclidean case, $\Lambda$ is an IR cutoff which we ultimately take to vanish, and we keep $\epsilon > 0$ as a UV cutoff. That is, the physical two-point function is obtained as the distributional boundary value (in the limit $\epsilon \to 0$) from the function~\eqref{eq:twopf} which is analytic for all $\epsilon > 0$. The two-point function~\eqref{eq:twopf} can be decomposed as
\begin{equation}
\label{eq:twopf_hadamard}
G^+(x,y) = H^+(x,y) + \frac{\mathi}{2 \pi} \ln\left( \frac{\Lambda}{\mu} \right) - \mathi W(x,y) \eqend{,}
\end{equation}
where $\mu$ is an arbitrary fixed scale (which for a massive theory could be a multiple of $m^2$) and
\begin{equation}
\label{eq:hadamard}
H^+(x,y) \equiv \frac{\mathi}{4 \pi} \ln\left[ \mu^2 ( \epsilon + \mathi u ) ( \epsilon + \mathi v ) \right]
\end{equation}
is the Hadamard parametrix containing the singular part of the two-point function, which is the same for all Hadamard states. The Feynman propagator and parametrix are the time-ordered versions of~\eqref{eq:twopf} and~\eqref{eq:hadamard}, and read
\begin{equations}[eq:feynman]
\begin{split}
G^\text{F}(x,y) &\equiv \Theta(x^0-y^0) G^+(x,y) + \Theta(y^0-x^0) G^+(y,x) \\
&= H^\text{F}(x,y) + \frac{\mathi}{2 \pi} \ln\left( \frac{\Lambda}{\mu} \right) - \mathi W(x,y) \eqend{,}
\end{split} \\
\begin{split}
H^\text{F}(x,y) &\equiv \Theta(x^0-y^0) H^+(x,y) + \Theta(y^0-x^0) H^+(y,x) \\
&= \frac{\mathi}{4 \pi} \ln\left[ \mu^2 ( - u v + \mathi \epsilon \abs{u + v} + \epsilon^2 ) \right] \eqend{,}
\end{split}
\end{equations}
where we used that $W(x,y)$ is symmetric. We note that the time-ordered Hadamard parametrix $H^\text{F}$ is a fundamental solution of the massless Klein--Gordon equation:
\begin{splitequation}
\label{eq:hf_fundamental_solution}
\partial^2 H^\text{F}(x,y) &= - 4 \partial_u \partial_v H^\text{F}(x,y) = \frac{2 \epsilon}{\pi ( u^2 + \epsilon^2 )} \delta(u+v) \\
&\to 2 \delta(u) \delta(v) = \delta(x-y) \quad (\epsilon \to 0) \eqend{,}
\end{splitequation}
where the second equality is a straightforward computation using the well-known results $\abs{x}' = \sgn(x)$, $\Theta'(x) = \delta(x)$ and $\sgn'(x) = 2 \delta(x)$, and the limit
\begin{equation}
\label{eq:sokhotski_plemelj}
\lim_{\epsilon \to 0} \frac{\epsilon}{x^2 + \epsilon^2} = \Im \lim_{\epsilon \to 0} \frac{1}{x - \mathi \epsilon} = \pi \delta(x)
\end{equation}
is the Sokhotski--Plemelj theorem.

The algebra $\mathfrak{A}_0$ can be completed in the H{\"o}rmander (weak) topology~\cite{hoermander} to the \emph{free-field algebra} $\overline{\mathfrak{A}}_0$. On a practical level, this completion is obtained by introducing normal-ordered products $\mathcal{N}$ in the usual way. Normal ordering can be performed with respect to the full two-point function $G^+$ or the Hadamard parametrix $H^+$; the latter choice has the advantage that the normal-ordered products transform covariantly under diffeomorphisms, or Lorentz transformations in Minkowski space. In $\overline{\mathfrak{A}}_0$, it is then possible to take limits such as $f(x,y) \to g(x) \delta(x-y)$ with $f \in \mathcal{S}(\mathbb{R}^2 \times \mathbb{R}^2)$ and $g \istest$ to obtain the normal-ordered Wick power $\mathcal{N}\left[ \phi^2(g) \right]$, which is a well-defined element of $\overline{\mathfrak{A}}_0$ with finite expectation values in any Hadamard state. Analogously to equation~\eqref{eq:normal_ordering_exponential_general}, we have for exponentials
\begin{equation}
\label{eq:normal_ordering_exponential_mink}
\mathcal{N}_G\left[ \mathe^{\mathi (J,\phi)} \right] = \mathe^{\frac{\mathi}{2} (J, G^+ \ast J)} \mathe_\star^{\mathi (J,\phi)}
\end{equation}
with $J \istest$ (again understood as a formal power series in $J$) and the analogous formula with $G$ replaced by $H$, from which the normal-ordering of monomials is obtained by functional differentiation with respect to $J$, assuming that the right-hand side is a well-defined element of $\mathfrak{A}_0$. In the massless case, that means that we have to take $\int J(x) \total^2 x = 0$, which ensures that it can be written as $J(x) = \partial_\mu J^\mu(x)$ (see, e.g.,~\cite[App.~B and~C]{froebhackhiguchi2017}) such that $(J,\phi) = - (J^\mu,\partial_\mu \phi)$ and only derivatives of $\phi$ enter. The expectation value of a normal-ordered quantity is given by
\begin{equation}
\label{eq:expectation_normal_ordered_mink}
\omega^{\Lambda,\epsilon} \left( \mathcal{N}_G\left[ (J_1,\phi) \cdots (J_n,\phi) \right] \right) = \delta_{n,0}
\end{equation}
analogous to equation~\eqref{eq:expectation_normal_ordered}, and for a change in normal ordering we have
\begin{equation}
\label{eq:change_normal_ordering_mink}
\mathcal{N}_G\left[ \mathe^{\mathi (J,\phi)} \right] = \mathe^{\frac{\mathi}{2} (J, (G^+-H^+) \ast J)} \, \mathcal{N}_H\left[ \mathe^{\mathi (J,\phi)} \right] \eqend{,}
\end{equation}
both formulas with the above restriction on $J$ and the second one understood as a formal power series in $J$. The vertex operators $V_\alpha(x)$ are formally given by the exponentials $\mathe^{\mathi \alpha \phi(x)}$, but those cannot be defined if we only consider derivatives of $\phi$. Instead, we proceed as follows: we add Hadamard-normal-ordered operators $\mathcal{N}_H\left[ V_\alpha(f) \right]$ to the generators of the algebra $\mathfrak{A}_0$, which should behave like the formal expression $\int f(x) \, \mathcal{N}_H\left[ \mathe^{\mathi \alpha \phi(x)} \right] \total^2 x$. Setting formally $J(y) = \alpha \delta(x-y)$ in equation~\eqref{eq:change_normal_ordering_exponential}, we see that we need to postulate the change of normal ordering
\begin{equation}
\label{eq:change_normal_ordering_vertex}
\mathcal{N}_G\left[ V_\alpha(x) \right] = \mathe^{\frac{\mathi}{2} \alpha^2 (G^+-H^+)(x,x)} \, \mathcal{N}_H\left[ V_\alpha(x) \right] = \left( \frac{\Lambda}{\mu} \right)^{- \frac{\alpha^2}{4 \pi}} \mathe^{\frac{\alpha^2}{2} W(x,x)} \, \mathcal{N}_H\left[ V_\alpha(x) \right]
\end{equation}
for the integral kernel, where we used the decomposition of the two-point function~\eqref{eq:twopf_hadamard}, or
\begin{equation}
\label{eq:change_normal_ordering_vertex_smeared}
\mathcal{N}_G\left[ V_\alpha(f) \right] = \mathcal{N}_H\left[ V_\alpha(f_{\alpha,W}) \right] \quad\text{with}\quad f_{\alpha,W}(x) \equiv \left( \frac{\Lambda}{\mu} \right)^{- \frac{\alpha^2}{4 \pi}} \mathe^{\frac{\alpha^2}{2} W(x,x)} f(x) \eqend{.}
\end{equation}
Furthermore, for their expectation value we set
\begin{equation}
\label{eq:expectation_vertex}
\omega^{\Lambda,\epsilon}\left( \mathcal{N}_G\left[ V_\alpha(f) \right] \right) = 1 \eqend{.}
\end{equation}

To derive the required bounds on correlation functions, we will later on need to compute expectation values of $\star$-products of normal-ordered vertex operators with other elements of $\mathfrak{A}_0$. For this, it is most useful to introduce further generators of the general form $\mathcal{N}_H\left[ V_{\alpha_1}(f_1) \cdots V_{\alpha_n}(f_n) \, \phi(g_1) \cdots \phi(g_k) \right]$ with $f_i, g_j \istest$ and $\int g_j(x) \total^2 x = 0$ for all $i \in \{1,\ldots,n\}, j \in \{1,\ldots,k\}$, i.e., normal-ordered products of multiple vertex operators and derivatives of $\phi$, and analogous expressions with $G$ instead of $H$. We postulate that
\begin{equation}
\label{eq:expectation_vertex_multiple}
\omega^{\Lambda,\epsilon}\left( \mathcal{N}_G\left[ V_{\alpha_1}(f_1) \cdots V_{\alpha_n}(f_n) \, \phi(g_1) \cdots \phi(g_k) \right] \right) = \delta_{k,0} \eqend{,}
\end{equation}
consistent with the interpretation of the vertex operators as exponentials of the field. Furthermore, we need relations between the $\star$-product of normal-ordered vertex operators and these generators, which we derive formally and then postulate; i.e., we quotient the free algebra by these relations. For this, we note that equation~\eqref{eq:normal_ordering_exponential_mink} implies (formally)
\begin{splitequation}
\label{eq:normal_ordering_starproduct}
&\mathcal{N}_G\left[ \mathe^{\mathi (J,\phi)} \right] \star \mathcal{N}_G\left[ \mathe^{\mathi (K,\phi)} \right] = \mathe^{\frac{\mathi}{2} (J, G^+ \ast J) + \frac{\mathi}{2} (K, G^+ \ast K)} \mathe_\star^{\mathi (J,\phi)} \star \mathe_\star^{\mathi (K,\phi)} \\
&\quad= \exp\left[ \frac{\mathi}{2} \Big( (J, G^+ \ast J) + (K, G^+ \ast K) - (J, \Delta \ast K) \Big) \right] \mathe_\star^{\mathi (J+K,\phi)} \\
&\quad= \exp\left[ - \mathi (J, G^+ \ast K) \right] \, \mathcal{N}_G\left[ \mathe^{\mathi (J+K,\phi)} \right]
\end{splitequation}
using the Baker--Campbell--Hausdorff formula~\cite{achillesbonfiglioli2012} and the commutation relations~\eqref{eq:commutation} and~\eqref{eq:commutation_g_delta}, as well as the analogous formula with $G$ replaced by $H$. Taking functional derivatives with respect to $J$ or $K$ and taking into account the condition $\int J(x) \total^2 x = 0$, we obtain relations for terms involving derivatives of $\phi$, while setting $J(x) = \alpha \delta(x-y)$ and interpreting $\mathe^{\mathi \alpha \phi(y)}$ as vertex operator $V_\alpha(y)$, we obtain the relations for them. As explained above, this derivation is not rigorous but provides us with the relations that we then postulate for $\mathfrak{A}_0$ (and its completion). For example, for two vertex operators we quotient $\mathfrak{A}_0$ by the relation
\begin{equation}
\mathcal{N}_H\left[ V_\alpha(y) \right] \star \mathcal{N}_H\left[ V_\beta(z) \right] = \exp\left[ - \mathi \alpha \beta H^+(y,z) \right] \, \mathcal{N}_H\left[ V_\alpha(y) V_\beta(z) \right] \eqend{,}
\end{equation}
which agrees with~\eqref{eq:normal_ordering_starproduct} (with $G \to H$) when setting $J(x) = \alpha \delta(x-y)$ and $K(x) = \beta \delta(x-z)$. Taking into account the analogous relation for $G$, the change of normal-ordering~\eqref{eq:change_normal_ordering_vertex} for single vertex operators, and the decomposition of the two-point function~\eqref{eq:twopf_hadamard}, we also obtain
\begin{equation}
\mathcal{N}_G\left[ V_\alpha(y) V_\beta(z) \right] = \left( \frac{\Lambda}{\mu} \right)^{- \frac{(\alpha+\beta)^2}{4 \pi}} \mathe^{\frac{\alpha^2}{2} W(y,y) - \alpha \beta W(y,z) - \frac{\beta^2}{2} W(z,z)} \, \mathcal{N}_H\left[ V_\alpha(y) V_\beta(z) \right] \eqend{,}
\end{equation}
the generalisation of the relation~\eqref{eq:change_normal_ordering_vertex}. We refrain from giving explicit expressions for all the resulting relations, since they are long and can in any case be easily obtained from equation~\eqref{eq:normal_ordering_starproduct}.

A central object in pAQFT are time-ordered products $\mathcal{T}$, which can be defined as multilinear maps from classical expressions into $\mathfrak{A}_0$. They are constructed inductively, using that the ones with single entries are equal to the Hadamard-normal-ordered products
\begin{equation}
\label{eq:timeordered_single}
\mathcal{T}\left[ \op(x) \right] = \mathcal{N}_H\left[ \op(x) \right] \eqend{,}
\end{equation}
where $\op(x)$ is a local functional, while the higher ones are defined outside the diagonal by causal factorisation:
\begin{splitequation}
\label{eq:timeordered_causal}
&\mathcal{T}\left[ \op_1(x_1) \otimes \cdots \otimes \op_n(x_n) \right] \\
&\quad= \mathcal{T}\left[ \op_1(x_1) \otimes \cdots \otimes \op_k(x_k) \right] \star \mathcal{T}\left[ \op_{k+1}(x_{k+1}) \otimes \cdots \otimes \op_n(x_n) \right]
\end{splitequation}
if none of the $x_1,\ldots,x_k$ lie in the past light cone of any of the $x_{k+1},\ldots,x_n$. The extension to the total diagonal $x_1 = \cdots = x_n$ corresponds to renormalisation. We then want to prove the analogue of Lemma~\ref{lemma_euclid_exponential_phi} (for $n = 2$), which is
\begin{lemma}
\label{lemma_mink_exponential_phi}
The expectation value of a time-ordered product of vertex operators and a bilinear operator is given by
\begin{splitequation}
\label{eq:correlator_exponential_2phi}
&\omega^{\Lambda,\epsilon}\left( \mathcal{T}\left[ \bigotimes_{j=1}^n V_{\alpha_j}(x_j) \otimes \left( \partial_{\vec{\mu}} \phi \, \partial_{\vec{\nu}} \phi \right)(z) \right] \right) = \exp\left[ - \mathi \sum_{1 \leq i < j \leq n} \alpha_i \alpha_j H^\text{F}(x_i,x_j) \right] \\
&\qquad\times \left[ \left. \partial^z_{\vec{\mu}} \partial^{z'}_{\vec{\nu}} W(z,z') \right\rvert_{z' = z} + \sum_{i,j=1}^n \alpha_i \alpha_j \partial_{\vec{\mu}} G^\text{F}(z,x_i) \partial_{\vec{\nu}} G^\text{F}(z,x_j) \right] \\
&\qquad\times \exp\left[ - \frac{1}{2} \sum_{i,j=1}^n \alpha_i \alpha_j W(x_i,x_j) \right] \left( \frac{\Lambda}{\mu} \right)^\frac{\left( \sum_{k=1}^n \alpha_k \right)^2}{4 \pi} \eqend{,}
\end{splitequation}
where $\partial_{\vec{\mu}} \equiv \partial_{\mu_1} \cdots \partial_{\mu_k}$ and $\partial_{\vec{\nu}} \equiv \partial_{\nu_1} \cdots \partial_{\nu_\ell}$ with $k,\ell \geq 1$, and the derivatives on the Feynman propagator $G^\text{F}$ act on the first argument $z$.
\end{lemma}
\begin{proof}
The result~\eqref{eq:correlator_exponential_2phi} follows immediately from the expression
\begin{splitequation}
\label{eq:timeordered_exponential_2phi}
&\mathcal{T}\left[ \bigotimes_{j=1}^n V_{\alpha_j}(x_j) \otimes \left( \partial_{\vec{\mu}} \phi \, \partial_{\vec{\nu}} \phi \right)(z) \right] = \exp\left[ - \mathi \sum_{1 \leq i < j \leq n} \alpha_i \alpha_j H^\text{F}(x_i,x_j) \right] \\
&\quad\times \Bigg[ \sum_{i,j=1}^n \alpha_i \alpha_j \partial_{\vec{\mu}} G^\text{F}(z,x_i) \partial_{\vec{\nu}} G^\text{F}(z,x_j) \, \mathcal{N}_G\left[ \prod_{j=1}^n V_{\alpha_j}(x_j) \right] \\
&\qquad- \sum_{i=1}^n \alpha_i \partial_{\vec{\mu}} G^\text{F}(z,x_i) \, \mathcal{N}_G\left[ \partial_{\vec{\nu}} \phi(z) \prod_{j=1}^n V_{\alpha_j}(x_j) \right] \\
&\qquad- \sum_{i=1}^n \alpha_i \partial_{\vec{\nu}} G^\text{F}(z,x_i) \, \mathcal{N}_G\left[ \partial_{\vec{\mu}} \phi(z) \prod_{j=1}^n V_{\alpha_j}(x_j) \right] \\
&\qquad+ \mathcal{N}_G\left[ \left( \partial_{\vec{\mu}} \phi \, \partial_{\vec{\nu}} \phi \right)(z) \prod_{j=1}^n V_{\alpha_j}(x_j) \right] + \left. \partial^z_{\vec{\mu}} \partial^{z'}_{\vec{\nu}} W(z,z') \right\rvert_{z' = z} \, \mathcal{N}_G\left[ \prod_{j=1}^n V_{\alpha_j}(x_j) \right] \Bigg] \\
&\qquad\times \exp\left[ - \frac{1}{2} \sum_{i,j=1}^n \alpha_i \alpha_j W(x_i,x_j) \right] \left( \frac{\Lambda}{\mu} \right)^\frac{\left( \sum_{k=1}^n \alpha_k \right)^2}{4 \pi}
\end{splitequation}
for the time-ordered product, using that the expectation value of a normal-ordered expression (with respect to $G$) involving powers of $\phi$ vanishes, while it is equal to $1$ if only vertex operators appear. To prove equation~\eqref{eq:timeordered_exponential_2phi}, we first have to show the related result
\begin{splitequation}
\label{eq:timeordered_exponential}
\mathcal{T}\left[ \bigotimes_{j=1}^n V_{\alpha_j}(x_j) \right] &= \exp\left[ - \mathi \sum_{1 \leq i < j \leq n} \alpha_i \alpha_j H^\text{F}(x_i,x_j) \right] \mathcal{N}_G\left[ \prod_{j=1}^n V_{\alpha_j}(x_j) \right] \\
&\quad\times \exp\left[ - \frac{1}{2} \sum_{i,j=1}^n \alpha_i \alpha_j W(x_i,x_j) \right] \left( \frac{\Lambda}{\mu} \right)^\frac{\left( \sum_{k=1}^n \alpha_k \right)^2}{4 \pi}
\end{splitequation}
by induction in $n$. Note that both sides are symmetric under a permutation of the $\alpha_j$ and $x_j$, the right-hand side because the Feynman Hadamard parametrix $H^\text{F}$~\eqref{eq:feynman} is symmetric in its arguments. For $n = 1$, we compute
\begin{splitequation}
\mathcal{T}\left[ V_\alpha(x) \right] = \mathcal{N}_H\left[ V_\alpha(x) \right] = \left( \frac{\Lambda}{\mu} \right)^\frac{\alpha^2}{4 \pi} \mathe^{- \frac{\alpha^2}{2} W(x,x)} \, \mathcal{N}_G\left[ V_\alpha(x) \right]
\end{splitequation}
using equation~\eqref{eq:timeordered_single} and the change of normal-ordering for vertex operators~\eqref{eq:change_normal_ordering_vertex}, which is the correct result. Assume thus that equation~\eqref{eq:timeordered_exponential} is fulfilled for all $m \leq n$, and consider the time-ordered products with $n+1$ vertex operators. We start with the case that not all points coincide, and afterwards extend the result to the diagonal. If not all points coincide, $k$ of them are not in the past light cone of any of the other $n+1-k$ for some $1 \leq k \leq n$, and by relabeling we may assume that these are the first. Using the causal factorisation~\eqref{eq:timeordered_causal}, we thus obtain
\begin{splitequation}
\label{eq:timeordered_exponential_factor1}
&\mathcal{T}\left[ \bigotimes_{j=1}^{n+1} V_{\alpha_j}(x_j) \right] = \mathcal{T}\left[ \bigotimes_{j=1}^k V_{\alpha_j}(x_j) \right] \star \mathcal{T}\left[ \bigotimes_{j=k+1}^{n+1} V_{\alpha_j}(x_j) \right] \\
&\quad= \exp\left[ - \mathi \sum_{1 \leq i < j \leq k} \alpha_i \alpha_j H^\text{F}(x_i,x_j) - \mathi \sum_{k+1 \leq i < j \leq n+1} \alpha_i \alpha_j H^\text{F}(x_i,x_j) \right] \\
&\qquad\times \mathcal{N}_G\left[ \prod_{i=1}^k V_{\alpha_i}(x_i) \right] \star \mathcal{N}_G\left[ \prod_{j=k+1}^{n+1} V_{\alpha_j}(x_j) \right] \left( \frac{\Lambda}{\mu} \right)^\frac{\left( \sum_{i=1}^k \alpha_i \right)^2 + \left( \sum_{j=k+1}^{n+1} \alpha_j \right)^2}{4 \pi} \\
&\qquad\times \exp\left[ - \frac{1}{2} \sum_{i,j=1}^k \alpha_i \alpha_j W(x_i,x_j) - \frac{1}{2} \sum_{i,j=k+1}^{n+1} \alpha_i \alpha_j W(x_i,x_j) \right] \eqend{,}
\end{splitequation}
using the induction hypothesis in the second equality. Employing equation~\eqref{eq:normal_ordering_starproduct} for the star product of two normal-ordered expressions (with the exponentials interpreted as vertex operators as explained there) and the decomposition of the two-point function~\eqref{eq:twopf_hadamard}, we obtain
\begin{splitequation}
\label{eq:timeordered_exponential_factor2}
&\mathcal{N}_G\left[ \prod_{i=1}^k V_{\alpha_i}(x_i) \right] \star \mathcal{N}_G\left[ \prod_{j=k+1}^{n+1} V_{\alpha_j}(x_j) \right] \\
&\quad= \exp\left[ - \mathi \sum_{i=1}^k \sum_{j=k+1}^{n+1} \alpha_i \alpha_j G^+(x_i,x_j) \right] \, \mathcal{N}_G\left[ \prod_{j=1}^{n+1} V_{\alpha_j}(x_j) \right] \\
&\quad= \exp\left[ - \mathi \sum_{i=1}^k \sum_{j=k+1}^{n+1} \alpha_i \alpha_j H^+(x_i,x_j) \right] \mathcal{N}_G\left[ \prod_{j=1}^{n+1} V_{\alpha_j}(x_j) \right] \\
&\qquad\quad\times \exp\left[ - \sum_{i=1}^k \sum_{j=k+1}^{n+1} \alpha_i \alpha_j W(x_i,x_j) \right] \left( \frac{\Lambda}{\mu} \right)^\frac{\sum_{i=1}^k \sum_{j=k+1}^{n+1} \alpha_i \alpha_j}{2 \pi} \eqend{.}
\end{splitequation}
Using that $H^+(x_i,x_j) = H^\text{F}(x_i,x_j)$ if $x_i$ does not lie in the past light cone of $x_j$ as we have assumed, and inserting the result~\eqref{eq:timeordered_exponential_factor2} into equation~\eqref{eq:timeordered_exponential_factor1}, one easily sees that the various terms combine into the required form~\eqref{eq:timeordered_exponential}, which therefore holds at least outside the diagonal. To extend the result to the diagonal, i.e., to the case where all points coincide, we note that since the right-hand side is a smooth function of the $x_i$ if $\epsilon > 0$, it simply extends by continuity. This is even true in the limit of vanishing UV cutoff $\epsilon$ if $\alpha_i^2 < 4 \pi$ for all $i$, since then the scaling degree of $\exp\left[ - \mathi \sum_{1 \leq i < j \leq n+1} \alpha_i \alpha_j H^\text{F}(x_i,x_j) \right]$ is less than $2 (m-1) = (m-1) \dim \mathbb{R}^2$ on each subdiagonal where $m$ points coincide, such that the singularities that arise there are integrable.

We can now prove equation~\eqref{eq:timeordered_exponential_2phi}, which we also do by induction in $n$. For $n = 0$, we compute
\begin{splitequation}
\mathcal{T}\left[ \left( \partial_{\vec{\mu}} \phi \, \partial_{\vec{\nu}} \phi \right)(z) \right] &= \mathcal{N}_H\left[ \left( \partial_{\vec{\mu}} \phi \, \partial_{\vec{\nu}} \phi \right)(z) \right] \\
&= \mathcal{N}_G\left[ \left( \partial_{\vec{\mu}} \phi \, \partial_{\vec{\nu}} \phi \right)(z) \right] + \left. \partial^z_{\vec{\mu}} \partial^{z'}_{\vec{\nu}} W(z,z') \right\rvert_{z' = z}
\end{splitequation}
using equation~\eqref{eq:timeordered_single}, the change of normal-ordering~\eqref{eq:change_normal_ordering_mink} at second order in $J$ [taking once $J(x) = \partial_{\vec{\mu}} \delta(z-x)$ and once $J(x) = \partial_{\vec{\nu}} \delta(z-x)$], as well as the decomposition of the two-point function~\eqref{eq:twopf_hadamard}. Since this agrees with equation~\eqref{eq:timeordered_exponential_2phi} for $n = 0$, the base case is proven. Assume thus that equation~\eqref{eq:timeordered_exponential_2phi} is fulfilled for all $m \leq n$, and consider the time-ordered products with $n+1$ vertex operators. If not all points coincide, again $k$ of them will not lie in the past light cone of the other $n+2-k$ for some $1 \leq k \leq n+1$. We thus can again use causal factorisation~\eqref{eq:timeordered_causal}, but now have to distinguish two cases, namely whether the distinguished point $z$ is among the group of $k$ or among the group of the $n+2-k$ points. In the first case, we have
\begin{splitequation}
&\mathcal{T}\left[ \bigotimes_{j=1}^{n+1} V_{\alpha_j}(x_j) \otimes \left( \partial_{\vec{\mu}} \phi \, \partial_{\vec{\nu}} \phi \right)(z) \right] \\
&\quad= \mathcal{T}\left[ \bigotimes_{j=1}^{k-1} V_{\alpha_j}(x_j) \otimes \left( \partial_{\vec{\mu}} \phi \, \partial_{\vec{\nu}} \phi \right)(z) \right] \star \mathcal{T}\left[ \bigotimes_{j=k}^{n+1} V_{\alpha_j}(x_j) \right] \eqend{,}
\end{splitequation}
and inserting the induction hypothesis and the previous result~\eqref{eq:timeordered_exponential} on the right-hand side, we obtain a number of star products of normal-ordered expressions, which are too long to display explicitly. To evaluate them, we need in addition to equation~\eqref{eq:timeordered_exponential_factor2} also
\begin{splitequation}
\label{eq:timeordered_exponential_phi_star}
&\mathcal{N}_G\left[ \partial_{\vec{\mu}} \phi(z) \prod_{j=1}^{k-1} V_{\alpha_j}(x_j) \right] \star \mathcal{N}_G\left[ \prod_{j=k}^{n+1} V_{\alpha_j}(x_j) \right] \\
&\quad= \exp\left[ - \mathi \sum_{i=1}^{k-1} \sum_{j=k}^{n+1} \alpha_i \alpha_j G^+(x_i,x_j) \right] \, \mathcal{N}_G\left[ \partial_{\vec{\mu}} \phi(z) \prod_{j=1}^{n+1} V_{\alpha_j}(x_j) \right] \\
&\qquad- \sum_{j=k}^{n+1} \alpha_j \partial_{\vec{\mu}} G^+(z,x_j) \exp\left[ - \mathi \sum_{i=1}^{k-1} \sum_{j=k}^{n+1} \alpha_i \alpha_j G^+(x_i,x_j) \right] \, \mathcal{N}_G\left[ \prod_{j=1}^{n+1} V_{\alpha_j}(x_j) \right]
\end{splitequation}
and
\begin{splitequation}
\label{eq:timeordered_exponential_phi2_star}
&\mathcal{N}_G\left[ \left( \partial_{\vec{\mu}} \phi \, \partial_{\vec{\nu}} \phi \right)(z) \prod_{j=1}^{k-1} V_{\alpha_j}(x_j) \right] \star \mathcal{N}_G\left[ \prod_{j=k}^{n+1} V_{\alpha_j}(x_j) \right] \\
&\quad= \exp\left[ - \mathi \sum_{i=1}^{k-1} \sum_{j=k}^{n+1} \alpha_i \alpha_j G^+(x_i,x_j) \right] \vast[ \mathcal{N}_G\left[ \left( \partial_{\vec{\mu}} \phi \, \partial_{\vec{\nu}} \phi \right)(z) \prod_{j=1}^{n+1} V_{\alpha_j}(x_j) \right] \\
&\qquad\qquad- 2 \sum_{j=k}^{n+1} \alpha_j \partial_{(\vec{\mu}} G^+(z,x_j) \, \mathcal{N}_G\left[ \partial_{\vec{\nu})} \phi(z) \prod_{j=1}^{n+1} V_{\alpha_j}(x_j) \right] \\
&\qquad\qquad+ \sum_{i,j=k}^{n+1} \alpha_i \alpha_j \partial_{\vec{\mu}} G^+(z,x_i) \partial_{\vec{\nu}} G^+(z,x_j) \, \mathcal{N}_G\left[ \prod_{j=1}^{n+1} V_{\alpha_j}(x_j) \right] \vast] \eqend{,}
\end{splitequation}
which are obtained by first taking one or two functional derivatives of equation~\eqref{eq:normal_ordering_starproduct} with respect to $J$, and then interpreting the exponentials as vertex operators, setting $J(x) = \sum_{j=1}^{k-1} \alpha_j \delta(x-x_j)$ and $K(x) = \sum_{j=k}^{n+1} \alpha_j \delta(x-x_j)$. Using further the decomposition of the two-point function~\eqref{eq:twopf_hadamard} and the fact that for $x_i$ not in the past light cone of $x_j$ we have $H^+(x_i,x_j) = H^\text{F}(x_i,x_j)$, the result~\eqref{eq:timeordered_exponential_2phi} follows in this case.

A similar computation yields equation~\eqref{eq:timeordered_exponential_2phi} also in the second case where the distinguished point $z$ is among the second group of $n+2-k$ points, for which we need the analogues of equations~\eqref{eq:timeordered_exponential_phi_star} and~\eqref{eq:timeordered_exponential_phi2_star} with the factors reversed. We omit the details. We have thus shown that equation~\eqref{eq:timeordered_exponential_2phi} holds at least outside the diagonal, but since the right-hand side is a smooth function of the $x_i$ if $\epsilon > 0$, it extends to the diagonal by continuity. However, this is no longer true for $\epsilon = 0$, and we resolve the renormalisation problem in section~\ref{sec_mink_renorm}.
\end{proof}

\begin{remark*}
While the choice of taking the Hadamard-normal-ordered expressions for time-ordered products with a single entry~\eqref{eq:timeordered_single} may seem strange to someone acquainted with flat-space quantum field theory, it is actually indispensable in curved spacetimes, since otherwise the renormalisation freedom is unacceptably large~\cite{hollandswald2001}. Moreover, in flat space of three or more dimensions the Hadamard parametrix actually coincides with the vacuum two-point function. It is only in two dimensions or for non-vacuum states that the difference becomes relevant, and as we see from equation~\eqref{eq:correlator_exponential_2phi} it is the correct choice to ensure the superselection rule in analogy with the Euclidean case.
\end{remark*}

\subsection{Proof of theorem~\ref{thm_mink_renorm} (Renormalisation)}
\label{sec_mink_renorm}

We begin again with $\op_{\mu\nu} = \partial_\mu \phi \, \partial_\nu \phi$. Using Lemma~\ref{lemma_mink_exponential_phi}, the decomposition of the two-point function and the explicit form of the Hadamard parametrix~\eqref{eq:feynman}, we obtain (with $\sigma_j = \pm 1$)
\begin{splitequation}
\label{eq:mink_renorm_expect}
&\omega^{\Lambda,\epsilon}\left( \mathcal{T}\left[ \bigotimes_{j=1}^n V_{\sigma_j \beta}(x_j) \otimes \op_{\mu\nu}(z) \right] \right) = \prod_{1 \leq i < j \leq n} \left[ \mu^2 ( - u_{ij} v_{ij} + \mathi \epsilon ) \right]^{\sigma_i \sigma_j \frac{\beta^2}{4 \pi}} \\
&\qquad\times \Bigg[ \left. \partial^z_\mu \partial^{z'}_\nu W(z,z') \right\rvert_{z' = z} - \beta^2 \sum_{i,j=1}^n \sigma_i \sigma_j \partial_\mu W(z,x_i) \partial_\nu W(z,x_j) \\
&\qquad\qquad+ \frac{\beta^2}{4 \pi} \sum_{i,j=1}^n \sigma_i \sigma_j \left[ H_\mu(z,x_i) \partial_\nu W(z,x_j) + H_\nu(z,x_i) \partial_\mu W(z,x_j) \right] \\
&\qquad\qquad- \frac{\beta^2}{(4 \pi)^2} \sum_{i,j=1}^n \sigma_i \sigma_j H_\mu(z,x_i) H_\nu(z,x_j) \Bigg] \\
&\qquad\times \exp\left[ - \frac{\beta^2}{2} \sum_{i,j=1}^n \sigma_i \sigma_j W(x_i,x_j) \right] \left( \frac{\Lambda}{\mu} \right)^\frac{\beta^2 \left( \sum_{k=1}^n \sigma_k \right)^2}{4 \pi} \eqend{,}
\end{splitequation}
where we defined
\begin{equations}[eq:mink_renorm_hmu_def]
H_u(x,y) &\equiv - 4 \pi \mathi \, \partial_u H^\text{F}(x,y) = \frac{\Theta(u+v)}{u - \mathi \epsilon} + \frac{\Theta(-(u+v))}{u + \mathi \epsilon} \eqend{,} \\
H_v(x,y) &\equiv - 4 \pi \mathi \, \partial_v H^\text{F}(x,y) = \frac{\Theta(u+v)}{v - \mathi \epsilon} + \frac{\Theta(-(u+v))}{v + \mathi \epsilon} \eqend{,}
\end{equations}
using the light cone coordinates $u = u(x,y)$ and $v = v(x,y)$ defined in equation~\eqref{eq:lightconecoords}, and set (for better readability)
\begin{equation}
\label{eq:lightcone_ij}
u_{ij} \equiv u(x_i,x_j) \eqend{,} \quad v_{ij} \equiv v(x_i,x_j) \eqend{.}
\end{equation}
Analogously to the Euclidean case, the terms $\left[ \mu^2 ( - u^{ij} v^{ij} + \mathi \epsilon ) \right]^{\sigma_j \sigma_k \frac{\beta^2}{4 \pi}}$ are singular in the physical limit $\epsilon \to 0$ if $\sigma_j \sigma_k = -1$, but the singularity is integrable since we are in the finite regime $\beta^2 < 4 \pi$. Since $W$ is a smooth function, it is not singular, but since the scaling degree of $H_\mu$~\eqref{eq:mink_renorm_hmu_def} is $-1$ and the integration measure in light cone coordinates~\eqref{eq:lightconecoords} is $\total^2 x = \frac{1}{2} \total u \total v$, terms involving $H_\mu$ can potentially be problematic and may need renormalisation. The terms of the form $H_\mu(z,x_i) \partial_\nu W(z,x_j)$ are by definition~\eqref{eq:mink_renorm_hmu_def} equal to $- 4 \pi \mathi \, \partial^z_\mu H^\text{F}(z,x_i) \partial_\nu W(z,x_j)$, and we can integrate the $z$ derivative by parts such that it either acts on $W$ or the smearing function. Since the singularity of $H^\text{F}$ is integrable, the result is a well-defined distribution and we can take the limit $\epsilon \to 0$ with impunity. The same holds for the terms $H_\mu(z,x_i) H_\nu(z,x_j)$ with $i \neq j$ since their scaling degree (when all points coincide) is $2 < \dim \mathbb{R}^4$, but as in the Euclidean case the terms $H_\mu(z,x_j) H_\nu(z,x_j)$ are problematic since their scaling degree is $2 = \dim \mathbb{R}^2$ such that we expect a logarithmic singularity as $\epsilon \to 0$.

\begin{lemma}
\label{lemma_mink_distribution}
Consider the family of distributions $H_{\rho\sigma}^\epsilon$ defined for $\epsilon > 0$ by
\begin{equation}
\label{eq:hmunu_def}
H_{\rho\sigma}^\epsilon(f) \equiv \int H_\rho(x,0) H_\sigma(x,0) f(x) \total^2 x \eqend{,}
\end{equation}
which are well-defined for $f \istest$ with $f(0) = 0$ also in the limit $\epsilon \to 0$. If $f(0) \neq 0$, it holds that
\begin{equation}
\label{eq:hmunu_decomposition}
\lim_{\epsilon \to 0} \left[ H_{\rho\sigma}^\epsilon(f) - H^\text{div}_{\rho\sigma}(f) - H^\text{ren}_{\rho\sigma}(f) \right] = 0 \eqend{,}
\end{equation}
where the divergent part $H_{\rho\sigma}^\text{div}$ is given by
\begin{equation}
\label{eq:hmunu_div}
H^\text{div}_{\rho\sigma}(f) = 4 \pi \mathi \ln(2 \mu \epsilon) \eta_{\rho\sigma} f(0) + \mathi \pi f(0) \eqend{,}
\end{equation}
and the renormalised part $H_{\rho\sigma}^\text{ren}$ reads
\begin{equation}
\label{eq:hmunu_ren}
H_{\rho\sigma}^\text{ren}(f) = 4 \pi \mathi \, \int H^\text{F}(x,0) \partial_\rho \partial_\sigma f(x) \total^2 x - 4 \pi^2 \eta_{\rho\sigma} \int \left[ H^\text{F}(x,0) \right]^2 \partial^2 f(x) \total^2 x \eqend{.}
\end{equation}
Requiring that $H_{\rho\sigma}^\text{ren}(f) = \lim_{\epsilon \to 0} u_{\rho\sigma}^\epsilon(f)$ for all $f \istest$ with $f(0) = 0$ and that it preserves Lorentzian covariance and the scaling degree, $H_{\rho\sigma}^\text{ren}$ is unique up to the choice of an arbitrary scale, which we have identified with the (also arbitrary) scale $\mu$ in the Hadamard parametrix.
\end{lemma}
\begin{proof}
We first show that $H^\epsilon_{\rho\sigma}$ is a well-defined distribution for all $f \istest$ with $f(0) = 0$ for all $\epsilon$ including the limit. We start with the $\rho = \sigma = u$ components, and compute first
\begin{splitequation}
\label{eq:mink_renorm_huhu}
H_u(x,0) H_u(x,0) &= \frac{\Theta(u+v)}{( u - \mathi \epsilon )^2} + \frac{\Theta(-(u+v))}{( u + \mathi \epsilon )^2} \\
&= - \partial_u H_u(x,0) + \frac{2 \mathi \epsilon}{u^2 + \epsilon^2} \delta(u+v) \\
&= 4 \pi \mathi \partial_u^2 H^\text{F}(x,0) + \frac{2 \mathi \epsilon}{u^2 + \epsilon^2} \delta(u+v) \eqend{.}
\end{splitequation}
It follows that
\begin{splitequation}
H_{uu}^\epsilon(f) &= \frac{1}{2} \int \left[ 4 \pi \mathi \partial_u^2 H^\text{F}(x,0) + \frac{2 \mathi \epsilon}{u^2 + \epsilon^2} \delta(u+v) \right] f(u,v) \total u \total v \\
&= 2 \pi \mathi \int H^\text{F}(x,0) \partial_u^2 f(u,v) \total u \total v + \int \frac{\mathi \epsilon}{u^2 + \epsilon^2} f(u,-u) \total u \eqend{,}
\end{splitequation}
where we switched to light cone coordinates~\eqref{eq:lightconecoords} with $y = 0$, for which
\begin{equation}
\label{eq:lightcone_dx}
\total^2 x = \frac{1}{2} \total u(x) \total v(x) \eqend{,} \quad u(x) = x^0 - x^1 \eqend{,} \quad v(x) = x^0 + x^1 \eqend{.}
\end{equation}
Since the singularity of $H^\text{F}$~\eqref{eq:feynman} is logarithmic and thus integrable for all $\epsilon$, the first term is a well-defined distribution, and for the second one we use the Sokhotski--Plemelj formula~\eqref{eq:sokhotski_plemelj} to obtain in the limit $\epsilon \to 0$
\begin{equation}
\lim_{\epsilon \to 0} H_{uu}^\epsilon(f) = 2 \pi \mathi \int H^\text{F}(x,0) \partial_u^2 f(u,v) \total u \total v + \mathi \pi f(0) \eqend{.}
\end{equation}

In the same way, we obtain
\begin{equation}
\lim_{\epsilon \to 0} H_{vv}^\epsilon(x,0) = 4 \pi \mathi \, \partial_v^2 H^\text{F}(x,0) + \mathi \pi \delta(x) \eqend{,}
\end{equation}
but the mixed terms are more complicated. For them, we compute for $\epsilon > 0$ that
\begin{splitequation}
&H_u(x,0) H_v(x,0) = \frac{\Theta(u+v)}{( u - \mathi \epsilon ) ( v - \mathi \epsilon )} + \frac{\Theta(-(u+v))}{( u + \mathi \epsilon ) ( v + \mathi \epsilon )} \\
&\quad= - 8 \pi^2 \partial_u \partial_v \left[ H^\text{F}(x,0) \right]^2 - \frac{2 \mathi \epsilon \ln\left[ \mu^2 \left( u^2 + \epsilon^2 \right) \right]}{u^2 + \epsilon^2} \delta(u+v) \eqend{,}
\end{splitequation}
using that $H^\text{F}$ is a fundamental solution~\eqref{eq:hf_fundamental_solution}. For the second term, we obtain furthermore
\begin{equation}
\frac{2 \mathi \epsilon \ln\left[ \mu^2 \left( u^2 + \epsilon^2 \right) \right]}{u^2 + \epsilon^2} = \partial_u \left[ \operatorname{Li}_2\left( \frac{u - \mathi \epsilon}{u + \mathi \epsilon} \right) - \operatorname{Li}_2\left( \frac{u + \mathi \epsilon}{u - \mathi \epsilon} \right) + 2 \ln(2 \mu \epsilon) \ln\left( \frac{\epsilon + \mathi u}{\epsilon - \mathi u} \right) \right] \eqend{,}
\end{equation}
where $\operatorname{Li}_2$ is the dilogarithm, defined for $\abs{z} \leq 1$ by~\cite[Eq.~(25.12.1)]{dlmf}
\begin{equation}
\operatorname{Li}_2(z) \equiv \sum_{k=1}^\infty \frac{z^k}{k^2} \eqend{.}
\end{equation}
It follows that
\begin{splitequation}
H_{uv}^\epsilon(f) &= - 4 \pi^2 \int \left[ H^\text{F}(x,0) \right]^2 \partial_u \partial_v f(u,v) \total u \total v \\
&\quad+ \frac{1}{2} \int \left[ \operatorname{Li}_2\left( \frac{u - \mathi \epsilon}{u + \mathi \epsilon} \right) - \operatorname{Li}_2\left( \frac{u + \mathi \epsilon}{u - \mathi \epsilon} \right) \right] \partial_u f(u,u) \total u \\
&\quad+ \ln(2 \mu \epsilon) \int \ln\left( \frac{\epsilon + \mathi u}{\epsilon - \mathi u} \right) \partial_u f(u,u) \total u \eqend{,}
\end{splitequation}
and we consider the limit $\epsilon \to 0$. Since the singularity of $H^\text{F}$~\eqref{eq:feynman} is logarithmic, also the singularity of $\left[ H^\text{F}(x,0) \right]^2$ is logarithmic and the first term is a well-defined distribution including in the limit $\epsilon \to 0$. For the second term, we note that $\abs{(u \pm \mathi \epsilon)/(u \mp \mathi \epsilon)} = 1$ such that the argument of the dilogarithm is a pure phase and the dilogarithm is bounded by its value at $z = 1$. We can thus interchange the limit $\epsilon \to 0$ with the integral, which causes this term to vanish. For the third term, let us only consider the integral without the prefactor $\ln(2 \mu \epsilon)$. Then we can also interchange the limit $\epsilon \to 0$ with the integral because the logarithm is bounded by $\pi$, but the limit of the integrand is discontinuous:
\begin{equation}
\lim_{\epsilon \to 0} \ln\left( \frac{\epsilon + \mathi u}{\epsilon - \mathi u} \right) = \mathi \pi \sgn u \eqend{.}
\end{equation}
It follows that
\begin{splitequation}
\lim_{\epsilon \to 0} \int \ln\left( \frac{\epsilon + \mathi u}{\epsilon - \mathi u} \right) \partial_u f(u,u) \total u = \mathi \pi \int \sgn u \, \partial_u f(u,u) \total u = - 2 \pi \mathi f(0) \eqend{,}
\end{splitequation}
such that $\lim_{\epsilon \to 0} H_{uv}^\epsilon(f)$ is finite whenever $f(0) = 0$.

For a general $f$, we use that in light cone coordinates~\eqref{eq:lightconecoords} we have $\eta_{uu} = \eta_{vv} = 0$ and $\eta_{uv} = \eta_{vu} = - \frac{1}{2}$ as well as $\partial^2 = - 4 \partial_u \partial_v$. The decomposition~\eqref{eq:hmunu_decomposition} is then immediate for $H_{uu}^\epsilon(f)$ and $H_{vv}^\epsilon(f)$, while for $H_{uv}^\epsilon(f)$ we need to perform an integration by parts and use that $H^\text{F}$ is a fundamental solution of the massless Klein--Gordon equation~\eqref{eq:hf_fundamental_solution} to obtain
\begin{equation}
H_{uv}^\text{ren}(f) = - \mathi \pi f(0) - 4 \pi^2 \int \left[ H^\text{F}(x,0) \right]^2 \partial_u \partial_v f(x) \total u \total v \eqend{.}
\end{equation}
With this result, the decomposition~\eqref{eq:hmunu_decomposition} is seen to hold also for $H_{uv}^\epsilon(f)$. Considering now a different renormalisation, the difference can only be a distribution supported at $x = 0$, and enforcing Lorentzian covariance and the preservation of the scaling degree, it must be proportional to $\eta_{\rho\sigma} f(0)$. Clearly, such a term corresponds to changing the scale $\mu$ both in the logarithm of the divergent part~\eqref{eq:hmunu_div} and the Hadamard parametrices in the renormalised part~\eqref{eq:hmunu_ren}.
\end{proof}

As in the Euclidean case, the local term $H^\text{div}_{\rho\sigma}$ that is divergent as the UV cutoff $\epsilon$ is removed is logarithmic as we anticipated, and $H^\text{ren}_{\rho\sigma}$ is the extension of $H_\rho H_\sigma$ to the diagonal. To renormalise, we have to subtract the local term, which in the pAQFT framework is done by changing the time-ordered products by local terms. We note that in the conventional pAQFT treatment, one would just consider $H^\text{ren}_{\rho\sigma}$ as the required extension, and no explicit counterterms are introduced. However, to connect with the traditional approach using cutoffs, we extend the established pAQFT approach and perform a change of time-ordered products to cancel the divergent part explicitly. That is, we show explicitly that the removal of divergent parts can be done using pAQFT methods, which thus can be seen as making the traditional approach mathematically precise.

In our case, with the local term supported at the diagonal $z = x_i$, we thus have to change the time-ordered products with two entries, one of which is $\op_{\mu\nu}(z)$ and one of which is a vertex operator $V_{\sigma_i \beta}(x_i)$. Since the local term that we want to subtract has scaling dimension 2, the same as $\op_{\mu\nu}$, the time-ordered product must be proportional to the vertex operator, and we make the ansatz
\begin{equation}
\label{eq:mink_renorm_redef}
\delta \mathcal{T}\left[ \op_{\mu\nu}(z) \otimes V_\alpha(x) \right] = c_\alpha \, H^\text{div}_{\mu\nu}(z,x) \, \mathcal{T}\left[ V_\alpha(x) \right]
\end{equation}
with a constant $c_\alpha$ to be determined. From the recursive construction of the time-ordered products, it follows that
\begin{equation}
\label{eq:mink_renorm_redef_2}
\delta \mathcal{T}\left[ \bigotimes_{j=1}^n V_{\sigma_j \beta}(x_j) \otimes \op_{\mu\nu}(z) \right] = \sum_{j=1}^n c_{\sigma_j \beta} H^\text{div}_{\mu\nu}(z,x_j) \, \mathcal{T}\left[ \bigotimes_{j=1}^n V_{\sigma_j \beta}(x_j) \right] \eqend{,}
\end{equation}
and hence the expectation value~\eqref{eq:mink_renorm_expect} changes as
\begin{splitequation}
\label{eq:mink_renorm_expect_change}
&\omega^{\Lambda,\epsilon}\left( \delta \mathcal{T}\left[ \bigotimes_{j=1}^n V_{\sigma_j \beta}(x_j) \otimes \op_{\mu\nu}(z) \right] \right) \\
&\quad= \sum_{j=1}^n c_{\sigma_j \beta} H^\text{div}_{\mu\nu}(z,x_j) \, \omega^{\Lambda,\epsilon}\left( \mathcal{T}\left[ \bigotimes_{j=1}^n V_{\sigma_j \beta}(x_j) \right] \right) \\
&\quad= \sum_{j=1}^n c_{\sigma_j \beta} H^\text{div}_{\mu\nu}(z,x_j) \, \prod_{1 \leq i < j \leq n} \left[ \mu^2 ( - u_{ij} v_{ij} + \mathi \epsilon ) \right]^{\sigma_i \sigma_j \frac{\beta^2}{4 \pi}} \\
&\qquad\times \exp\left[ - \frac{\beta^2}{2} \sum_{i,j=1}^n \sigma_i \sigma_j W(x_i,x_j) \right] \left( \frac{\Lambda}{\mu} \right)^\frac{\beta^2 \left( \sum_{j=1}^n \sigma_j \right)^2}{4 \pi}
\end{splitequation}
using the result~\eqref{eq:timeordered_exponential} and the explicit form of the Hadamard parametrix~\eqref{eq:feynman}. Adding this correction to~\eqref{eq:mink_renorm_expect}, we can cancel the divergent part by choosing
\begin{equation}
\label{eq:min_renorm_divconst}
c_{\pm \beta} = \frac{\beta^2}{(4 \pi)^2} \eqend{.}
\end{equation}
We see clearly that the renormalisation is state-independent, which is one of the central insights of pAQFT. Moreover, as in the Euclidean case we only obtain a non-vanishing result in the limit where the IR cutoff $\Lambda \to 0$ if the sum of all $\sigma_i$ vanishes, the super-selection criterion of the vacuum sector~\cite{swieca1977}.

For the renormalised expectation value of the stress tensor $T_{\mu\nu} = \op_{\mu\nu} - \frac{1}{2} \eta_{\mu\nu} \op_\rho{}^\rho + g \eta_{\mu\nu} ( V_\beta + V_{-\beta} )$, we again obtain a sum of four terms, the first two of which can be read off from equations~\eqref{eq:mink_renorm_expect} and~\eqref{eq:mink_renorm_expect_change}. As in the Euclidean case, since the part in~\eqref{eq:hmunu_div} that diverges in the limit $\epsilon \to 0$ is proportional to $\eta_{\mu\nu}$, it cancels out between the first two terms and the stress tensor is finite even without subtracting the counterterms. For the last two terms, we take the expectation value of equation~\eqref{eq:timeordered_exponential} and use the explicit form of the Hadamard parametrix~\eqref{eq:feynman} to obtain
\begin{splitequation}
\label{eq:mink_renorm_expect_vertex}
&\omega^{\Lambda,\epsilon}\left( \mathcal{T}\left[ \bigotimes_{j=1}^n V_{\sigma_j \beta}(x_j) \otimes V_\beta(z) \right] \right) = \prod_{1 \leq i < j \leq n} \left[ \mu^2 ( - u_{ij} v_{ij} + \mathi \epsilon ) \right]^{\sigma_i \sigma_j \frac{\beta^2}{4 \pi}} \\
&\qquad\times \prod_{j=1}^n \left[ \mu^2 ( - u(x_j,z) v(x_j,z) + \mathi \epsilon ) \right]^{\sigma_j \frac{\beta^2}{4 \pi}} \exp\left[ - \frac{\beta^2}{2} \sum_{i,j=1}^n \sigma_i \sigma_j W(x_i,x_j) \right] \\
&\qquad\times \exp\left[ - \beta^2 \sum_{j=1}^n \sigma_j W(x_j,z) - \frac{\beta^2}{2} W(z,z) \right] \left( \frac{\Lambda}{\mu} \right)^\frac{\beta^2 \left( 1 + \sum_{j=1}^n \sigma_j \right)^2}{4 \pi} \eqend{.}
\end{splitequation}
Since we are in the finite regime $\beta^2 < 4 \pi$, the singularities that arise for $\epsilon = 0$ as $x_j \to x_k$ and $x_j \to z$ are integrable, and so for this term no further renormalisation beyond the normal-ordering is required. Moreover, we again see how the neutrality condition appears: as $\Lambda \to 0$, we obtain a vanishing result unless $\sum_{j=1}^n \sigma_j = -1$. The last term with $V_{-\beta}$ results in the same result with $\sigma_j$ replaced by $-\sigma_j$ on the right-hand side. \hfill\squareforqed

\subsection{Proof of theorem~\ref{thm_mink_conv} (Convergence)}
\label{sec_mink_conv}

As in the Euclidean case, we tacitly employ Fubini's theorem to interchange absolutely convergent integrals in this whole section, and consider numerator and denominator of equation~\eqref{eq:thm_mink_conv_series} separately. Starting with the denominator, we use the result~\eqref{eq:timeordered_exponential} from the proof of Lemma~\ref{lemma_mink_exponential_phi} to obtain
\begin{splitequation}
\label{eq:mink_conv_expectvertex}
&\omega^{0,0}\left( \mathcal{T}\left[ \bigotimes_{j=1}^n V_{\sigma_j \beta}(x_j) \right] \right) = \lim_{\Lambda,\epsilon \to 0} \vast[ \exp\left[ - \mathi \beta^2 \sum_{1 \leq i < j \leq n} \sigma_i \sigma_j H^\text{F}(x_i,x_j) \right] \\
&\qquad\qquad\times \exp\left[ - \frac{\beta^2}{2} \sum_{i,j=1}^n \sigma_i \sigma_j W(x_i,x_j) \right] \left( \frac{\Lambda}{\mu} \right)^\frac{\beta^2 \left( \sum_{k=1}^n \sigma_k \right)^2}{4 \pi} \vast] \\
&\quad= \delta_{0, \sum_{k=1}^n \sigma_k} \prod_{1 \leq i < j \leq n} ( - \mu^2 )^{\sigma_i \sigma_j \frac{\beta^2}{4 \pi}} ( u_{ij} v_{ij} )_-^{\sigma_i \sigma_j \frac{\beta^2}{4 \pi}} \exp\left[ - \frac{\beta^2}{2} \sum_{i,j=1}^n \sigma_i \sigma_j W(x_i,x_j) \right] \eqend{,}
\end{splitequation}
where we use the notation
\begin{equation}
(uv)_\pm^\alpha \equiv \lim_{\epsilon \to 0} ( u v \pm \mathi \epsilon \abs{u+v} )^\alpha = \abs{u v}^\alpha \mathe^{\pm \mathi \alpha \pi \Theta(- u v)} = \left[ (uv)_\pm \right]^\alpha
\end{equation}
for the distributional boundary value, which is a well-defined integrable function for $\abs{\alpha} < 1$. Since $\sigma_j = \pm 1$, to obtain a non-vanishing result we must have $n = 2m$ with $m$ positive $\sigma_j$ and $m$ negative ones. We then rename the $x_j$ with $\sigma_j = -1$ to $y_j$ and renumber them. Taking into account that there are $\binom{n}{m} = (2m)!/(m!)^2$ possibilities to choose $m$ positive $\sigma_j$ from a total of $n = 2m$ ones (since equation~\eqref{eq:mink_conv_expectvertex} is symmetric under a permutation of the (renamed) $x_i$ and $y_j$ among themselves), the denominator of equation~\eqref{eq:thm_mink_conv_series} reduces to
\begin{splitequation}
\label{eq:mink_conv_denom}
&\sum_{m=0}^\infty \frac{(-1)^m}{(m!)^2} \int\dotsi\int \omega^{0,0}\left( \mathcal{T}\left[ \bigotimes_{j=1}^m V_\beta(x_j) \otimes V_{-\beta}(y_j) \right] \right) \prod_{i=1}^m g(x_i) g(y_i) \total^2 x_i \total^2 y_i \\
&= \sum_{m=0}^\infty \frac{(-1)^m}{(m!)^2} \int\dotsi\int \left[ \frac{\prod_{1 \leq j < k \leq m} [ u(x_j,x_k) v(x_j,x_k) ]_- [ u(y_j,y_k) v(y_j,y_k) ]_-}{(-1)^m \prod_{j,k=1}^m [ u(x_j,y_k) v(x_j,y_k) ]_-} \right]^\frac{\beta^2}{4 \pi} \\
&\qquad\times \mu^{- m \frac{\beta^2}{2 \pi}} \exp\left[ - \frac{\beta^2}{2} \sum_{i,j=1}^m [ W(x_i,x_j) - W(y_i,x_j) - W(x_i,y_j) + W(y_i,y_j) ] \right] \\
&\qquad\times \prod_{i=1}^m g(x_i) g(y_i) \total^2 x_i \total^2 y_i \eqend{.}
\end{splitequation}
To bound the terms at order $2m$, we change the integration measure to light cone coordinates~\eqref{eq:lightcone_dx} and use that~\eqref{eq:lightconecoords}
\begin{equation}
u(x,y) = u(x) - u(y) \eqend{,} \quad v(x,y) = v(x) - v(y) \eqend{.}
\end{equation}
The absolute value of the terms in brackets factorises into a part depending on the $u$ and a part depending on the $v$, and for each of them we use the Cauchy determinant formula:
\begin{splitequation}
\abs{ \frac{\prod_{1 \leq j < k \leq m} u(x_j,x_k) u(y_j,y_k)}{\prod_{j,k=1}^m u(x_j,y_k)} }^p &= \abs{ \det\left( \frac{1}{u(x_i,y_j)} \right)_{i,j=1}^m }^p \\
&\leq \abs{ \sum_{\pi} \prod_{j=1}^m \frac{1}{\abs{u(x_j,y_{\pi(j)})}} }^p \eqend{,}
\end{splitequation}
where the sum runs over all permutations $\pi$ of $\{ 1,\ldots,n \}$, and we used the estimate~\eqref{eq:cauchy_estimate1}. For the exponential, we use the second assumption on $W$
\begin{equation}
\sum_{i,j=1}^m [ W(x_i,x_j) - W(y_i,x_j) - W(x_i,y_j) + W(y_i,y_j) ] \geq 0 \eqend{,}
\end{equation}
which lets us estimate the exponential by 1. It follows that the denominator~\eqref{eq:mink_conv_denom} can be bounded by
\begin{splitequation}
\label{eq:mink_conv_bounddenom}
&\sum_{m=0}^\infty \frac{2^{-2m}}{(m!)^2} \int\dotsi\int \left[ \sum_{\pi} \prod_{j=1}^m \frac{1}{\mu \abs{u(x_j,y_{\pi(j)})}} \right]^\frac{\beta^2}{4 \pi} \left[ \sum_{\pi} \prod_{j=1}^m \frac{1}{\mu \abs{v(x_j,y_{\pi(j)})}} \right]^\frac{\beta^2}{4 \pi} \\
&\qquad\times \prod_{i=1}^m \abs{ g(x_i) } \abs{ g(y_i) } \total u(x_i) \total v(x_i) \total u(y_i) \total v(y_i) \eqend{.}
\end{splitequation}

We then estimate
\begin{splitequation}
\label{eq:mink_conv_gbound}
\abs{ g(x_i) } &\leq \big[ 1 + \mu^2 u(x_i)^2 \big]^{-2} \big[ 1 + \mu^2 v(x_i)^2 \big]^{-2} \\
&\quad\times \norm{ \big[ 1 + \mu^2 u(\blank)^2 \big]^2 \big[ 1 + \mu^2 v(\blank)^2 \big]^2 g(\blank) }_\infty \eqend{,}
\end{splitequation}
such that the integrals over the $u$ and the $v$ factorise. We thus need to bound
\begin{equation}
\label{eq:mink_conv_bounddenom_1}
\int\dotsi\int \left[ \sum_{\pi} \prod_{j=1}^m \frac{1}{\mu \abs{u(x_j,y_{\pi(j)})}} \right]^\frac{\beta^2}{4 \pi} \prod_{i=1}^m \frac{\total u(x_i)}{\big[ 1 + \mu^2 u(x_i)^2 \big]^2} \frac{\total u(y_i)}{\big[ 1 + \mu^2 u(y_i)^2 \big]^2} \eqend{,}
\end{equation}
and the analogous expression with $u$ replaced by $v$. A straightforward estimate using the same techniques as in the Euclidean case is possible, but would result in a factor $(m!)^2$ at order $m$ in the bound~\eqref{eq:mink_conv_bounddenom} for the denominator. With this, the convergence of the sum would only be possible for an adiabatic cutoff function $g$ of small support. However, it is possible to obtain a smaller power of $m!$ without restrictions on $g$ by employing more complicated estimates. Namely, using the Hölder inequality~\eqref{eq:hoelder} with $r = \rho > 1$, the expression~\eqref{eq:mink_conv_bounddenom_1} can be bounded by
\begin{splitequation}
\label{eq:mink_conv_bounddenom_2}
&\vast[ \int\dotsi\int \left[ \sum_{\pi} \prod_{j=1}^m \frac{1}{\mu \abs{u(x_j,y_{\pi(j)})}} \right]^{\rho \frac{\beta^2}{4 \pi}} \prod_{i=1}^m \frac{\total u(x_i)}{[ 1 + \mu^2 u(x_i)^2 ]^\rho} \frac{\total u(y_i)}{[ 1 + \mu^2 u(y_i)^2 ]^\rho} \vast]^\frac{1}{\rho} \\
&\quad\times \left[ \int\dotsi\int \prod_{i=1}^m \frac{\total u(x_i)}{\big[ 1 + \mu^2 u(x_i)^2 \big]^\frac{\rho}{\rho-1}} \frac{\total u(y_i)}{\big[ 1 + \mu^2 u(y_i)^2 \big]^\frac{\rho}{\rho-1}} \right]^\frac{\rho-1}{\rho} \eqend{,}
\end{splitequation}
and choosing $\rho$ such that $\rho \frac{\beta^2}{4 \pi} < 1$ (which is possible since we are in the finite regime $\beta^2 < 4 \pi$), we can use the estimate~\eqref{eq:cauchy_estimate2} to obtain
\begin{equation}
\left[ \sum_{\pi} \prod_{j=1}^m \frac{1}{\mu \abs{u(x_j,y_{\pi(j)})}} \right]^{\rho \frac{\beta^2}{4 \pi}} \leq \sum_\pi \prod_{j=1}^m \Big[ \mu \abs{u(x_j,y_{\pi(j)})} \Big]^{- \rho \frac{\beta^2}{4 \pi}} \eqend{.}
\end{equation}
Since the remainder of the integrand is invariant under a permutation of the $y_i$, the sum over permutations $\pi$ just gives a factor $m!$, so that we only need to bound
\begin{equation}
\label{eq:mink_conv_boundyoung}
\iint \Big[ \mu \abs{u(x_j,y_j)} \Big]^{- \rho \frac{\beta^2}{4 \pi}} \frac{\total u(x_j)}{[ 1 + \mu^2 u(x_j)^2 ]^\rho} \frac{\total u(y_j)}{[ 1 + \mu^2 u(y_j)^2 ]^\rho} \eqend{.}
\end{equation}
In the region where $\mu \abs{ u(x_j,y_j) } > 1$, we estimate the first term by $1$ and bound equation~\eqref{eq:mink_conv_boundyoung} by $\left[ \int ( 1 + \mu^2 u^2 )^{-\rho} \total u \right]^2 \leq \pi^2 \mu^{-2}$. In the region where $\mu \abs{ u(x_j,y_j) } \leq 1$, we use again Young's inequality in the form~\eqref{eq:young3} with the exponents~\eqref{eq:young_qpr_choice}, but now in one dimension and with $\beta^2$ replaced by $\rho \beta^2$. This gives
\begin{splitequation}
&\norm{ ( 1 + \mu^2 \blank^2 )^{-\rho} }_p^2 \norm{ \Theta(1 - \mu \abs{\blank}) (\mu \abs{\blank})^{- \rho \frac{\beta^2}{4 \pi}} }_q \\
&\quad= \mu^{-2} \left( \frac{\sqrt{\pi} \Gamma\left( p \rho - \frac{1}{2} \right)}{\Gamma(p \rho)} \right)^\frac{2}{p} \left( \frac{8 \pi}{4 \pi - \rho \beta^2 q} \right)^\frac{1}{q} \eqend{,}
\end{splitequation}
where the result for the second norm is equation~\eqref{eq:young_thetamu_norm}, and the first norm is a straightforward computation using~\cite[Eq.~(5.12.3)]{dlmf} and~\cite[Eq.~(5.12.1)]{dlmf}. The result is finite with the choice we made for $p$, $q$ and $\rho$, in particular $\rho \beta^2 q = 4 \pi - ( 4 \pi - \rho \beta^2 ) ( 8 \pi - \rho \beta^2 ) / (8 \pi) < 4 \pi$ and $p \rho > 1$. On the other hand, for the second factor in equation~\eqref{eq:mink_conv_bounddenom_2} we have the simple bound
\begin{equation}
\left[ \int \frac{\total u}{\big( 1 + \mu^2 u^2 \big)^\frac{\rho}{\rho-1}} \right]^\frac{\rho-1}{\rho} \leq \rho \mu^\frac{1-\rho}{\rho} \eqend{.}
\end{equation}
Taking all together, we can bound equation~\eqref{eq:mink_conv_bounddenom_2} (and thus equation~\eqref{eq:mink_conv_bounddenom_1}) by $(m!)^\frac{1}{\rho} \hat{K}^m$, where $\hat{K}$ is a constant depending on $\beta$ and $g$, and we recall that $\rho > 1$. Inserting this result into equation~\eqref{eq:mink_conv_bounddenom}, it follows that the denominator of equation~\eqref{eq:thm_mink_conv_series} is bounded by
\begin{equation}
\label{eq:mink_conv_denombound}
\sum_{m=0}^\infty \frac{2^{-2m}}{(m!)^2} (m!)^\frac{2}{\rho} \hat{K}^{2m} = \sum_{m=0}^\infty (m!)^{\frac{\beta^2}{4\pi} - 1} K^m < \infty \eqend{,}
\end{equation}
with the new constant $K = \hat{K}^2/4$, and where we made the (admissible) choice $\rho = 8 \pi/(4 \pi + \beta^2)$. We remark that the bounds~\eqref{eq:mink_conv_denombound} are not new and were already derived in~\cite{bahnsrejzner2018}. However, a technical improvement over the proof of~\cite{bahnsrejzner2018} is that we admit arbitrary adiabatic cutoff functions $g \istest$, without any restriction on their support.

Consider thus the numerator of equation~\eqref{eq:thm_mink_conv_series}, where the renormalised expectation values are given by the sum of~\eqref{eq:mink_renorm_expect} and~\eqref{eq:mink_renorm_expect_change} with the choice~\eqref{eq:min_renorm_divconst} to cancel the divergent part. We see that in the physical limit $\Lambda \to 0$ again only even terms with $n = 2m$ contribute, and taking into account the symmetry under the exchange of variables and renaming integration variables as in the case of the denominator, the numerator reads
\begin{splitequation}
\label{eq:mink_conv_num}
&\sum_{m=0}^\infty \frac{(-1)^m}{(m!)^2} \int f(z) \int\dotsi\int \left[ \frac{\prod_{1 \leq j < k \leq m} [ u(x_j,x_k) v(x_j,x_k) ]_- [ u(y_j,y_k) v(y_j,y_k) ]_-}{(-1)^m \prod_{j,k=1}^m [ u(x_j,y_k) v(x_j,y_k) ]_-} \right]^\frac{\beta^2}{4 \pi} \\
&\qquad\times \Bigg[ \left. \partial^z_\mu \partial^{z'}_\nu W(z,z') \right\rvert_{z' = z} + \frac{\beta^2}{8 \pi^2} \sum_{i,j=1}^m H_{(\mu}(z,x_i) H_{\nu)}(z,y_j) \\
&\qquad\qquad- \beta^2 \sum_{i,j=1}^m \partial_\mu [ W(z,x_i) - W(z,y_i) ] \partial_\nu [ W(z,x_j) - W(z,y_j) ] \\
&\qquad\qquad+ \frac{\beta^2}{2 \pi} \sum_{i,j=1}^m [ H_{(\mu}(z,x_i) - H_{(\mu}(z,y_i) ] \partial_{\nu)} [ W(z,x_j) - W(z,y_j) ] \\
&\qquad\qquad- \frac{\beta^2}{8 \pi^2} \sum_{1 \leq i < j \leq m} \left[ H_{(\mu}(z,x_i) H_{\nu)}(z,x_j) + H_{(\mu}(z,y_i) H_{\nu)}(z,y_j) \right] \\
&\qquad\qquad- \frac{\beta^2}{16 \pi^2} \sum_{j=1}^m \Big[ H_{\mu\nu}^\text{ren}(z,x_j) + H_{\mu\nu}^\text{ren}(z,y_j) \Big] \Bigg] \\
&\qquad\times \mu^{- m \frac{\beta^2}{2 \pi}} \exp\left[ - \frac{\beta^2}{2} \sum_{i,j=1}^m [ W(x_i,x_j) - W(y_i,x_j) - W(x_i,y_j) + W(y_i,y_j) ] \right] \\
&\qquad\times \prod_{i=1}^m g(x_i) g(y_i) \total^2 x_i \total^2 y_i \total^2 z \eqend{,}
\end{splitequation}
where here and in the following indices in parentheses are to be symmetrised according to $A_{(\mu} B_{\nu)} = ( A_\mu B_\nu + A_\nu B_\mu )/2$. We see that there are various different types of terms, some of which are equal since we can interchange $x_j$ and $y_j$ without changing the result, and we will bound all of them separately. Consider first the terms that only contain $W$ but no derivatives $H_\mu$ of the Hadamard parametrix. By the first assumption on $W$, it and its derivatives grow at most polynomially such that
\begin{equations}[eq:w_bounds]
\abs{ \left. \partial^z_\mu \partial^{z'}_\nu W(z,z') \right\rvert_{z' = z} } &\leq w \Big[ 1 + \mu^2 u(z)^2 + \mu^2 v(z)^2 \Big]^k \eqend{,} \\
\abs{ \partial_\mu W(z,x) } &\leq w \Big[ 1 + \mu^2 u(z)^2 + \mu^2 v(z)^2 \Big]^k \Big[ 1 + \mu^2 u(x)^2 + \mu^2 v(x)^2 \Big]^k
\end{equations}
for some constant $w > 0$ and some $k \in \mathbb{N}$. It follows that the contribution of the terms that only contain $W$ but no derivatives $H_\mu$ to the numerator~\eqref{eq:mink_conv_num} can be bounded by
\begin{splitequation}
\label{eq:mink_conv_num_1}
&\sum_{m=0}^\infty \frac{\mu^{- m \frac{\beta^2}{2 \pi}}}{(m!)^2} \int \abs{ f(z) } \int\dotsi\int \abs{ \frac{\prod_{1 \leq j < k \leq m} u(x_j,x_k) v(x_j,x_k) u(y_j,y_k) v(y_j,y_k)}{\prod_{j,k=1}^m u(x_j,y_k) v(x_j,y_k)} }^\frac{\beta^2}{4 \pi} \\
&\qquad\times \Bigg[ w \Big[ 1 + \mu^2 u(z)^2 + \mu^2 v(z)^2 \Big]^k + \beta^2 w^2 \Big[ 1 + \mu^2 u(z)^2 + \mu^2 v(z)^2 \Big]^{2k} \\
&\qquad\qquad\times \left[ \sum_{i=1}^m \Big[ 1 + \mu^2 u(x_i)^2 + \mu^2 v(x_i)^2 \Big]^k + \Big[ 1 + \mu^2 u(y_i)^2 + \mu^2 v(y_i)^2 \Big]^k \right]^2 \Bigg] \\
&\qquad\times \exp\left[ - \frac{\beta^2}{2} \sum_{i,j=1}^m [ W(x_i,x_j) - W(y_i,x_j) - W(x_i,y_j) + W(y_i,y_j) ] \right] \\
&\qquad\times \prod_{i=1}^m \abs{ g(x_i) } \abs{ g(y_i) } \total^2 x_i \total^2 y_i \total^2 z \eqend{.}
\end{splitequation}
The integral over $z$ can be estimated by a constant $C$ depending on the test function $f$ and the constants $w$, $k$ and $\beta^2$. Absorbing the terms $\Big[ 1 + \mu^2 u(x_i)^2 + \mu^2 v(x_i)^2 \Big]^k$ into the test functions $g(x_i)$ to obtain new test functions $\tilde{g}$, and taking into account that the middle sum in equation~\eqref{eq:mink_conv_num_1} contributes $m^2$ terms, which by renaming of integration variables all give the same contribution, we can then repeat the derivation of the denominator estimates for the remaining terms. It follows that equation~\eqref{eq:mink_conv_num_1} is bounded by
\begin{equation}
\label{eq:mink_conv_num_1_bound}
C \sum_{m=0}^\infty (m!)^{\frac{\beta^2}{4\pi} - 1} m^2 K^m < \infty \eqend{,}
\end{equation}
where $K$ depends on $\tilde{g}$ and thus also on $W$. Next consider the mixed terms involving $W$ and $H_\mu = - 4 \pi \mathi \partial_\mu H^\text{F}$~\eqref{eq:mink_renorm_hmu_def}. Integrating the derivative by parts and using the estimates~\eqref{eq:w_bounds} for $W$, the contribution of these terms to the numerator~\eqref{eq:mink_conv_num} can be bounded by
\begin{splitequation}
\label{eq:mink_conv_num_2}
&\sum_{m=0}^\infty \frac{\mu^{- m \frac{\beta^2}{2 \pi}}}{(m!)^2} \int \int\dotsi\int \abs{ \frac{\prod_{1 \leq j < k \leq m} u(x_j,x_k) v(x_j,x_k) u(y_j,y_k) v(y_j,y_k)}{\prod_{j,k=1}^m u(x_j,y_k) v(x_j,y_k)} }^\frac{\beta^2}{4 \pi} \\
&\qquad\times 8 m^2 w \beta^2 \left[ \abs{ f(z) } + \sup_{\mu} \abs{ \partial_\mu f(z) } \right] \Big[ 1 + \mu^2 u(z)^2 + \mu^2 v(z)^2 \Big]^k \abs{ H^\text{F}(z,x_1) } \\
&\qquad\times \exp\left[ - \frac{\beta^2}{2} \sum_{i,j=1}^m [ W(x_i,x_j) - W(y_i,x_j) - W(x_i,y_j) + W(y_i,y_j) ] \right] \\
&\qquad\times \prod_{i=1}^m \abs{ \tilde{g}(x_i) } \abs{ \tilde{g}(y_i) } \total^2 x_i \total^2 y_i \total^2 z \eqend{,}
\end{splitequation}
where we have again absorbed terms $\Big[ 1 + \mu^2 u(x_i)^2 + \mu^2 v(x_i)^2 \Big]^k$ into the test functions $g(x_i)$, and the factor of $m^2$ arises because all the terms in the sum give the same contribution. Here we used that the Hadamard parametrix $H^\text{F}$~\eqref{eq:feynman} is an integrable function of the light cone variables $u$ and $v$, uniformly bounded for all $\epsilon$. We can thus bound it by simply taking the absolute value of this integrable function, denoted by $\abs{ H^\text{F} }$, which in the limit $\epsilon \to 0$ reads
\begin{equation}
\label{eq:mink_conv_hadamard_bound}
\abs{ H^\text{F}(z,x) } \leq \frac{1}{4 \pi} \abs{ \ln\abs{ \mu u(z,x) } } + \frac{1}{4 \pi} \abs{ \ln\abs{ \mu v(z,x) } } + \frac{1}{4} \eqend{.}
\end{equation}
The contribution of the last term is bounded by
\begin{equation}
\int 2 w \beta^2 \left[ \abs{ f(z) } + \sup_{\mu} \abs{ \partial_\mu f(z) } \right] \Big[ 1 + \mu^2 u(z)^2 + \mu^2 v(z)^2 \Big]^k \total^2 z \leq C \eqend{,}
\end{equation}
where the constant $C$ depends on $f$ and $W$ through $w$ and $k$. For the contributions of the logarithms in~\eqref{eq:mink_conv_hadamard_bound}, we first bound
\begin{splitequation}
&8 w \beta^2 \left[ \abs{ f(z) } + \sup_{\mu} \abs{ \partial_\mu f(z) } \right] \Big[ 1 + \mu^2 u(z)^2 + \mu^2 v(z)^2 \Big]^k \\
&\quad\leq C \Big[ 1 + \mu^2 u(z)^2 \Big]^{-1} \Big[ 1 + \mu^2 v(z)^2 \Big]^{-1} \eqend{,}
\end{splitequation}
where the constant $C$ depends on $f$ and $W$ through $w$ and $k$. We then change the $z$ integration to light cone coordinates~\eqref{eq:lightcone_dx}, such that the integral over $z$ factors, and then use the estimate
\begin{equation}
\label{eq:mink_conv_num_2_int}
\int \frac{\abs{ \ln\abs{ \mu u(z,x) } }^k}{1 + \mu^2 u(z)^2} \total u(z) \leq \frac{2}{\mu} c_k \ln^k\left( 2 + \mu \abs{u(x)} \right) \eqend{.}
\end{equation}
and the analogous one with $u$ replaced by $v$. It follows that the integral over $z$ in equation~\eqref{eq:mink_conv_num_2} is bounded by a constant $C$ depending on $f$, $g$ and $W$. For the remaining terms we repeat the derivation of the denominator estimates with $g(x_i) \to g(x_i) \ln\left( 2 + \mu \abs{u(x_i)} \right)$, since for those estimates we only need that $\abs{g(x_i)} \ln\left( 2 + \mu \abs{u(x_i)} \right)$ is a rapidly decreasing function. It follows that also the contribution of the terms involving both $W$ and $H_\mu$ to the numerator~\eqref{eq:mink_conv_num} is bounded by a sum of the form~\eqref{eq:mink_conv_num_1_bound}.

The bound~\eqref{eq:mink_conv_num_2_int} is proven as follows: we estimate first that
\begin{splitequation}
\int \frac{\abs{ \ln\abs{ \mu u(z,x) } }^k}{1 + \mu^2 u(z)^2} \total u(z) &= \frac{1}{\mu} \int \frac{\abs{ \ln\abs{ s } }^k}{1 + [ s + \mu u(x) ]^2} \total s \\
&\leq \frac{2}{\mu} \int_0^\infty \frac{\abs{ \ln s }^k}{1 + [ s - \mu \abs{ u(x) } ]^2} \total s \eqend{,}
\end{splitequation}
and thus only have to show that
\begin{equation}
\int_0^\infty \frac{\abs{ \ln s }^k}{1 + (s-a)^2} \total s \leq c_k \ln^k(2+a)
\end{equation}
for $a \geq 0$. We compute
\begin{splitequation}
&\int_0^\infty \frac{\abs{ \ln s }^k}{1 + (s-a)^2} \total s = \int_0^1 \frac{(-\ln s)^k}{1 + (s-a)^2} \total s + \int_1^\infty \frac{\ln^k s}{1 + (s-a)^2} \total s \\
&\quad\leq \int_0^1 (-\ln s)^k \total s + \int_1^{1+a} \frac{\ln^k s}{1 + (s-a)^2} \total s + \int_{1+a}^\infty \frac{\ln^k s}{1 + (s-a)^2} \total s \\
&\quad\leq k! + \ln^k(1+a) \int_1^{1+a} \frac{1}{1 + (s-a)^2} \total s + \int_1^\infty \frac{\ln^k(s+a)}{1 + s^2} \total s \\
&\quad\leq k! + \frac{3}{4} \pi \ln^k(1+a) + \int_1^\infty \frac{[ \ln(s) + \ln(2+a) ]^k}{1 + s^2} \total s
\end{splitequation}
using the inequality (valid for $a \geq 0$, $b \geq 1$)
\begin{equation}
\label{eq:mink_conv_log_ineq}
\ln(a+b) = \ln b + \ln\left( 1 + \frac{a}{b} \right) \leq \ln b + \ln(2+a) \eqend{.}
\end{equation}
Using further that
\begin{equation}
\int_1^\infty \frac{\ln^k(s)}{1 + s^2} \total s \leq k! \eqend{,}
\end{equation}
we obtain
\begin{splitequation}
\int_0^\infty \frac{\abs{ \ln s }^k}{1 + (s-a)^2} \total s &\leq k! + \frac{3}{4} \pi \ln^k(1+a) + \sum_{m=0}^k \binom{k}{m} \ln^{k-m}(2+a) m! \\
&\leq \left( 2 k! + \frac{3}{4} \pi \right) \ln^k(2+a)
\end{splitequation}
as required.

The remaining terms in equation~\eqref{eq:mink_conv_num} either involve two derivatives $H_\mu$ of the Hadamard parametrix at different points, or the renormalised distribution $H_{\mu\nu}^\text{ren}(z,x)$~\eqref{eq:hmunu_ren}. We start with the latter ones, which contain local terms proportional to $\delta(z-x)$ as well as derivatives acting on the Hadamard parametrix and its square. The local terms allow to perform the integral over $z$, resulting in a factor $f(x_i)$ or $f(y_j)$ which can be estimated by $\norm{ f }_\infty$, and for the remaining terms we can repeat the derivation of the denominator estimates. The terms with derivatives acting on the Hadamard parametrix and its square are integrated by parts to act on $f$, and as before we introduce factors of $[ 1 + \mu^2 u(z)^2 ]$ and $[ 1 + \mu^2 v(z)^2 ]$ and use that $\abs{ [ 1 + \mu^2 u(z)^2 ] [ 1 + \mu^2 v(z)^2 ] \partial_\mu \partial_\nu f(z) } \leq C$. The remaining integral over $z$ is then bounded using equation~\eqref{eq:mink_conv_num_2_int}. We can again absorb the logarithm in the test functions $g$, and since there are $2 m \leq 2 m^2$ terms involving $H_{\mu\nu}^\text{ren}(z,x)$, it follows that also their contribution to the numerator~\eqref{eq:mink_conv_num} is bounded by a sum of the form~\eqref{eq:mink_conv_num_1_bound}.

To bound the remaining terms with two derivatives $H_\mu$ of the Hadamard parametrix at different points, we first consider the case where $\mu = u$, $\nu = v$. Using that $H^\text{F}$ is a fundamental solution of the Klein--Gordon equation~\eqref{eq:hf_fundamental_solution} and the definition of $H_\mu$~\eqref{eq:mink_renorm_hmu_def}, a straightforward computation results in
\begin{splitequation}
\label{eq:mink_conv_hu_hv_ibp}
\int H_{(u}(z,x) H_{v)}(z,y) f(z) \total^2 z &= - 2 \pi^2 H^\text{F}(x,y) [ f(x) + f(y) ] \\
&\quad- 8 \pi^2 \int H^\text{F}(z,x) H^\text{F}(z,y) \partial_u \partial_v f(z) \total^2 z \eqend{.}
\end{splitequation}
To estimate the second term, we introduce factors of $[ 1 + \mu^2 u(z)^2 ]$ and $[ 1 + \mu^2 v(z)^2 ]$ and use Hölder's inequality~\eqref{eq:hoelder} with $r = 2$ to obtain
\begin{splitequation}
&\abs{ \int H^\text{F}(z,x) H^\text{F}(z,y) \partial_u \partial_v f(z) \total^2 z } \leq \norm{ [ 1 + \mu^2 u(\blank)^2 ] [ 1 + \mu^2 v(\blank)^2 ] \partial_u \partial_v f(\blank) }_\infty \\
&\quad\times \left[ \int \frac{\abs{ H^\text{F}(z,x) }^2}{[ 1 + \mu^2 u(z)^2 ] [ 1 + \mu^2 v(z)^2 ]} \total^2 z \int \frac{\abs{ H^\text{F}(z,y) }^2}{[ 1 + \mu^2 u(z)^2 ] [ 1 + \mu^2 v(z)^2 ]} \total^2 z \right]^\frac{1}{2} \eqend{.}
\end{splitequation}
Using the bound~\eqref{eq:mink_conv_hadamard_bound} for the Hadamard parametrix and the estimate~\eqref{eq:mink_conv_num_2_int}, we obtain the bound
\begin{splitequation}
\abs{ \int H^\text{F}(z,x) H^\text{F}(z,y) \partial_u \partial_v f(z) \total^2 z } &\leq C \ln\left( 2 + \mu \abs{u(x)} \right) \ln\left( 2 + \mu \abs{v(x)} \right) \\
&\quad\times \ln\left( 2 + \mu \abs{u(y)} \right) \ln\left( 2 + \mu \abs{v(y)} \right) \eqend{,}
\end{splitequation}
where the constant $C$ depends on $f$, and we can again absorb the logarithms in the test functions $g$. Repeating the derivation of the denominator estimates and taking into account that there are $m^2 + m(m-1) < 2 m^2$ terms with two derivatives $H_\mu$ of the Hadamard parametrix at different points, it follows that also the contribution of the second term in equation~\eqref{eq:mink_conv_hu_hv_ibp} to the numerator~\eqref{eq:mink_conv_num} is bounded by a sum of the form~\eqref{eq:mink_conv_num_1_bound}. On the other hand, the first term in equation~\eqref{eq:mink_conv_hu_hv_ibp} is directly bounded using the bound~\eqref{eq:mink_conv_hadamard_bound} for the Hadamard parametrix, which is however logarithmically divergent for small $u(x,y)$ or $v(x,y)$. We can then almost repeat the derivation of the denominator estimates, except that we need to bound equation~\eqref{eq:mink_conv_boundyoung} in the case that an additional logarithm is present, i.e., we need to bound
\begin{equation}
\label{eq:mink_conv_boundyoung_log}
\iint \Big[ \mu \abs{u(x_j,y_j)} \Big]^{- \rho \frac{\beta^2}{4 \pi}} \abs{\ln\abs{\mu u(x_j,y_k)}} \frac{\total u(x_j)}{[ 1 + \mu^2 u(x_j)^2 ]^\rho} \frac{\total u(y_j)}{[ 1 + \mu^2 u(y_j)^2 ]^\rho}
\end{equation}
in the two cases $k = j$ and $k \neq j$. We start with the case $k = j$, and use the well-known bound\footnote{It can be proven in the standard way, noting that for $f(z) = \ln z - \frac{1}{r} ( z^r - 1 )$ we have $f(1) = 0$ and $z f'(z) = 1 - z^r \leq 0$ for $z \geq 1$, $z f'(z) \geq 0$ for $z \leq 1$.}
\begin{equation}
\ln z \leq \frac{1}{r} ( z^r - 1 ) \leq \frac{1}{r} z^r
\end{equation}
for $z,r \geq 0$ to estimate
\begin{equation}
\Big[ \mu \abs{u(x_j,y_j)} \Big]^{- \rho \frac{\beta^2}{4 \pi}} \abs{\ln\abs{\mu u(x_j,y_j)}} \leq \frac{8 \pi}{\rho \beta^2} \Big[ \mu \abs{u(x_j,y_j)} \Big]^{- \rho \frac{\beta^2}{8 \pi}} \eqend{.}
\end{equation}
In the region $\mu \abs{u(x_j,y_j)} > 1$, we then estimate this term by $8 \pi/(\rho \beta^2)$, and bound equation~\eqref{eq:mink_conv_boundyoung_log} by $8 \pi/(\rho \beta^2) \left[ \int ( 1 + \mu^2 u^2 )^{-\rho} \total u \right]^2 \leq 8 \pi^3/(\rho \beta^2 \mu^2)$. In the region where $\mu \abs{ u(x_j,y_j) } \leq 1$, we use again Young's inequality in the form~\eqref{eq:young3} with the exponents~\eqref{eq:young_qpr_choice}, but now in one dimension and with $\beta^2$ replaced by $\rho \beta^2/2$. This gives
\begin{splitequation}
&\norm{ ( 1 + \mu^2 \blank^2 )^{-\rho} }_p^2 \norm{ \Theta(1 - \mu \abs{\blank}) (\mu \abs{\blank})^{- \rho \frac{\beta^2}{8 \pi}} }_q \\
&\quad= \mu^{-2} \left( \frac{\sqrt{\pi} \Gamma\left( \frac{p \rho - 1}{2} \right)}{\Gamma\left( \frac{p \rho}{2} \right)} \right)^\frac{2}{p} \left( \frac{16 \pi}{8 \pi - \rho \beta^2 q} \right)^\frac{1}{q} \eqend{,}
\end{splitequation}
which is finite with the choice we made for $p$, $q$ and $\rho$ as before. For $k \neq j$, we use again Hölder's inequality~\eqref{eq:hoelder} with $r = ( 4 \pi - \rho \beta^2 )/( 2 \rho \beta^2 )$ to obtain
\begin{splitequation}
&\iint \Big[ \mu \abs{u(x_j,y_j)} \Big]^{- \rho \frac{\beta^2}{4 \pi}} \abs{\ln\abs{\mu u(x_j,y_k)}} \frac{\total u(x_j)}{[ 1 + \mu^2 u(x_j)^2 ]^\rho} \frac{\total u(y_j)}{[ 1 + \mu^2 u(y_j)^2 ]^\rho} \\
&\quad\leq \left[ \iint \Big[ \mu \abs{u(x_j,y_j)} \Big]^{- r \rho \frac{\beta^2}{4 \pi}} \frac{\total u(x_j)}{[ 1 + \mu^2 u(x_j)^2 ]^\rho} \frac{\total u(y_j)}{[ 1 + \mu^2 u(y_j)^2 ]^\rho} \right]^\frac{1}{r} \\
&\qquad\times \left[ \iint \abs{\ln\abs{\mu u(x_j,y_k)}}^\frac{r}{r-1} \frac{\total u(x_j)}{[ 1 + \mu^2 u(x_j)^2 ]^\rho} \frac{\total u(y_j)}{[ 1 + \mu^2 u(y_j)^2 ]^\rho} \right]^\frac{r-1}{r} \eqend{.}
\end{splitequation}
With this choice of $r$, we have $r > 1$ and $r \rho \beta^2/(4\pi) < 1$ such that both integrals are convergent, and we can bound each of them by repeating the estimates used to bound~\eqref{eq:mink_conv_boundyoung_log}. Since there are $m^2 + m(m-1) < 2 m^2$ terms with two derivatives $H_\mu$ of the Hadamard parametrix at different points, it follows that also the contribution of the first term in equation~\eqref{eq:mink_conv_hu_hv_ibp} to the numerator~\eqref{eq:mink_conv_num} is bounded by a sum of the form~\eqref{eq:mink_conv_num_1_bound}.

We thus consider the remaining terms with two derivatives of the Hadamard parametrix at different points for $\mu = \nu = u$; the case $\mu = \nu = v$ is completely analogous. Using the definition of $H_\mu$~\eqref{eq:mink_renorm_hmu_def}, we compute
\begin{splitequation}
&H_u(z,x) H_u(z,y) = \sum_{a,b=\pm} \Theta(a(u+v)(z,x)) \Theta(b(u+v)(z,y)) \\
&\hspace{6em}\times \frac{\partial}{\partial u(z)} \frac{\ln[(u(z,x) - a \mathi \epsilon)^2] - \ln[(u(z,y) - b \mathi \epsilon)^2]}{2 u(x,y) + 2 (a-b) \mathi \epsilon} \\
&= \left[ \frac{\partial}{\partial u(z)} - \frac{\partial}{\partial v(z)} \right] \sum_{a,b=\pm} \Theta(a(u+v)(z,x)) \Theta(b(u+v)(z,y)) \\
&\qquad\times \frac{\ln[(u(z,x) - a \mathi \epsilon)^2] - \ln[(u(z,y) - b \mathi \epsilon)^2]}{2 u(x,y) + 2 (a-b) \mathi \epsilon} \\
&= \left[ \frac{\partial}{\partial u(z)} - \frac{\partial}{\partial v(z)} \right]^2 \sum_{a,b=\pm} \Theta(a(u+v)(z,x)) \Theta(b(u+v)(z,y)) \\
&\quad\times \frac{(u(z,x) - a \mathi \epsilon) \ln[(u(z,x) - a \mathi \epsilon)^2] - (u(z,y) - b \mathi \epsilon) \ln[(u(z,y) - b \mathi \epsilon)^2]}{2 u(x,y) + 2 (a-b) \mathi \epsilon} \eqend{.}
\end{splitequation}
We would like to take the limit $\epsilon \to 0$ of the last fraction in the distributional sense. Clearly, it is an absolutely integrable function for all $\epsilon$ including $\epsilon = 0$ except for a possible singularity at $u(x,y) = 0$. However, this singularity is integrable: in the limit $u(x,y) \to 0$ we have
\begin{splitequation}
&\frac{(u(z,x) - a \mathi \epsilon) \ln[(u(z,x) - a \mathi \epsilon)^2] - (u(z,y) - b \mathi \epsilon) \ln[(u(z,y) - b \mathi \epsilon)^2]}{2 u(x,y) + 2 (a-b) \mathi \epsilon} \\
&= - 1 - \frac{1}{2} \ln[(u(z,x) - a \mathi \epsilon)^2] + \bigo{ u(x,y) } \eqend{,}
\end{splitequation}
which is an absolutely integrable function of $u(z,x)$ for all $\epsilon$. The fraction is therefore bounded by an absolutely integrable function, and using dominated convergence we may thus even take the pointwise limit $\epsilon \to 0$, which results in
\begin{equation}
H_u(z,x) H_u(z,y) = \frac{\partial^2}{\partial u(z)^2} \left[ \frac{u(z,x) \ln[u(z,x)^2] - u(z,y) \ln[u(z,y)^2]}{2 u(x,y)} \right] \eqend{.}
\end{equation}
Using the mean value theorem for $f(u(x)) = u(z,x) \ln[ u(z,x)^2 ]$, we have
\begin{equation}
\frac{u(z,x) \ln[u(z,x)^2] - u(z,y) \ln[u(z,y)^2]}{u(x,y)} = f'(u(a)) = - 2 - \ln[ u(z,a)^2 ]
\end{equation}
for some point $a$ such that $\min(u(x),u(y)) \leq u(a) \leq \max(u(x),u(y))$, and thus
\begin{equation}
H_u(z,x) H_u(z,y) = - \frac{1}{2} \frac{\partial^2}{\partial u(z)^2} \ln\left[ \mu^2 u(z,a)^2 \right] \eqend{.}
\end{equation}
Introducing as before factors $[ 1 + \mu^2 u(z)^2 ]$, it follows that we can estimate
\begin{splitequation}
\abs{ \int H_u(z,x) H_u(z,y) f(z) \total^2 z } &\leq \frac{1}{2} \norm{ [ 1 + \mu^2 u(\blank)^2 ] [ 1 + \mu^2 v(\blank)^2 ] \partial_u^2 f(\blank) }_\infty \\
&\quad\times \int \abs{\ln \abs{u(z,a)}} \frac{\total u(z)}{1 + \mu^2 u(z)^2} \int \frac{\total v(z)}{1 + \mu^2 v(z)^2} \eqend{,}
\end{splitequation}
and the integrals are estimated using equation~\eqref{eq:mink_conv_num_2_int}. Finally, we use
\begin{equation}
\ln\left( 2 + \mu \abs{u(a)} \right) \leq \ln\left( 2 + \mu \abs{u(x)} \right) + \ln\left( 2 + \mu \abs{u(y)} \right) \eqend{,}
\end{equation}
absorb the logarithms in the test functions $g$ and repeat the derivation of the denominator estimates. It follows that also the contribution of the terms with two derivatives of the Hadamard parametrix at different points for $\mu = \nu = u$ and $\mu = \nu = v$ to the numerator~\eqref{eq:mink_conv_num} is bounded by a sum of the form~\eqref{eq:mink_conv_num_1_bound}, using again that we have $m^2 + m(m-1) < 2 m^2$ of this type.

For the first two terms of the stress tensor $T_{\mu\nu} = \op_{\mu\nu} - \frac{1}{2} \eta_{\mu\nu} \op_\rho{}^\rho + g \eta_{\mu\nu} ( V_\beta + V_{-\beta} )$ we can take over the above bounds. For the third term, we use the result~\eqref{eq:mink_renorm_expect_vertex} and thus have to bound
\begin{splitequation}
\label{eq:mink_conv_expect_vertex}
&\sum_{n=0}^\infty \frac{\mathi^n}{n!} \int\dotsi\int \sum_{\sigma_i = \pm 1} \omega^{0,0}\left( \mathcal{T}\left[ V_\beta(g f) \otimes \bigotimes_{j=1}^n V_{\sigma_j \beta}(x_j) \right] \right) \prod_{i=1}^n g(x_i) \total^2 x_i \\
&= \mathi \sum_{m=0}^\infty \frac{(-1)^m}{m! (m+1)!} \mu^{-(m+1) \frac{\beta^2}{2 \pi}} \int f(z) g(z) \int\dotsi\int \\
&\quad\times \left[ \frac{\prod_{1 \leq j < k \leq m} [ u(x_j,x_k) v(x_j,x_k) ]_- \prod_{1 \leq j < k \leq m+1} [ u(y_j,y_k) v(y_j,y_k) ]_-}{(-1)^m \prod_{j=1}^m \prod_{k=1}^{m+1} [ u(x_j,y_k) v(x_j,y_k) ]_-} \right]^\frac{\beta^2}{4 \pi} \\
&\quad\times \left[ \frac{1}{- [ u(y_{m+1},z) v(y_{m+1},z) ]_-} \prod_{j=1}^m \frac{[ u(x_j,z) v(x_j,z) ]_-}{[ u(y_j,z) v(y_j,z) ]_-} \right]^\frac{\beta^2}{4 \pi} \\
&\quad\times \exp\left[ - \frac{\beta^2}{2} \sum_{i=1}^m \sum_{j=1}^m [ W(x_i,x_j) - W(y_i,x_j) - W(x_i,y_j) + W(y_i,y_j) ] \right] \\
&\quad\times \exp\left[ - \beta^2 \sum_{i=1}^m [ W(x_i,z) - W(y_i,z) - W(x_i,y_{m+1}) + W(y_i,y_{m+1}) ] \right] \\
&\quad\times \exp\left[ - \frac{\beta^2}{2} [ W(y_{m+1},y_{m+1}) - 2 W(y_{m+1},z) + W(z,z) ] \right] \\
&\quad\times \total^2 y_{m+1} \prod_{i=1}^m g(x_i) g(y_i) \total^2 x_i \total^2 y_i \total^2 z \eqend{,}
\end{splitequation}
where we used that because of the neutrality condition only odd terms $n = 2m+1$ give a non-vanishing contribution. Of these, $m$ have a positive $\sigma_j$ and $m+1$ have a negative one, such that the sum over the $\sigma_j$ resulted in a factor of $\binom{2m+1}{m} = (2m+1)!/(m!(m+1)!)$, and as before we renamed the integration variables with a negative $\sigma_j$ to $y_j$.
The second assumption on $W$ shows that the exponentials can be bounded by $1$, and setting $x_{m+1} \equiv z$ the terms in brackets combine to
\begin{equation}
\left[ \frac{\prod_{1 \leq j < k \leq {m+1}} [ u(x_j,x_k) v(x_j,x_k) ]_- [ u(y_j,y_k) v(y_j,y_k) ]_-}{(-1)^{m+1} \prod_{j,k=1}^{m+1} [ u(x_j,y_k) v(x_j,y_k) ]_-} \right]^\frac{\beta^2}{4 \pi} \eqend{.}
\end{equation}
We can then use the same steps as in bounding the denominator~\eqref{eq:mink_conv_denom}, with the result that the series~\eqref{eq:mink_conv_expect_vertex} is bounded by
\begin{equation}
\norm{ f }_\infty \sum_{m=0}^\infty \frac{1}{m! (m+1)!} [ (m+1)! ]^{1+\frac{\beta^2}{4 \pi}} K^{m+1} \leq C \sum_{m=0}^\infty (m!)^{\frac{\beta^2}{4 \pi}-1} m^2 K^m \eqend{,}
\end{equation}
with the constant $C$ now also depending on $K$ and thus on $g$. The same bound is obtained analogously for the fourth term in the stress tensor involving $V_{-\beta}$, which switches $x_i$ with $y_i$ in equation~\eqref{eq:mink_conv_expect_vertex}.

Since we have shown that the denominator of the Gell-Mann--Low formula~\eqref{eq:thm_mink_conv_series} is a convergent series in $g$~\eqref{eq:mink_conv_denombound} starting with $1$, and is thus bounded by $\frac{1}{2}$ from below for sufficiently small $\norm{ g }_\infty$, the required bounds for the expectation values of $\op_{\mu\nu}(f)$ and $T_{\mu\nu}(f)$ follow. \hfill\squareforqed

\subsection{Proof of theorem~\ref{thm_mink_cons} (Conservation)}
\label{sec_mink_cons}

As in the Euclidean case, to show that the expectation value of the interacting stress tensor is conserved it is enough to show that the numerator of the Gell-Mann--Low formula vanishes when smeared with a test function of the form $\partial^\mu f$ with $f \istest$. Consider thus the numerator for $\op_{\mu\nu} - \frac{1}{2} \eta_{\mu\nu} \eta^{\rho\sigma} \op_{\rho\sigma}$, obtained as a linear combination of the numerator for $\op_{\mu\nu}$~\eqref{eq:mink_conv_num}, and smear it with $\partial^\mu f$. The result contains now three different types of terms: the ones with the renormalised $H_{\mu\nu}^\text{ren}(x,y)$~\eqref{eq:hmunu_ren}, the ones involving double sums of $H_\mu$, and the ones involving $W$. We start with the latter type, which contains double sums involving $W$, a coincidence limit of derivatives of $W$, and a mixed double sum involving $W$ and $H_\mu$. For the double sums involving $W$, we compute
\begin{splitequation}
&\int \partial^\mu f(z) \left[ \partial_{(\mu} W(z,x) \partial_{\nu)} W(z,y) - \frac{1}{2} \eta_{\mu\nu} \partial^\rho W(z,x) \partial_\rho W(z,y) \right] \total^2 z \\
&\quad= - \frac{1}{2} \int f(z) \left[ \partial^2 W(z,x) \partial_\nu W(z,y) + \partial_\nu W(z,x) \partial^2 W(z,y) \right] \total^2 z \eqend{,}
\end{splitequation}
where $x$ and $y$ do not need to be distinct. Since $W$ is a bisolution of the Klein--Gordon equation $\partial^2_x W(x,y) = \partial^2_y W(x,y) = 0$, these terms vanish. For the terms involving mixed derivatives acting on $W$, we use Synge's rule~\cite{poissonpoundvega2011}
\begin{equation}
\label{eq:synge_rule}
\partial_\mu^z f(z,z) = \left[ \partial_\mu^z f(z,z') + \partial_\mu^{z'} f(z,z') \right]_{z' = z} \eqend{,}
\end{equation}
and compute
\begin{splitequation}
&\int \partial^\mu f(z) \left[ \partial^z_\mu \partial^{z'}_\nu W(z,z') - \frac{1}{2} \eta_{\mu\nu} \partial^z_\rho \partial_{z'}^\rho W(z,z') \right]_{z' = z} \total^2 z \\
&\quad= - \int f(z) \left[ \partial_z^2 \partial^{z'}_\nu W(z,z') - \frac{1}{2} \partial^z_\rho \left( \partial^z_\nu - \partial^{z'}_\nu \right) \partial_{z'}^\rho W(z,z') \right]_{z' = z} \total^2 z \eqend{.}
\end{splitequation}
The first term again vanishes since $W$ is a bisolution, while the second one vanishes because $W(x,y) = W(y,x)$ is symmetric in $x$ and $y$, exchanging $z$ and $z'$ in one of the two parts. Lastly, for the mixed terms involving both $W$ and $H_\mu$, we obtain
\begin{splitequation}
&\int \partial^\mu f(z) \left[ H_{(\mu}(z,x) \partial_{\nu)} W(z,y) - \frac{1}{2} \eta_{\mu\nu} H^\rho(z,x) \partial_\rho W(z,y) \right] \total^2 z \\
&\quad= - \int f(z) \left[ \frac{1}{2} \partial^\mu H_\mu(z,x) \partial_\nu W(z,y) + \frac{1}{2} H_\nu(z,x) \partial^2 W(z,y) \right] \total^2 z \eqend{,}
\end{splitequation}
and since $W$ is a bisolution, the last term vanishes. However, for $\partial^\mu H_\mu$ we obtain instead using the definition of $H_\mu$~\eqref{eq:mink_renorm_hmu_def}
\begin{splitequation}
\partial^\mu H_\mu(z,x) = - 4 \pi \mathi \, \partial^2 H^\text{F}(z,x) = - 4 \pi \mathi \, \delta(z-x) \eqend{,}
\end{splitequation}
since the Hadamard parametrix $H^\text{F}$ is not a bisolution of the (massless) Klein--Gordon equation, but instead a fundamental solution~\eqref{eq:hf_fundamental_solution}, and hence
\begin{splitequation}
&\int \partial^\mu f(z) \left[ H_{(\mu}(z,x) \partial_{\nu)} W(z,y) - \frac{1}{2} \eta_{\mu\nu} H^\rho(z,x) \partial_\rho W(z,y) \right] \total^2 z \\
&\quad= 2 \pi \mathi f(x) \, \partial_\nu W(x,y) \eqend{.}
\end{splitequation}
Similarly, for the terms involving double sums containing $H_{(\mu}(z,x_i) H_{\nu)}(z,x_j)$ with $i \neq j$, we obtain
\begin{splitequation}
&\int \partial^\mu f(z) \left[ H_{(\mu}(z,x) H_{\nu)}(z,y) - \frac{1}{2} \eta_{\mu\nu} H^\rho(z,x) H_\rho(z,y) \right] \total^2 z \\
&\quad= 2 \pi \mathi \, \left[ f(x) \, H_\nu(x,y) + f(y) \, H_\nu(y,x) \right] = 2 \pi \mathi \, \left[ f(x) - f(y) \right] H_\nu(x,y) \eqend{,}
\end{splitequation}
and for the terms involving the renormalised $H_{\mu\nu}^\text{ren}(x,y)$ it follows that
\begin{splitequation}
&\int \partial^\mu f(z) \left[ H_{\mu\nu}^\text{ren}(z,x) - \frac{1}{2} \eta_{\mu\nu} H_\rho^\rho{}^\text{ren}(z,x) \right] \total^2 z \\
&\quad= 2 \pi \mathi \int \partial_\nu \partial^2 f(z) H^\text{F}(z,x) \total^2 z = 2 \pi \mathi \, \partial_\nu f(x) \eqend{,}
\end{splitequation}
where we inserted the explicit form of $H_{\mu\nu}^\text{ren}(x,y)$~\eqref{eq:hmunu_ren}, integrated by parts, and used that $H^\text{F}$ is a fundamental solution~\eqref{eq:hf_fundamental_solution} of the Klein--Gordon equation.

Analogous to the Euclidean case~\eqref{eq:euclid_cons_det_derivative}, we compute
\begin{splitequation}
&\partial_\nu^{x_\ell} \ln \left( \left[ \frac{\prod_{1 \leq j < k \leq m} [ u(x_j,x_k) v(x_j,x_k) ]_- [ u(y_j,y_k) v(y_j,y_k) ]_-}{(-1)^m \prod_{j,k=1}^m [ u(x_j,y_k) v(x_j,y_k) ]_-} \right]^\frac{\beta^2}{4 \pi} \right) \\
&\quad= \delta_\nu^u \frac{\beta^2}{4 \pi} \left[ \sum_{k\neq\ell} \frac{1}{[ u(x_\ell,x_k) ]_{-\sgn (u+v)}} - \sum_{k=1}^m \frac{1}{[ u(x_\ell,y_k) ]_{-\sgn (u+v)}} \right] \\
&\qquad+ \delta_\nu^v \frac{\beta^2}{4 \pi} \left[ \sum_{k\neq\ell} \frac{1}{[ v(x_\ell,x_k) ]_{-\sgn (u+v)}} - \sum_{k=1}^m \frac{1}{[ v(x_\ell,y_k) ]_{-\sgn (u+v)}} \right] \eqend{,}
\end{splitequation}
where we recall that
\begin{equations}[]
[ u(x_\ell,x_k) ]_{- \sgn (u+v)} &= \lim_{\epsilon \to 0} [ u(x_\ell,x_k) - \mathi \epsilon \sgn( (u+v)(x_\ell,x_k) ) ] \eqend{,} \\
[ v(x_\ell,x_k) ]_{- \sgn (u+v)} &= \lim_{\epsilon \to 0} [ v(x_\ell,x_k) - \mathi \epsilon \sgn( (u+v)(x_\ell,x_k) ) ]
\end{equations}
are the distributional boundary values obtained in the physical limit. We multiply by $f(x_\ell)$, sum over $\ell$ and rename summation indices to obtain
\begin{splitequation}
&\sum_{k=1}^m f(x_k) \partial_\nu^{x_k} \ln \left( \left[ \frac{\prod_{1 \leq j < k \leq m} [ u(x_j,x_k) v(x_j,x_k) ]_- [ u(y_j,y_k) v(y_j,y_k) ]_-}{(-1)^m \prod_{j,k=1}^m [ u(x_j,y_k) v(x_j,y_k) ]_-} \right]^\frac{\beta^2}{4 \pi} \right) \\
&\quad= \delta_\nu^u \frac{\beta^2}{4 \pi} \left[ \sum_{1 \leq j < k \leq m}^m \frac{f(x_j) - f(x_k)}{[ u(x_j,x_k) ]_{-\sgn (u+v)}} - \sum_{j,k=1}^m \frac{f(x_j)}{[ u(x_j,y_k) ]_{-\sgn (u+v)}} \right] \\
&\qquad+ \delta_\nu^v \frac{\beta^2}{4 \pi} \left[ \sum_{j=1}^m \sum_{k=j+1}^m \frac{f(x_j) - f(x_k)}{[ v(x_j,x_k) ]_{-\sgn (u+v)}} - \sum_{j,k=1}^m \frac{f(x_j)}{[ v(x_j,y_k) ]_{-\sgn (u+v)}} \right] \\
&\quad= \frac{\beta^2}{4 \pi} \left[ \sum_{1 \leq j < k \leq m} [ f(x_j) - f(x_k) ] H_\nu(x_j,x_k) - \sum_{j,k=1}^m f(x_j) H_\nu(x_j,y_k) \right] \eqend{,}
\end{splitequation}
where in the last equality we used equation~\eqref{eq:mink_renorm_hmu_def} in the limit $\epsilon \to 0$, which can be written as
\begin{equation}
H_u(x,y) = \frac{1}{[ u(x,y) ]_{-\sgn (u+v)}} \eqend{,} \quad H_v(x,y) = \frac{1}{[ v(x,y) ]_{-\sgn (u+v)}} \eqend{.}
\end{equation}
By the same procedure, we also obtain the analogous equation with $x$ and $y$ exchanged. Similarly, we compute
\begin{splitequation}
&\sum_{k=1}^m f(x_k) \partial_\nu^{x_k} \exp\left[ - \frac{\beta^2}{2} \sum_{i,j=1}^m [ W(x_i,x_j) - W(y_i,x_j) - W(x_i,y_j) + W(y_i,y_j) ] \right] \\
&\quad= - \beta^2 \exp\left[ - \frac{\beta^2}{2} \sum_{i,j=1}^m [ W(x_i,x_j) - W(y_i,x_j) - W(x_i,y_j) + W(y_i,y_j) ] \right] \\
&\qquad\times \sum_{j,k=1}^m f(x_k) \partial^{x_k}_\nu [ W(x_k,x_j) - W(x_k,y_j) ] \eqend{,}
\end{splitequation}
where the derivative acts on the first argument of $W$, as well as the analogous equation with $x$ and $y$ exchanged. It follows that the numerator of the Gell-Mann--Low formula for the stress tensor $T_{\mu\nu} = \op_{\mu\nu} - \frac{1}{2} \eta_{\mu\nu} \eta^{\rho\sigma} \op_{\rho\sigma} + g \eta_{\mu\nu} ( V_\beta + V_{-\beta} )$, smeared with $\partial^\mu f$, reduces to
\begin{splitequation}
\label{eq:mink_cons_tmunu_1}
&\sum_{n=0}^\infty \frac{1}{n!} \int\dotsi\int \sum_{\sigma_i = \pm 1} \omega^{0,0}\left( \mathcal{T}\left[ T_{\mu\nu}(\partial^\mu f) \otimes \bigotimes_{j=1}^n V_{\sigma_j \beta}(x_j) \right] \right) \prod_{i=1}^n g(x_i) \total^2 x_i \\
&= - \mathi \sum_{m=0}^\infty \frac{(-1)^m}{(m!)^2} \int\dotsi\int \left[ \frac{\prod_{1 \leq j < k \leq m} [ u(x_j,x_k) v(x_j,x_k) ]_- [ u(y_j,y_k) v(y_j,y_k) ]_-}{(-1)^m \prod_{j,k=1}^m [ u(x_j,y_k) v(x_j,y_k) ]_-} \right]^\frac{\beta^2}{4 \pi} \\
&\quad\times \Bigg[ \sum_{i=1}^m [ \partial_\nu f(x_i) + \partial_\nu f(y_i) ] - \frac{\beta^2}{8 \pi} \sum_{i=1}^m [ \partial_\nu f(x_i) + \partial_\nu f(y_i) ] \Bigg] \\
&\quad\times \mu^{- m \frac{\beta^2}{2 \pi}} \exp\left[ - \frac{\beta^2}{2} \sum_{i,j=1}^m [ W(x_i,x_j) - W(y_i,x_j) - W(x_i,y_j) + W(y_i,y_j) ] \right] \\
&\quad\times \prod_{i=1}^m g(x_i) g(y_i) \total^2 x_i \total^2 y_i \\
&+ \mathi \sum_{m=0}^\infty \frac{(-1)^m}{m! (m+1)!} \int\dotsi\int \left[ \frac{\prod_{1 \leq j < k \leq m+1} [ u(x_j,x_k) v(x_j,x_k) ]_- [ u(y_j,y_k) v(y_j,y_k) ]_-}{(-1)^{m+1} \prod_{j,k=1}^{m+1} [ u(x_j,y_k) v(x_j,y_k) ]_-} \right]^\frac{\beta^2}{4 \pi} \\
&\quad\times \mu^{-(m+1) \frac{\beta^2}{2 \pi}} \exp\left[ - \frac{\beta^2}{2} \sum_{i,j=1}^{m+1} [ W(x_i,x_j) - W(y_i,x_j) - W(x_i,y_j) + W(y_i,y_j) ] \right] \\
&\quad\times [ \partial_\nu f(x_{m+1}) + \partial_\nu f(y_{m+1}) ] \prod_{i=1}^{m+1} g(x_i) g(y_i) \total^2 x_i \total^2 y_i \eqend{,}
\end{splitequation}
where we integrated some derivatives by parts, used that by assumption $g$ is constant on the support of $f$, and for the last two terms in the stress tensor involving the vertex operators used the result~\eqref{eq:mink_conv_expect_vertex} with $z$ renamed to $x_{m+1}$ and the analogous result for $V_{-\beta}$ with $z$ renamed to $y_{m+1}$. Shifting in the last sum the summation index $m \to m-1$, and using that because of the symmetry of the integrand we can replace
\begin{equation}
\partial_\nu f(x_m) + \partial_\nu f(y_m) \to \frac{1}{m} \sum_{i=1}^m [ \partial_\nu f(x_i) + \partial_\nu f(y_i) ] \eqend{,}
\end{equation}
the last sum in equation~\eqref{eq:mink_cons_tmunu_1} coming from the vertex operators in the stress tensor cancels the first sum in brackets in the first sum in equation~\eqref{eq:mink_cons_tmunu_1}, completely analogous to the Euclidean case.

The remaining term is of the same form as the contribution of $V_{\pm \beta}$~\eqref{eq:mink_conv_expect_vertex}, and as in the Euclidean case it follows that a modified stress tensor is conserved in the quantum theory: we have
\begin{equation}
\sum_{n=0}^\infty \frac{1}{n!} \int\dotsi\int \sum_{\sigma_i = \pm 1} \omega^{0,0}\left( \mathcal{T}\left[ \hat{T}_{\mu\nu}(\partial^\mu f) \otimes \bigotimes_{j=1}^n V_{\sigma_j \beta}(x_j) \right] \right) \prod_{i=1}^n g(x_i) \total^2 x_i = 0
\end{equation}
with~\eqref{eq:thm_mink_cons_stresstensor}
\begin{equation}
\label{eq:mink_cons_modified_stress}
\hat{T}_{\mu\nu} \equiv T_{\mu\nu} - \frac{\beta^2}{8 \pi} g \eta_{\mu\nu} ( V_\beta + V_{-\beta} ) \eqend{.}
\end{equation}
One might ask if a further redefinition of time-ordered products could be used to get rid of the extra term in equation~\eqref{eq:mink_cons_modified_stress}, such that the classical stress tensor would also be conserved in the quantum theory. However, this is impossible since the term in question is proportional to $\eta_{\mu\nu}$, and modifying $\op_{\mu\nu}$ by any such term has no effect on the stress tensor. Moreover, redefinitions of time-ordered products only involving $V_\beta$ are not allowed by dimensional reasons in the finite regime $\beta^2 < 4 \pi$, so the modified stress tensor~\eqref{eq:mink_cons_modified_stress} is unique. \hfill\squareforqed

\begin{acknowledgements}
This work is supported by the Deutsche Forschungsgemeinschaft (DFG, German Research Foundation) --- project no. 396692871 within the Emmy Noether grant CA1850/1-1 and project no. 406116891 within the Research Training Group RTG 2522/1. We thank the anonymous referees for useful comments and suggestions.
\end{acknowledgements}

\bibliography{literature}

\end{document}